%% file: main.tex
\documentclass[letterpaper,twocolumn,10pt]{article}
\usepackage{usenix-2020-09}
\input{defs}

\begin{document}

\title{\sysname: \update{Scalable} Federated Learning with Private Weights}
\author{
\rm Yiyi Sun$^*$ \quad Mushtari Sadia$^*$ \quad Amrita Roy Chowdhury$^*$\\
$^*$University of Michigan, Ann Arbor
}
\maketitle

\begin{abstract}
The twin goals of privacy-preserving federated learning (FL) -- harnessing distributed data while protecting the parties who contribute it--demands three properties: model privacy, data privacy, and integrity. Achieving all three in a realistic deployment--with a single coordinator and millions of users--remains an open challenge that prior work has yet to fully resolve. We present \sysname{}, the first scalable MPC-based FL system to provide all three guarantees simultaneously. \sysname{} operates under an honest-majority assumption in the semi-honest model, with an extension to the malicious-with-abort setting. At its core, \sysname{} combines a multi-committee MPC architecture with packed Shamir secret sharing, enabling SIMD-style secure computation that scales far beyond what monolithic n-party protocols can support. We further introduce ML-aware optimizations and new cryptographic primitives that substantially reduce per-user overhead. Our evaluation shows \sysname~is practically deployable at unprecedented scale: it reduces per-user computation to just <350 ms (semi-honest) and <700 ms (malicious) for per iteration of training MNIST, achieving three orders-of-magnitude improvements over prior work while seamlessly supporting populations of over $10^{6}$ users. This demonstrates, for the first time, that MPC-based FL can meet the scale demands of real-world systems without compromising on privacy or integrity.

\end{abstract}

\input{1_intro}
\input{2_background}
\input{3_problem_setting}

\input{4_TechnicalOverview}
\input{5_protocol_details}
\input{6_sys_description}

\input{7_security_analysis}

\input{8_evaluation}
\input{9_discussion}
\input{10_related_work}
\input{11_conclusion}

\section*{Ethical Considerations}

This paper studies privacy preserving federated learning using MPC to jointly achieve data privacy, model privacy, and integrity in a single server setting. Federated learning involves sensitive user data and proprietary model parameters, so we consider the ethical implications of both our system design and its potential deployment.

\noindent\textbf{Stakeholders.}
Stakeholders include end users whose private data is used for training, service providers deploying federated learning systems, organizations relying on trained models, practitioners implementing such systems, and the broader research community.

\noindent\textbf{Principles.}
We consider the Menlo Report principles of Beneficence, Respect for Persons, Justice, and Respect for Law and Public Interest.

\noindent\textbf{Potential harms.}
A primary concern is misuse. Even if training is privacy preserving, the resulting model may be used in harmful or unethical applications. Second, the system’s guarantees depend on the stated threat model, including the honest majority assumption within committees and the malicious with abort setting; deploying the system outside these assumptions could lead to overestimated privacy or integrity guarantees. Third, while our protocol enforces well formed gradients, it does not ensure that the trained model is unbiased, fair, or safe for downstream use.

Another concern is dual use. Techniques that enable scalable privacy preserving training could be adopted in settings where users may not fully understand or consent to how their data contributes to model training, even if the protocol itself preserves confidentiality.

\noindent\textbf{Mitigations.}
We mitigate these risks through clear design and transparency. First, all user data remains secret shared and is never revealed to the server or other users. Second, model weights are protected from disclosure to users, preserving provider confidentiality. Third, integrity checks ensure that malformed or adversarial updates are detected within the defined threat model. 

We explicitly state all assumptions and limitations, including the threat model and deployment constraints. Our evaluation uses only standard public datasets and does not involve human subjects, private user data, live deployments, or vulnerability disclosure. We do not release sensitive data or operational system details.

\noindent\textbf{Decision.}
We believe it is ethical to proceed and publish because this work advances practical mechanisms for protecting both user data and model confidentiality in federated learning. The system reduces reliance on weaker or ad hoc privacy approaches, while making all assumptions explicit. Publication enables scrutiny, informed adoption, and responsible deployment.

\bibliographystyle{plain}
\bibliography{references}
\appendix
\input{12_appendix}

\end{document}

%% file: defs.tex
\usepackage{enumitem, amsmath, amssymb}
\usepackage{amsthm}
\usepackage{xcolor}
\usepackage{todonotes}
\usepackage{stmaryrd} 
\usepackage{graphicx}
\usepackage{float}
\usepackage{multirow}
\usepackage{booktabs}
\usepackage{subcaption}

\usepackage{tikz}
\usepackage{varwidth}
\usepackage[operators,sets]{cryptocode}
\usepackage[breakable]{tcolorbox}
\usepackage{cuted}
\usepackage{subcaption}
\usepackage{comment}
\usepackage{makecell}
\usepackage{xcolor}

\usepackage{pifont}
\newcommand{\cmark}{\ding{51}}
\newcommand{\xmark}{\ding{55}}
\usepackage{tikz}
\newcommand{\dotsize}{0.65ex}
\newcommand{\fullc}{%
\tikz[baseline=-0.6ex]\fill (0,0) circle (\dotsize);}
\newcommand{\emptyc}{%
\tikz[baseline=-0.6ex]\draw (0,0) circle (\dotsize);}
\newcommand{\halfc}{%
\tikz[baseline=-0.6ex]{
  \begin{scope}
    \clip (-\dotsize,-\dotsize) rectangle (0,\dotsize);
    \fill (0,0) circle (\dotsize);
  \end{scope}
  \draw (0,0) circle (\dotsize);
}}

\renewcommand{\ell}{b}   

\newcommand{\sysname}{\textsc{Nomai}}
\newcommand{\totalC}{N}
\newcommand{\Co}{\mathcal{C}}
\newcommand{\us}{u}
\newcommand{\size}{n}
\newcommand{\sizeprime}{n'}
\newcommand{\tile}{\mathcal{T}}
\newcommand{\pf}{l}

\newcommand{\thr}{t}
\newcommand{\num}{c}
\newcommand{\dropout}{D}
\newcommand{\Ser}{\mathcal{S}}
\newcommand{\calG}{\mathcal{G}}
\newcommand{\g}{\mathfrak{g}}
\newcommand{\gradDim}{B_{\mathrm{grad}}}
\newcommand{\R}{\mathbb{R}}
\newcommand{\scr}[1]{\mathsf{#1}}
\newcommand{\PSS}{\textsf{PSS}}
\newcommand{\vx}[1]{\overrightarrow{#1}}
\newcommand{\adv}{\mathcal{A}}
\newcommand{\Sim}{\mathsf{Sim}}
\newcommand{\numPSS}{s}
\newcommand{\msg}[2]{\mathsf{msg}_{#1\to #2}}
\newcommand{\numvec}{n_{\mathrm{vec}}}

\newcommand{\update}[1]{\textcolor{black}{#1}}

\newcommand{\onop}[1]{\overset{\raisebox{0.5ex}{\scriptsize \textbf{online}}}{#1}}

\newcommand{\Trunc}[1]{\operatorname{Trunc}_{#1}}
\newcommand{\PiShare}{\Pi_{\mathrm{Share}}}
\newcommand{\PiReveal}{\Pi_{\mathrm{Reveal}}}
\newcommand{\PssToAdd}{\Pi_{\mathrm{PSS}\to\mathrm{Add}}}
\newcommand{\Share}[2]{\PiShare^{(#2)}(#1)}
\newcommand{\Reveal}[1]{\PiReveal(#1)}
\newcommand{\msb}{\mathrm{msb}}

\newcommand{\Cleanup}[1]{\textsf{Cleanup}(#1)}
\newcommand{\DegRed}{\Pi^{\mathsf{Pack}}_{\mathsf{DegRed}}}
\newcommand{\TruncDegRed}{\Pi^{\mathsf{Pack}}_{\mathsf{Trunc\text{-}DegRed}}}
\newcommand{\DegRedRepack}{\Pi^{\mathsf{Pack}}_{\mathsf{DegRed\text{-}Repack}}}
\newcommand{\F}{\mathbb F}
\newcommand{\Fp}{\mathbb F_p}

\newcommand{\pss}[2]{\llbracket #1 \rrbracket_{#2}}
\newcommand{\pssf}[3]{\llbracket #1 \rrbracket_{#2}^{#3}}
\newcommand{\sss}[2]{[#1]_{#2}}
\newcommand{\sssf}[3]{[#1]^{#3}_{#2}}
\newcommand{\ash}[1]{\langle #1 \rangle}

\newcommand{\lag}[3]{\lambda^{(#1)}_{#2}(#3)}
\newcommand{\Cin}{C_{\mathrm{in}}}
\newcommand{\Cout}{C_{\mathrm{out}}}
\newcommand{\Winograd}[2]{F(#1\times #1,#2\times #2)}
\newcommand{\Ntile}{N_{\mathrm{tile}}}
\newcommand{\Had}{\odot}

\newcommand{\nrow}{n_1}
\newcommand{\ncol}{n_2}
\newcommand{\Ldiag}[2]{L_{#1}(#2)}
\newcommand{\Udiag}[2]{U_{#1}(#2)}
\newcommand{\UdiagPrime}[2]{U'_{#1}(#2)}
\newcommand{\rot}[1]{\rho^{#1}}

\theoremstyle{definition}
\newtheorem{definition}{Definition}
\newtheorem{theorem}{Theorem}
\newtheorem{lemma}[theorem]{Lemma}

\newcommand{\bcwarch}{\textbf{BCW-MLP}}
\newcommand{\creditarch}{\textbf{Credit-MLP}}
\newcommand{\mnistarch}{\textbf{MNIST-CNN}}
\newcommand{\fmnistarch}{\textbf{FMNIST-CNN}}
\newcommand{\svhnarch}{\textbf{SVHN-CNN}}
\newcommand{\cifararch}{\textbf{CIFAR-MLP-4K}}

 \newcommand{\squishlist}{
	\begin{list}{$\bullet$}
		{
			\setlength{\itemsep}{0pt}
			\setlength{\parsep}{0.5pt}
			\setlength{\topsep}{0.5pt}
			\setlength{\partopsep}{0pt}
			\setlength{\leftmargin}{1em}
			\setlength{\labelwidth}{1em}
			\setlength{\labelsep}{0.5em} } }
	
\newcommand{\squishend}{
\end{list}  }
\newenvironment{mybox2}[1]{%
    \begin{tcolorbox}[title={#1}, 
    colback=white, breakable]%
    }{
    \end{tcolorbox}
}

%% file: 1_intro.tex
\section{Introduction}
\label{sec:intro}
Federated learning (FL) is a distributed training paradigm in which a massive number of users collaboratively train a global model without ever centralizing their raw data. FL is widely deployed, from Google's next-word prediction across Android devices~\cite{google} to Apple's on-device intelligence across iPhones~\cite{apple}. Yet this is not without serious security vulnerabilities. FL involves two parties with competing incentives--users and a coordinating server. Users aim to keep their local data private; in many deployments, this data captures highly sensitive behavior that must never be exposed to the server or other participants. The server, typically a commercial organization, seeks to protect its proprietary model weights as intellectual property and to prevent leakage of sensitive training dynamics.
The distributed nature of FL further expands the attack surface. Malicious users can inject poisoned or malformed data to degrade global model performance~\cite{biggio2021poisoning, mei2015teaching, Fang2020LocalMP, bhagoji2019analyzing, kairouz2019advances}, or embed backdoors that are only triggered at test time~\cite{chen2017targeted, bagdasaryan2018backdoor, Xie2020DBADB}.

Beyond security, any practical solution must satisfy two  constraints that are crucial in real-world deployments. \update{First, it must operate in a \emph{single-server} setting, with one persistent entity responsible for orchestrating training, distributing model updates, and aggregating results. Such a setting may use lightweight user helpers, but it cannot rely on more non-colluding servers. This reflects real-world FL deployments; introducing multiple servers would alter the deployment model and incur substantial infrastructure overhead.} Second, the solution must scale to \emph{millions of users}. The power of FL comes precisely from broad participation and any approach that fails to scale to this regime undermines the core promise of FL.

\noindent \textbf{Related work.} Prior work falls into two camps. Scalable single-server approaches primarily focus on either data privacy and integrity~\cite{eiffel, ma2025armadillo, bell2023acorn, lycklama2023rofl, rathee2023elsa, falcon}, but assume model weights are revealed in the clear to users. The only prior work that attempts to address both data and weight privacy in the single-server setting is POSEIDON~\cite{poseidon}. However, it falls short in several ways. First, it does not provide data integrity. Second, it operates in a semi-honest threat model. Finally, it relies on homomorphic encryption-based computation, which significantly limits scalability, as we empirically verify in Sec.~\ref{sec:eval:scale}. The other line of work considers secure training protocols that achieve all three security guarantees but assume multiple non-colluding servers. In the context of FL, this effectively transforms the problem into outsourced computation, where users delegate their data to a set of servers. Additionally, scalability remains a major bottleneck in this setting as well, as verified in our experiments. Achieving all three remains an open challenge that prior work
has yet to fully resolve.  We close this gap.\\
\noindent\textbf{Our Contributions.} We present \sysname{}, the first scalable FL system to provide all three guarantees in the single-\update{coordinating server, user-assisted} setting. \sysname{} operates under an honest-majority assumption in the semi-honest model, with an extension to the malicious-with-abort setting. \sysname{} achieves this through a suite of technical innovations:
\squishlist
\item \textbf{New Architecture.} \update{\sysname{} retains a single coordinating server while electing small, disjoint committees of participating users to serve as MPC helpers. Each user's data is assigned to one committee and secret-shared among its members, which jointly compute shares of the gradient aggregate for the assigned data shard. Committees operate in parallel, enabling scalability. The per-helper workload is designed for commodity user hardware, avoiding dedicated server-grade MPC infrastructure. Within each committee, \emph{packed Shamir secret sharing} (\PSS) enables SIMD-style parallelism and further reduces the computational cost.}

\item \textbf{New Cryptographic Primitives.} To further reduce per-user computation, we introduce a suite of new cryptographic primitives, which we believe to be of independent interest to the secure ML community.
\\\textit{(i) Efficient Linear Layers.} Linear layers dominate ML training cost. We introduce targeted optimizations that significantly reduce their overhead in the MPC setting. For convolutional layers, we apply an ML-inspired transformation that reduces the number of multiplications. We further design novel packing strategies for both convolutional and fully connected layers under \PSS{}, yielding additional savings. To our knowledge, this is the first work to bring these optimizations to \PSS{} for ML training.\footnote{While packing has been explored in homomorphic encryption, it differs from \PSS: the slot count there is much larger, leading to different tradeoffs. In \PSS, the packing factor is smaller, as it directly affects the polynomial degree and thus the cost of subsequent operations.}
\\\textit{(ii) Efficient Secure Exponentiation.} We design a new protocol that reduces the rounds required for secure exponentiation, accelerating polynomial evaluation. This is especially important for non-linear layers--typically approximated by high-degree polynomials.
\\\textit{(iii) Efficient Degree Reduction.} We propose a degree-reduction protocol that supports repacking of shares across layouts as data moves between layers, ensuring compatibility with our packing strategies.
\\\textit{(iv) Efficient Truncation.} We propose a new truncation protocol that is compatible with \PSS.
\item \textbf{End-to-End System.} We provide a complete system and evaluate it across 6 datasets. \sysname{} is practically deployable at unprecedented scale: it reduces per-user computation to <\update{350} ms (semi-honest) and <\update{700} ms (malicious) per MNIST training iteration, achieving $10^{\update{3}}$ orders-of-magnitude improvements over prior work~\cite{poseidon} while supporting populations of over $10^{\update{6}}$ users. This demonstrates, for the first time, that FL can meet real-world scale without compromising privacy or integrity.
\squishend

%% file: 2_background.tex
\vspace{-2mm}\section{Background}
\label{sec:background}
\subsection{Cryptographic Primitives}
\label{sec:background-pss}
We use $\vx{x}$ to denote vectors (e.g., $\vx{x}=(x_1,\ldots,x_l)$) and $\scr{f}$ to denote polynomials (e.g., $\scr{f}(X)\in\F_p[X]$).

\noindent\textbf{Shamir Secret Sharing.}
Shamir secret sharing is a threshold scheme for distributing a secret among multiple parties with privacy against coalitions. Formally, it is an $(\size,t)$-threshold scheme over $\F_p$~\cite{shamir1979share}: a secret $x\in\F_p$ is shared among $\size$ parties using a random polynomial $\scr{f}(X)$ of degree at most $t$ with $\scr{f}(0)=x$, and party $P_i$ receives $\scr{f}(\alpha_i)$ at public point $\alpha_i$. We denote this by $\sss{x}{t}$ or $\sssf{x}{t}{\scr{f}}$. Any $t+1$ shares reconstruct $x=\scr{f}(0)$ via interpolation, while $t$ shares reveal nothing. The scheme is linear: additions and multiplication by public constants are local, while multiplying two degree-$t$ sharings yields a degree-$2t$ sharing, requiring degree reduction.

\noindent\textbf{Packed Secret Sharing.}
Packed secret sharing encodes multiple secrets into a single sharing to enable SIMD-style computation~\cite{franklin1992communication}. Given $\vx{x}=(x_1,\dots,x_l)\in\F_p^l$, it is represented by a polynomial $\scr{f}(X)$ of degree $d=t+l-1$. Let $\beta_1,\dots,\beta_l$ be public slot points and $\alpha_1,\dots,\alpha_\size$ evaluation points. A packed sharing $\pss{\vx{x}}{d}$ (or $\pssf{\vx{x}}{d}{\scr{f}}$) satisfies $\scr{f}(\beta_j)=x_j$ for $j\in[l]$, with randomness in the remaining coefficients, and party $P_i$ receives $\scr{f}(\alpha_i)$ as its share. Any $d+1$ shares reconstruct $\vx{x}$, while $t$ shares reveal nothing. We denote sharing and reconstruction by $\PiShare^{(d)}(\vx{x})$ and $\PiReveal(\pss{\vx{x}}{d})$. Operations are slot-wise: additions and multiplication by public constants are local. Multiplying $\pssf{\vx{x}}{d}{\scr{f}}$ and $\pssf{\vx{y}}{d}{\scr{g}}$ yields $\pssf{\vx{x}\odot\vx{y}}{2d}{\scr{h}}$ with $\scr{h}(X)=\scr{f}(X)\scr{g}(X)$ and $(\vx{x}\odot\vx{y})_j=x_jy_j$. The degree increases to $2d$, so degree reduction is required; typically $\size\ge 2d+1$.

\noindent\textbf{Lagrange Interpolation.}
Lagrange interpolation expresses reconstruction as a linear combination of shares. For any $S\subseteq[\size]$ with $|S|=d+1$,
\begin{center}
$
\begin{aligned}
\scr{f}(\beta_j)
&= \sum_{i\in S} \lag{d}{\alpha_i}{\beta_j}\scr{f}(\alpha_i),
\qquad
\lag{d}{\alpha_i}{\beta_j}
= \prod_{\substack{m\in S\\ m\neq i}}
\frac{\beta_j-\alpha_m}{\alpha_i-\alpha_m}.
\end{aligned}
$
\end{center}

\noindent\textbf{Additive Sharing.}
An additive sharing represents a value as a sum of local contributions across parties. For each coordinate $j$, the shares $\{u_{i,j}\}_{i\in[\size]}$ satisfy
\vspace{-2mm}\begin{center}
$
\begin{aligned}\vspace{-4mm}
x_j &= \sum_{i=1}^\size u_{i,j}, \qquad
\{\ash{x_j}\}_{j\in[l]} \leftarrow \PssToAdd(\pssf{\vx{x}}{d}{\scr{f}})
\end{aligned}
$
\end{center}
\vspace{-2mm}
\subsection{Machine Learning Primitives}
\label{sec:background-ml}

Neural network training alternates between a forward pass and a backward pass. In the forward pass, the network computes $\hat{y}=f_\theta(x)$. A convolutional layer produces $y_j[u,v]=\sum_{i,a,b} g_{j,i}[a,b]\;x_i[u+a,v+b]$, while a fully connected layer computes $y=Wx$, followed by an activation $z=\sigma(y)$ and loss $\mathcal{L}=\ell(\hat{y},y)$. In the backward pass, gradients are computed as $\frac{\partial \mathcal{L}}{\partial x_i}=\sum_j \frac{\partial \mathcal{L}}{\partial y_j} g_{j,i}$, $\frac{\partial \mathcal{L}}{\partial g_{j,i}}=x_i \frac{\partial \mathcal{L}}{\partial y_j}$, $\frac{\partial \mathcal{L}}{\partial x}=W^T \frac{\partial \mathcal{L}}{\partial y}$, and $\frac{\partial \mathcal{L}}{\partial W}=\frac{\partial \mathcal{L}}{\partial y} x^T$. Parameters are then updated, completing one iteration.

\noindent \textbf{Winograd Convolution}: This technique~\cite{alam2022winograd} \textit{reduces} the number of multiplications by operating on small overlapping \emph{tiles} of the input. For a Winograd instance $\Winograd{m}{r}$, each computation produces an $m\times m$ output tile from an input tile of size $\tile\times \tile$, where $\tile=m+r-1$. The input is partitioned into tiles $d^{(p,q)}\in\R^{\tile\times \tile}$, each producing an output tile $y^{(p,q)}\in\R^{m\times m}$, with $\Ntile=\left\lceil \frac{H_{\mathrm{out}}}{m} \right\rceil \left\lceil \frac{W_{\mathrm{out}}}{m} \right\rceil$ tiles per channel. Computation follows a transform--multiply--inverse pipeline. Each input tile and kernel are first mapped to a transform domain, multiplied element-wise, and then mapped back:
$
U = G g G^T, \qquad
V = B^T d B, \qquad
Y = A^T (U \odot V) A
$ where $g\in\R^{r\times r}$ is the kernel, $d\in\R^{\tile\times \tile}$ is the input tile, $A,B,G$ are fixed public matrices, and $\odot$ denotes the Hadamard product. For output channel $j$, the transformed products are accumulated across input channels $M_j^{(p,q)}
= \sum_{i=1}^{\Cin}
\bigl(G g_{j,i} G^T\bigr)
\odot
\bigl(B^T d_i^{(p,q)} B\bigr),
y_j^{(p,q)}
= A^T M_j^{(p,q)} A.
$

%% file: 3_problem_setting.tex
\section{Problem Setting}
\label{sec:problem-setting}

\update{We consider a standard cross-device FL setting with $\totalC$ users and a single coordinating server. Each user $\us_i$ holds a private data point $x_i$\footnote{Can be generalized to local dataset; single datapoint is used for simplicity}. The server is untrusted and aims to train a global model $\mathcal{M}$ over the joint dataset $\mathbf{X}=\{x_i\}_{i\in[\totalC]}$ without directly accessing user data. Training proceeds over multiple rounds. In each round, the server maintains the current parameters of $\mathcal{M}$ and obtains the aggregate gradient $\calG=\sum_{i\in[\totalC]}\g_i$, where $\g_i$ is the gradient computed from user $\us_i$'s local data. The server then uses $\calG$ to update the model for the next round.}

\subsection{Security Goals}
\label{sec:overview:goals}

\squishlist
\item \textbf{Data Privacy (User's Goal).}  For every honest user, no other user should learn any information about its raw private data $x_i$. The server learns only the final aggregate $\calG$ and nothing beyond what is revealed by this aggregate.

\item \textbf{Weight Privacy (Server's Goal).} The model parameters held by the server should remain hidden from the users.

\item \textbf{Data Integrity (Server's Goal)\footnote{Note that this goal is valid only in the malicious threat model}}. The server seeks to ensure that individual user gradients are well-formed. To this end, it enforces a validation predicate $\textsf{Valid}(\cdot)$ that defines a syntax for the gradients. A gradient $\g$ is considered valid if and only if $\textsf{Valid}(\g) = 1$. For example, per-gradient checks such as norm bounds are natural instances of $\textsf{Valid}(\cdot).$
\squishend
Reflecting these goals, we consider a corrupted server and assume that an $\eta$ fraction of users may be adversarial, under both semi-honest and malicious-with-abort threat models~\cite{evans2018pragmatic}.

\noindent\textbf{Availability assumption.} In the main protocol and evaluation, we assume all users selected for a training iteration remain online until completion. The protocol tolerates up to $\dropout$ dropouts by overprovisioning each committee with additional members, as described and evaluated in App.~\ref{app:dropout}.

\subsection{Design Goals}


\squishlist
\item \update{\textbf{Single-coordinating server deployment.} \sysname{} should preserve the standard FL deployment model, in which a single server coordinates and orchestrates training. Lightweight MPC helpers are permitted, but the system must not require additional non-colluding servers.}

\item \textbf{No user pre-processing.} A common way to speed up MPC is a data-independent pre-processing phase that generates randomness (e.g., Beaver triples), via multiple rounds of interaction. In cross-device FL, however, users are intermittently available. Since requiring such pre-processing introduces significant coordination challenges, in \sysname{}, users should participate only in the online phase. We still allow independent \textit{local pre-computation} by users.

\item \textbf{Low computation cost for users.} Users are resource-constrained devices (e.g., edge devices). Therefore, \sysname~aims to ensure all user-side operations are efficient.
\squishend
\begin{table}[t]
\centering
\caption{\small{Comparison of \sysname{} with prior works.}}
\label{tab:comparison}

\setlength{\tabcolsep}{3pt}
\renewcommand{\arraystretch}{0.95}

\scriptsize
\resizebox{\columnwidth}{!}{{\renewcommand{\arraystretch}{1.2}
\begin{tabular}{lccccc}
\toprule
\textbf{System} 
& \textbf{Data Privacy}
& \textbf{Weight Privacy}
& \textbf{Integrity} 
& \textbf{\update{Single Coordinator}}
& \textbf{Threat Model} \\
\midrule

EIFFeL\cite{eiffel}
& \cmark & \xmark & \fullc & \cmark & Malicious \\

\parbox[t]{0.38\columnwidth}{%
Armadillo\cite{ma2025armadillo}, Flamingo\cite{ma2023flamingo},\\
ACORN\cite{bell2023acorn}}
& \cmark & \xmark & \emptyc & \cmark & Malicious \\

RoFL\cite{lycklama2023rofl}
& \cmark & \xmark & \emptyc & \cmark & Semi-honest \\

POSEIDON\cite{poseidon}
& \cmark & \cmark & \emptyc & \cmark & Semi-honest \\

ELSA\cite{rathee2023elsa}
& \cmark & \xmark & \emptyc & \xmark & Malicious \\

\midrule

\parbox[t]{0.38\columnwidth}{%
SecureML\cite{mohassel2017secureml}, HMMPC\cite{liu2024scalable}}
& \cmark & \cmark & \fullc & \xmark & Semi-honest \\

SecureNN\cite{wagh2019securenn} 
& \cmark & \cmark & \halfc & \xmark & Malicious \\

\parbox[t]{0.38\columnwidth}{%
ABY3\cite{mohassel2018aby3}, BLAZE\cite{patra2020blaze}, FLASH\cite{byali2020flash}, \\
Trident\cite{chaudhari2019trident}, Falcon\cite{falcon}}
& \cmark & \cmark & \fullc & \xmark & Malicious \\

\midrule
\textbf{\sysname{}}
& \cmark & \cmark & \fullc & \cmark & Malicious \\

\bottomrule
\end{tabular}}
}

\begin{minipage}{\columnwidth}
\scriptsize
\textbf{Top:} FL frameworks.
\textbf{Bottom:} MPC-based general PPML training frameworks.

\textbf{Integrity.}
\emptyc\ norm-bound based checks only;
\fullc\ general-purpose integrity;
\halfc\ malicious privacy with only semi-honest correctness (SecureNN~\cite{wagh2019securenn}).
\end{minipage}
\vspace{-4mm}
\end{table}

%% file: 4_TechnicalOverview.tex
\vspace{-1mm}
\section{Technical Overview} \label{sec:overview}

We present the key technical ideas underlying \sysname{}, organized around the design goals above.
\\\noindent\textit{Single \update{Coordinator} Deployment.}
Ensuring gradient privacy requires performing gradient computation under MPC, i.e., each training round is a secure computation. Prior work~\cite{mohassel2017secureml,bonawitz2017practical} typically assumes multiple non-colluding servers, which conflicts with our single \update{coordinating} server setting. A naive alternative is an $(\totalC+1)$-party protocol among all users and the server, but this is prohibitively expensive at scale.

\sysname{} overcomes this via two ideas. First, instead of a single $\totalC$-party protocol, users are partitioned into $\num$ disjoint committees of size $\size \ll \totalC$. Committees are constructed to bound corruption, using an efficient procedure with $\Omega(\num)$ work per user. Each user’s data is assigned to one committee, yielding disjoint data shards of size $\totalC/\num$. Within each committee, data is secret-shared and an MPC protocol is run independently to compute a local gradient aggregate. This reduces the cost from a $\totalC$-party to a $\size$-party computation, with $\num$ committees running in parallel. Second, we use \emph{packed Shamir secret sharing} to encode $\pf$ values per share, enabling SIMD-style computation and a further $\pf$-fold speedup within each committee. Together, these yield an overall speedup of roughly $\frac{\totalC}{\num \cdot \pf}$. Each committee produces a partial aggregate, which are then securely combined into the final global aggregate revealed only to the server (Fig.~\ref{fig:system}).
\\\noindent\textit{Eliminating pre-processing.}
Secure ML training typically relies on pre-processing to generate data-independent randomness: correlated randomness (e.g., Beaver triples) for multiplications and random masks for non-linear operations such as truncation. \sysname{} eliminates both through careful protocol design. For multiplications, appropriate parameters allow Shamir secret sharing to perform local share-wise multiplication followed by degree reduction, requiring no correlated randomness. For non-linear operations such as truncation, which typically require bit decomposition and additional randomness, \sysname{} avoids this overhead by operating over Mersenne primes ($2^\ell-1$) for efficient modular arithmetic and bit operations, and leveraging users to supply randomness efficiently under \PSS.

\noindent\textit{Computational Efficiency.}
While Shamir sharing supports local multiplication, each multiplication doubles the polynomial degree, requiring degree reduction to continue computation. Naive interpolation costs $O(\size^2)$, while FFT-based methods achieve $O(\size \log \size)$ but require FFT-friendly primes. Since \sysname{} uses Mersenne primes for efficient truncation, these optimizations are unavailable, making degree reduction a central bottleneck. We address this with techniques that reduce both the number and cost of these operations.
\\\noindent\textbf{Efficient Linear Layers.} Our first set of techniques targets the dominant source of multiplications in training--linear layers which involve different types of matrix operations. \sysname{} reduces this cost through two key ideas. 

First, we adopt Winograd convolution for convolutional layers, a transformation that reduces the number of multiplications by converting convolution into a sequence of small linear transforms and element-wise products. To the best of our knowledge, this is the \textit{first} use of Winograd-style techniques in secure ML training. 

However, a direct application of Winograd is not well-suited to $\PSS$. To address this, our second key idea is to design novel packing strategies that retain Winograd’s computational advantages while remaining compatible with packed representations. By carefully aligning data layout with computation, \sysname{} maximizes parallelism while minimizing non-scalar multiplications and hence the number of expensive degree--reduction steps.  Similarly, we design new packing strategies for fully connected layers as well.
To the best of our knowledge, we are the \textit{first} ones to design packing strategies tailored to Shamir secret sharing for secure ML training.
\\\noindent\textbf{Efficient Secure Exponentiation.} Nonlinearities in neural networks (e.g., ReLU, sigmoid, and softmax) could be implemented via polynomial approximations in secure computation for efficiency. Achieving high accuracy, however, often requires high-degree polynomials, which substantially increases the number of multiplications--and consequently, the number of costly degree--reduction protocols.
In particular, the key bottleneck arises from exponentiation: computing $x^k$ naively requires $O(\log k)$ sequential multiplications (and hence degree-reduction steps). We design a more efficient protocol that reduces the sequential depth to $\log q$ rounds of degree reduction, where $q$ is a parameter determined by the corruption threshold. This yields gains whenever $\log q < \log k$, which holds for realistic corruption levels and high-degree polynomials.
At a high level, our approach converts a \PSS{} sharing of $x$ into $q$ additive shares and applies the multinomial theorem to express $x^k$ as a sum of terms that can each be computed via \emph{local} multiplications and then combined with $\log q$ degree--reduction rounds.
\\\noindent\textbf{Efficient Degree Reduction Protocol.}
A central bottleneck in Shamir-based MPC is degree reduction: each secret--secret multiplication produces a degree-$2d$ sharing that must be reduced to degree $d$ before further computation. In $\PSS$, this cost is amplified by frequent reshaping and repacking across layers. Existing approaches either rely on FFT-based interpolation, incompatible with Mersenne primes, or incur overhead from explicit unpacking and repacking. \sysname{} addresses this with an efficient degree-reduction protocol. 
\begin{figure}[t]
\centering
\small
\setlength{\tabcolsep}{2pt}
\renewcommand{\arraystretch}{0.9}

\begin{tabular}{@{}m{0.40\columnwidth} m{0.58\columnwidth}@{}}

\includegraphics[width=0.75\linewidth]{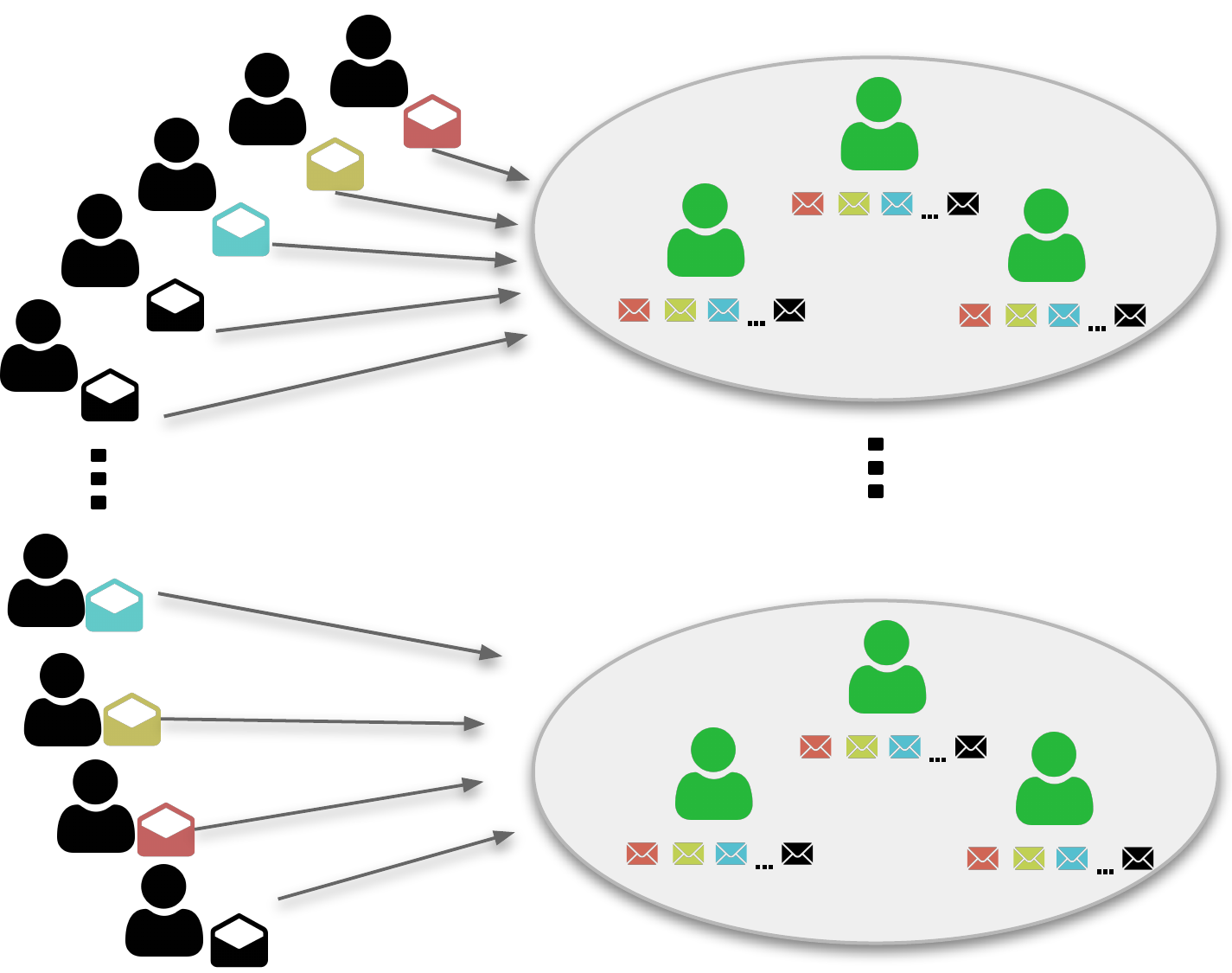} &
\vspace{0pt}\textbf{Data sharing}: $\totalC$ users are organized into $\num$ committees, each of size $\size$. A user’s data is secret-shared among the members of its assigned committee. \\[-3pt]

\vspace{0pt}\textbf{Training}: each committee \textit{independently} executes an MPC protocol to train and produce secret shares of the gradient aggregate for its members. &
\includegraphics[width=0.75\linewidth]{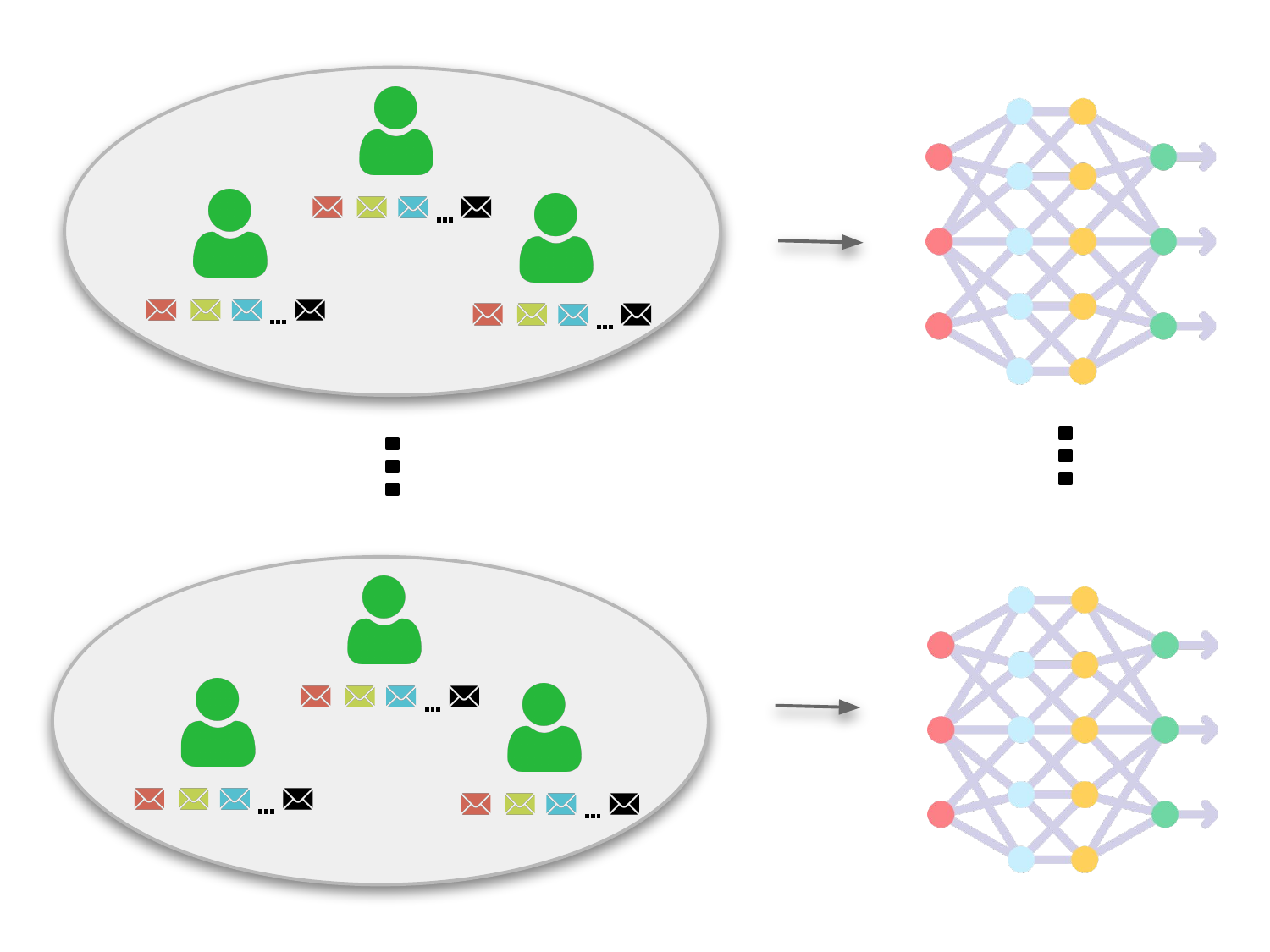} \\[-3pt]

\includegraphics[width=0.75\linewidth]{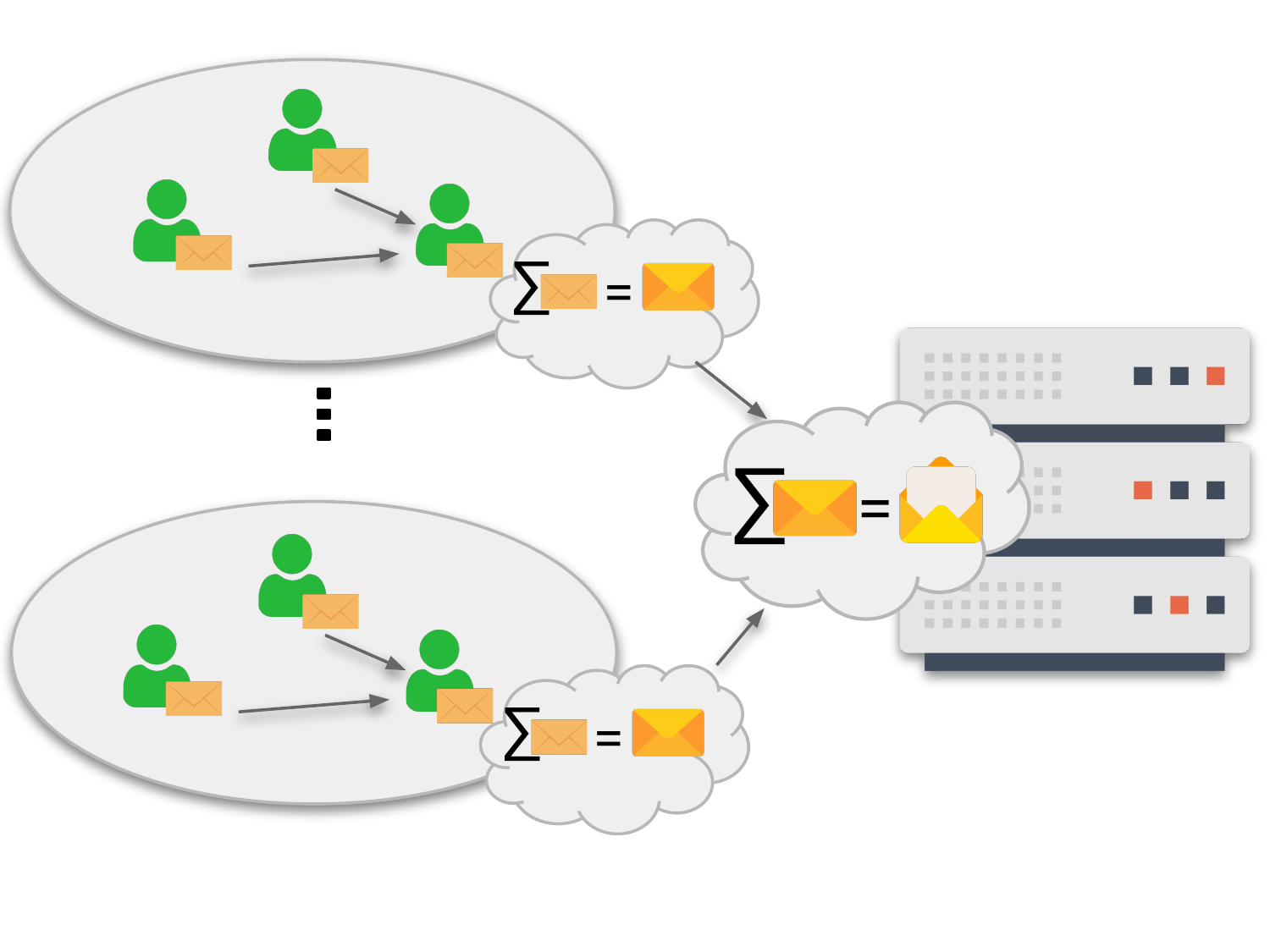} &
\vspace{0pt}\textbf{Aggregation}: updates are \textit{locally} aggregated within each committee, then combined securely \textit{across} committees at the server, with only the final result revealed.
\end{tabular}\vspace{-2mm}

\normalsize
\caption{High level overview of \sysname{}.}
\label{fig:system}
\vspace{-4mm}
\end{figure}
\\\noindent\textbf{Efficient Malicious Security.} Standard approaches to malicious security under Shamir secret sharing, such as Reed-Solomon decoding and verifiable secret sharing, incur significant computational overhead. While straightforward to integrate into our protocol, their cost becomes 
prohibitive at the scale of practical FL. \sysname{} achieves a significantly more 
efficient design via two new ideas that exploit unique properties of the FL setting. First, we make a crucial observation about the FL threat model: in addition to being malicious in the traditional sense of cryptographic protocols, malicious users in FL can corrupt the final aggregate by directly submitting malformed or poisoned data. Hence, as discussed in Sec.~\ref{sec:overview:goals}, ensuring \emph{data integrity}--that all gradients 
satisfy a prescribed validity condition--is a central server goal. This means that 
 simply performing the gradient computation followed by a secure aggregation with malicious security is \emph{not} sufficient; explicit well-formedness checks on each gradient are required regardless. \sysname{} turns this observation into an 
efficiency gain: since well-formedness must be checked anyway, we execute the expensive gradient computation using a cheap semi-honest protocol and enforce malicious security 
\textit{only} at the validity check. This dramatically reduces cost, as the heaviest part of the computation--gradient computation--runs in the cheaper semi-honest model. Second, we design a novel maliciously secure aggregation protocol that combines 
\emph{only} well-formed gradients into the final aggregate, ensuring that any malicious 
behavior is detected unless the adversary restricts itself to manipulations that pass 
the validity check.  Together, 
these two ideas deliver strong practical security at a fraction of the cost of a fully 
malicious protocol.

%% file: 5_protocol_details.tex
\vspace{-2mm}\section{Protocol Details}
\label{sec:protocol-details}
In this section, we present the core building blocks underlying \sysname{}. We introduce them in isolation, as they are of independent interest for secure training with Shamir secret sharing, and compose them into a full end-to-end FL system in the next section. We begin in the semi-honest setting; the extension to the malicious setting is deferred to Sec.~\ref{sec:mal}. Due to space constraints, we present only the constructions here and defer ideal functionalities and security proofs to App.~\ref{app:security}.

\begin{figure*}[t]
    \centering
    \includegraphics[width=0.8\linewidth]{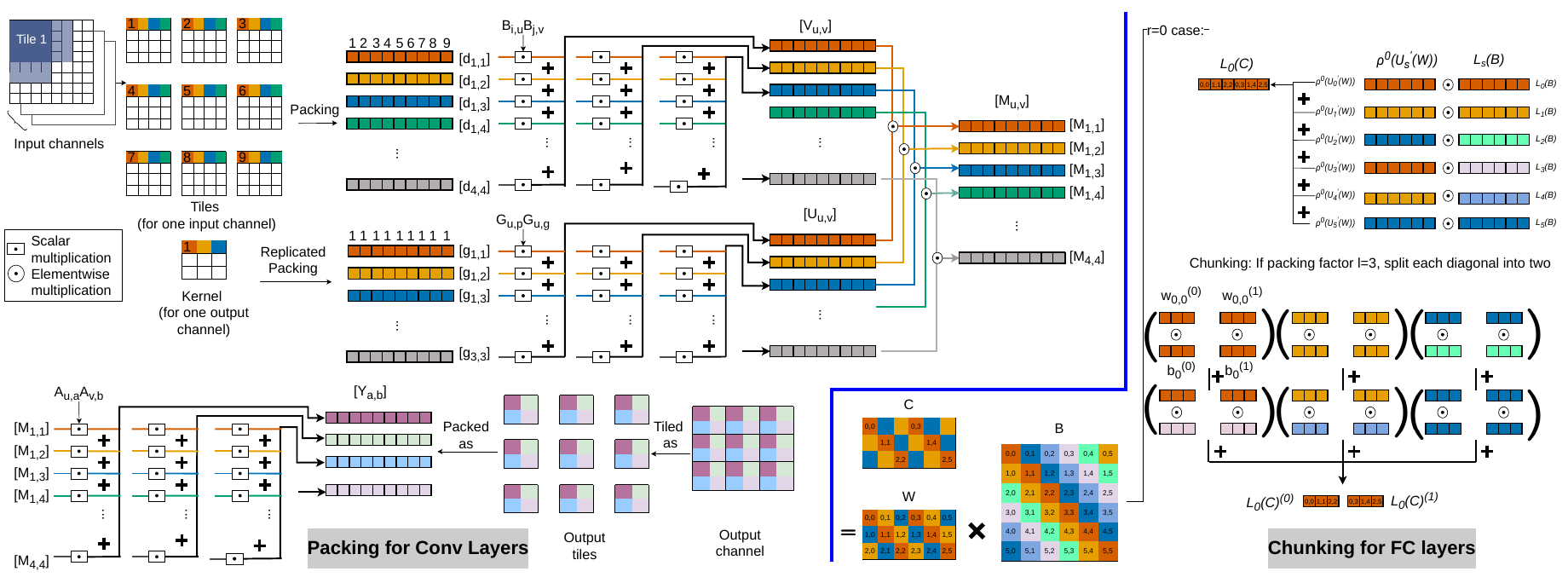}
   \vspace{-2mm} \caption{\footnotesize Left: Example flow of packed Winograd $\Winograd{2}{3}$ on a $6\times 6$ input with one input and one output channel. Right: Chunking example, showing that the chunked representation preserves the original diagonal structure.}
    \label{fig:placeholder}
    \vspace{-3mm}
\end{figure*}
\vspace{-2mm}
\subsection{Efficient Linear Layers}\label{sec:linear}
Linear layers dominate the computational cost in ML training, as they involve a large number of multiplications and, consequently, frequent degree-reduction steps. Hence, we begin by reducing the number of multiplications while remaining compatible with $\PSS$.
\subsubsection{Convolutional Layers}

For convolutional layers, we reduce this cost using two ideas. First, we adopt Winograd convolution, which expresses convolution as linear transforms and entry-wise products, reducing the cost of an $m\times m$ output block with an $r\times r$ kernel from $m^2 r^2$ multiplications to $\tile^2$, where $\tile = m + r - 1$. Concretely, it rewrites the computation as
\begin{center}
    $U = G g G^T, \qquad
    V = B^T d B, \qquad
    Y = A^T (U \odot V) A$
\end{center}
where $d\in\Fp^{\tile\times\tile}$ is an input tile, $g\in\Fp^{r\times r}$ is a kernel block, $Y\in\Fp^{m\times m}$ is the output tile, and $A,B,G$ are public transform matrices. Thus, the only non-linear step is the transformed-domain entry-wise product $U\Had V$. 

Second, we amplify these savings by tailoring the packing to the structure of Winograd. A natural approach is to pack different coordinates of the same tile into one sharing. However, this is inefficient under \PSS, since local multiplication by a public constant requires the \emph{same} constant in all slots, while Winograd applies different coefficients from $A$, $B$, or $G$ to different coordinates. Handling slot-dependent constants would require secret-share the public coefficients and perform a full secret--secret multiplication, introducing additional degree-reduction round and entirely wasting the benefit of those coefficients being public. Our idea therefore is to pack the \emph{same coordinate across many tiles} into one sharing. Each transform coefficient is then uniform across all slots, so the Winograd transforms become local operations. To the best of our knowledge, this is the \textit{first} use of Winograd-style techniques in secure ML training.

For clarity, we first describe the protocol for a single input channel and output channel. The full multi-channel version with accumulation over input channels, is given in App.~\ref{app:linear}. 
\\\noindent\textbf{Packing.} Since one \PSS{} can pack up to $\pf$ secrets, we partition the $\Ntile$ tiles into $\lceil N_{\mathrm{tile}}/l \rceil$ groups. Consider one tile group $\gamma$ with tile indices $t^{(\gamma)}_1,\dots,t^{(\gamma)}_l$. Then, for each input-tile coordinate $(i,j)\in\{0,\dots,\tile-1\}^2$, define the packed input vector 
\vspace{-2mm}
\begin{center}
\scalebox{0.9}{$
    \vx{d^{(\gamma)}_{i,j}}
    =
    \bigl(
        d^{(t^{(\gamma)}_1)}_{i,j},\dots,
        d^{(t^{(\gamma)}_l)}_{i,j}
    \bigr).
$}
\end{center}
\vspace{-2mm}
The parties hold the degree-$d$ \PSS{} $\pss{\vx{d^{(\gamma)}_{i,j}}}{d}$. Thus, each \PSS{} is indexed by a within-tile coordinate, while its slots range over different tiles. For simplicity, we assume that each group contains exactly $l$ tiles, padding the last group with zeros if necessary. Under this layout, all transform coefficients are public and uniform across slots, so the input, kernel, and output transforms are local linear operations. Since these coefficients depend only on the public matrices $A$, $B$, and $G$, they can be precomputed offline.

\noindent\textbf{Input Transform.} For each group $\gamma$ and Winograd-domain coordinates $(u,v)\in\{0,\dots,\tile-1\}^2$, the parties compute
\begin{center}
    $\pss{\vx{V^{(\gamma)}_{u,v}}}{d}
    =
    \sum_{i=0}^{\tile-1}\sum_{j=0}^{\tile-1}
    (B_{i,u}B_{j,v})\pss{\vx{d^{(\gamma)}_{i,j}}}{d}.$
\end{center}
\noindent\textbf{Kernel Transform.} Similarly, for each kernel coordinate $(p,q)\in\{0,\dots,r-1\}^2$, define the replicated \PSS{} $\pss{\vx{g_{p,q}}}{d} =    \bigl(g_{p,q}, \dots, g_{p,q}\bigr)$. Then for each $(u,v)\in\{0,\dots,\tile-1\}^2$, the kernel transformation is
\begin{center}
    $\pss{\vx{U_{u,v}}}{d}
    =
    \sum_{p=0}^{r-1}\sum_{q=0}^{r-1}
    (G_{u,p}G_{v,q})\pss{\vx{g_{p,q}}}{d}.$
\end{center}

\noindent\textbf{Multiplication.} After above transforms, the parties perform entry-wise multiplications in the Winograd domain. This is the only step requiring secret-secret multiplication. For each $(\gamma,u,v)$, compute
\vspace{-2mm}
\begin{center}
$\pss{\vx{M^{(\gamma)}_{u,v}}}{2d}
    =
    \pss{\vx{U_{u,v}}}{d}
    \cdot
    \pss{\vx{V^{(\gamma)}_{u,v}}}{d}.$
\end{center}
\noindent\textbf{Output Transform.} Finally, to transform back to spatial domain, compute for each output-tile coordinate $(a,b)\in\{0,\dots,m-1\}^2$, 
\vspace{-2mm}\[
    \pss{\vx{Y^{(\gamma)}_{a,b}}}{2d}
    =
    \sum_{u=0}^{\tile-1}\sum_{v=0}^{\tile-1}
    (A_{u,a}A_{v,b})\pss{\vx{M^{(\gamma)}_{u,v}}}{2d}.
\]
After the output transform, the parties apply degree reduction and truncation to return the result to degree $d$ and scale.
\vspace{-1mm}

\subsubsection{Fully Connected Layers}
\label{sec:fc}
By batching activations from multiple users, we express the fully connected layer as $W \in \Fp^{\nrow \times \ncol}$, $B \in \Fp^{\ncol \times \ncol}$, and $C = WB \in \Fp^{\nrow \times \ncol}$, where $\ncol$ is the input dimension, $\nrow$ is the output dimension, $W$ is the weight matrix, $B$ stacks $\ncol$ input vectors, and $C$ is the output. Matrix multiplication can be implemented via vector rotations, slot-wise multiplications, and summations, enabling parallelism under \PSS{}. Since the server holds $W$ pre-sharing, it can precompute rotated diagonals offline and share them once, eliminating rotation cost online. Therefore, the packing must ensure (i) operations map to slot-wise products and sums, and (ii) rotations apply only to $W$. We adopt the diagonal formulation of~\cite{moon2025thor}, which rewrites $C=WB$ as slot-wise products of matrix diagonals. However, unlike homomorphic encryption settings with large packing (e.g., $2^{15}$ slots), \PSS{} has limited packing (Sec.~\ref{sec:background-pss}), so even one diagonal may not fit in a single sharing. We therefore use a \textit{chunked} variant that preserves diagonal packing while respecting \PSS{} constraints. \update{We first describe the construction for $\nrow \mid \ncol$. Arbitrary $\nrow$ and $\ncol$ are handled below.}
\\
\noindent\textbf{Notation.} For a vector $\vx{v}=(v_0,\dots,v_{\ncol-1})\in\Fp^{\ncol}$, define the rotation as $
    (\rot{r}(\vx{v}))_t = \vx{v}_{t-r \bmod \ncol}$.
For $W\in\Fp^{\nrow \times \ncol}$, let
$\UdiagPrime{k}{W}\in\Fp^{\ncol}$ denote the extended upper diagonal of $W$:
\begin{center}
    $\UdiagPrime{k}{W}
    =
    \begin{cases}
        \Udiag{k}{W},
        & k < \nrow,\\[1mm]
        \rot{\nrow\lfloor k/\nrow\rfloor}\!\left(\Udiag{k\bmod \nrow}{W}\right),
        & \nrow \le k < \ncol.
    \end{cases}$
\end{center}

Here, $\Udiag{k}{W}\in\Fp^{\ncol}$ is the ordinary $k$-th upper diagonal of $W$, given by $\Udiag{k}{W}[t] = W[t\bmod \nrow,\; k+t \bmod \ncol].$
Then, for each output diagonal $r\in[\nrow]$, the diagonal formulation gives
\begin{center}
$\Ldiag{r}{C}
    =
    \sum_{s\in[\ncol]}
    \rot{r}\!\left(\UdiagPrime{s-r}{W}\right)
    \Had
    \Ldiag{s}{B}.$
\end{center}
Here, $\Ldiag{r}{C}$ is the $r$-th lower diagonal of $C$, and $\Had$ denotes entry-wise multiplication.
\\\noindent\textbf{Chunking.} Since each diagonal has length $\ncol$ and one \PSS{} holds at most $l$ values, we split each diagonal into $
    q=\left\lceil \frac{\ncol}{l}\right\rceil
$ chunks. For any $\vx{v}\in\Fp^{\ncol}$, write $
    \vx{v}=\bigl(\vx{v}^{(0)},\dots,\vx{v}^{(q-1)}\bigr),
$
where each chunk $\vx{v}^{(j)}$ is a vector containing up to $l$ consecutive entries, with zero padding in the final chunk if necessary. For each pair of diagonal indices $(s,r)$, define and chunk the rotated weight diagonal as $\vx{w_{s,r}}=\rot{r}(\UdiagPrime{s-r}{W})$, where $ \vx{w_{s,r}}=\bigl(\vx{w_{s,r}}^{(0)},\dots,\vx{w_{s,r}}^{(q-1)}\bigr)$. Similarly, define the input diagonal $\vx{b_s} = \Ldiag{s}{B}$, where $\vx{b_s} = \bigl(\vx{b_s}^{(0)},\dots,\vx{b_s}^{(q-1)}\bigr)$. The parties hold  $\pss{\vx{w_{s,r}}^{(j)}}{d}$ and $\pss{\vx{b_s}^{(j)}}{d}$ for the corresponding chunks. 
\\\noindent\textbf{Computation.} For each output diagonal $r$ and chunk $j$, the parties compute $
    \pss{\vx{S_r}^{(j)}}{2d}
    =
    \sum_s
    \pss{\vx{w_{s,r}}^{(j)}}{d}
    \cdot
    \pss{\vx{b_s}^{(j)}}{d}
$, which is the $j$-th chunk of the output diagonal $\Ldiag{r}{C}$. After all chunks are computed, the tuple $
    \bigl(\vx{S_r}^{(0)},\dots,\vx{S_r}^{(q-1)}\bigr)
$ represents the full diagonal $\Ldiag{r}{C}$. Correctness follows from the fact that chunking only partitions coordinates, while all operations remain coordinate-wise: $
    \Ldiag{r}{C}^{(j)}
    =
    \sum_s
    \vx{w_{s,r}}^{(j)}\Had \vx{b_s}^{(j)}.
$
Notice that the rotation must be applied before chunking: first compute $\vx{w_{s,r}}=\rot{r}(\Udiag{s-r}{W})$ as a full logical diagonal, and only then split it into chunks. Rotating each chunk independently would not, in general, give the correct alignment. After the products are accumulated, the parties perform the necessary degree reduction and truncation.
\\\noindent\textbf{\update{Arbitrary Dimensions.}} \update{For arbitrary $n_1,n_2$, let $h=\gcd(n_1,n_2)$. We partition $W$ and $C$ along the output dimension into $n_1/h$ row blocks $W_1,\ldots,W_{n_1/h}\in\mathbb{F}_p^{h\times n_2}$ and $C_1,\ldots,C_{n_1/h}\in\mathbb{F}_p^{h\times n_2}$, where each block consists of $h$ consecutive rows. For each block, $C_i=W_iB$. Since $h\mid n_2$, the diagonal and chunking construction above applies independently to each block, with $h$ replacing $n_1$. Concatenating the resulting row blocks recovers $C$. This changes only the public diagonal layout with all costs remain unchanged.}
\\\noindent\textbf{Note.} Prior secret-sharing-based secure training treat linear layers as simply matrix-matrix multiplications~\cite{mohassel2017secureml,chaudhari2019trident,liu2024scalable,mohassel2018aby3,falcon,byali2020flash}, and recent work has focused on optimizing non-linear layers~\cite{mohassel2017secureml,mohassel2018aby3,falcon,liu2024scalable}. We revisit the task of computing the linear layers and our techniques are therefore complementary to these directions (see App.~\ref{app:nonlinear} for details).
\subsection{Efficient Secure Exponentiation} \label{sec:exponentation}
Non-linearities in neural networks (e.g., ReLU, sigmoid, softmax) could be implemented via polynomial approximations in secure computation for efficiency. High accuracy often requires high-degree polynomials, increasing the number of multiplications and hence costly degree-reduction steps. The main bottleneck is exponentiation: computing $x^k$ naively requires $O(\log k)$ \textit{sequential} multiplications, each incurring a degree-reduction \textit{round}. In distributed settings such as FL, where communication dominates the cost, minimizing these sequential rounds is critical; accordingly, we design a more efficient exponentiation protocol under \PSS.

The key observation is that exponentiation becomes \emph{cheap} under additive sharing via the multinomial expansion. For ease of exposition, we first explain the idea for a scalar and then extend it to packed vectors. Let $\sss{x}{d}$ be a Shamir-sharing of a scalar $x$. Let $\ash{x}$ denote a $q$-out-of-$q$ additive sharing of $x$ among a set of designated parties $\mathcal Q=\{Q_1,\dots,Q_q\}, q\leq n$, where the share held by $Q_i$ is $\ash{x}_i$. Then, $x^k$ can be computed using multinomial expansion \vspace{-2mm}\begin{gather*}x^k
    =
    \left(\sum_{i=1}^q \langle x\rangle_i\right)^k
    =
    \sum_{\substack{a_1,\dots,a_q\ge 0\\ a_1+\cdots+a_q=k}}
    \binom{k}{a_1,\dots,a_q}
    \underbrace{\prod_{i=1}^q\left(\langle x\rangle_i \right)^{a_i}}_{\text{$O(\log q)$ Seq. rounds}}
\end{gather*}  
Each term $\left(\ash{x}_i\right)^{a_i}$, together with its multinomial coefficient, can be computed locally by $Q_i$, so the only interactive step is multiplying $\prod_{i=1}^q \left(\ash{x}_i\right)^{a_i}$, which requires $O(\log q)$
degree-reduction rounds, independent of $k$. 

To extend this to \PSS{}, let $\vx{x} = (x_1, \dots, x_l)$ be shared under \PSS. We use a two-level decomposition. First, the parties locally convert each slot into an additive sharing: $\{\ash{x_j}\}_{j \in [l]}\leftarrow \PssToAdd\!\left(\pssf{\vx{x}}{d}{\scr{f}}\right)$. However, directly assigning these additive shares to the designated
parties may violate the $t$-threshold privacy guarantee. To avoid this, each party \textit{re-splits} its per-slot shares $\ash{x_j}, j \in [l]$ into a fresh $q$-out-of-$q$ additive sharing and sends one piece to each party in $\mathcal{Q}$. Each designated party $Q_i$ then sums the values it receives by slot index, obtaining its additive share $\ash{\vx{x}}_i$ of the full vector. At this point, each $Q_i$ can locally compute the slot-wise powers $\left(\ash{x_1}_i\right)^{a_i}, \dots, \left(\ash{x_l}_i\right)^{a_i}$ for every required exponent $a_i$. Each $Q_i$ then reshares these as a fresh degree-$d$ \PSS, and the parties evaluate the product terms across these sharings in $O(\log q)$
degree-reduction rounds. The process is detailed in App. \ref{app:exp}. 

\noindent\textbf{Discussion.} The parameter $q$ controls the tradeoff between security and round complexity. The additive sharing over $\mathcal{Q}$ remains private as long as at least one designated party is honest. This can be achieved, for example, by forming $\mathcal Q$ as a small committee that contains at least one honest party with high probability using a committee election protocol~\cite{feige1999noncryptographic}. As discussed above, our approach is  beneficial when $2 + \log q < \log k$. In the extreme case of $q = 2$, exponentiation requires only a \textit{single} degree reduction round, independent of $k$. \update{While we emphasize that our experiments do not rely on this assumption, such a setting can be realistic in practice.} For example, a natural choice in FL is to include the server as one of the designated parties: as a commercial organization subject to accountability, it is unlikely to collude with users.

\subsection{Efficient Degree Reduction} \label{sec:degred}
As discussed in Sec.~\ref{sec:background}, share-wise multiplication of two degree-$d$ sharings produces a degree-$2d$ sharing, which must be reduced back to $d$ before further secret–secret multiplications. While degree reduction is well studied, we introduce a new protocol for two key reasons. First, our use of Mersenne primes rules out FFT-based optimizations. Second, our packed layout changes across layers, requiring frequent repacking of shares. A naive approach would fully unpack and repack at each step, incurring significant overhead. Instead, we fold repacking directly into degree-reduction, eliminating its overhead entirely. We also compare with prior work.

Suppose after multiplication the parties hold a degree-$2d$ \PSS{} $\pssf{\vx{z}}{2d}{\scr{h}}$ of $\vx{z}=(z_1,\dots,z_l)$, where $\scr{h}(\beta_i)=z_i$ for all $i\in[l]$ and party $P_j$ holds the share $\scr{h}(\alpha_j)$. Since $\scr{h}$ has degree-$2d$, each secret $\scr{h}(\beta_i)$ can be reconstructed by Lagrange interpolation as $\scr{h}(\beta_i)=\sum_{j=1}^\size \lambda^{(2d)}_{\alpha_j}(\beta_i)\, \scr{h}(\alpha_j)$. Hence the secret vector $(\scr{h}(\beta_1),\dots,\scr{h}(\beta_l))$ is a public linear function of the share vector $(\scr{h}(\alpha_1),\dots,\scr{h}(\alpha_\size))$. Degree reduction can therefore be viewed as two public linear maps: first recover the encoded slot values from the degree-$2d$ shares, then re-encode the same slots to shares of a fresh degree-$d$ \PSS{}.

Let $\scr{s}$ be a degree-$d$ polynomial satisfying
$
    \scr{s}(\beta_i)=\scr{h}(\beta_i), \forall i\in[l].
$
Let $\lambda^{(d)}_{\beta_i}(\alpha_k)$ denote the Lagrange coefficient for the degree-$d$ encoding map from the values at $\beta_1,\dots,\beta_l,\beta_{l+1},\cdots, \beta_d$ to the new share at $\alpha_k$. Then there exist public coefficients $\{c_{k,j}\}_{k,j\in[\size]}$ for 
\begin{center}
    $\scr{s}(\alpha_k)=\sum_{j=1}^\size c_{k,j}\, \scr{h}(\alpha_j),
\quad \text{where } c_{k,j}
    =
    \sum_{i=1}^l \lambda^{(d)}_{\beta_i}(\alpha_k)\lambda^{(2d)}_{\alpha_j}(\beta_i)$
\end{center}

Thus, $\scr{s}$ is a degree-$d$ polynomial holding the same secrets of $\scr{h}$. However, $\scr{s}$ is still a deterministic function of the original degree$-2d$ shares. Although it encodes the correct secrets, it is not distributed as a fresh degree-$d$ \PSS{}, since its remaining $t=d+1-l$ degrees of freedom are fixed rather than random. To re-randomize $\scr{s}$ without changing its values at the secret slots, each party contributes a random degree-$d$ zero encoding by sampling an independent random degree-$(t-1)$ polynomial $\scr{r}_j$. Let $\mu(X)=\prod_{i=1}^l (X-\beta_i)$, $\scr{w}_j(X)=\mu(X)\scr{r}_j(X)$. Then $\scr{w}_j$ has degree at most $d$ and vanishes at all secret slots. It preserves the secrets while randomizing the remaining $t$ degrees of freedom. Let
$
    \scr{r}(X)=\sum_{j=1}^\size \scr{r}_j(X), 
    \scr{w}(X)=\sum_{j=1}^\size \scr{w}_j(X)=\mu(X)\scr{r}(X),
$
The polynomial $\scr{q}(X)=\scr{s}(X)+\mu(X)\scr{r}(X)$ satisfies $\scr{q}(\beta_i)=\scr{h}(\beta_i)$ and is distributed exactly as a fresh degree-$d$ packed Shamir sharing of $z$.

The protocol realizes this transformation in one communication round. For each pair of parties $(P_j,P_k)$, sender $P_j$ computes
\begin{center}
  \vspace{-2mm}  $\msg{j}{k}
=
\underbrace{c_{k,j}}_{\textbf{offline}}
\;\onop{\cdot}\;
\underbrace{\scr{h}(\alpha_j)}_{\textbf{local share}}
\;\onop{+}\;
\underbrace{\mu(\alpha_k)\,\scr{r}_j(\alpha_k)}_{\textbf{offline}}$
\end{center}
and sends it to $P_k$. Party $P_k$ then sums all received messages:
\begin{center}
$\begin{aligned}
    \scr{q}(\alpha_k)
    =
    \overbrace{\sum_{j=1}^{\size} \msg{j}{k}}^{\textbf{online sum}}
    &=
    \sum_{j=1}^\size c_{k,j}\scr{h}(\alpha_j)
    +
    \mu(\alpha_k)\sum_{j=1}^\size \scr{r}_j(\alpha_k) \\
    &=
    \scr{s}(\alpha_k)+\mu(\alpha_k)\scr{r}(\alpha_k).
\end{aligned}$
\end{center}
Hence the parties obtain shares $\scr{q}(\alpha_1),\dots,\scr{q}(\alpha_\size)$, which form a fresh degree-$d$ \PSS{} $\pssf{\vx{z}}{d}{\scr{q}}$ of the same secret vector $\vx{z}$.

\subsubsection{Degree Reduction with Repacking}\label{sec:degred-repack}
The same idea extends to any public linear map on the secret slots. Instead of re-encoding the original secret vector, we may first apply a public linear map to the slot values, and then directly encode the transformed slots into the new degree-$d$ sharing. Hence slot permutation and more general repacking across one or more \PSS{}s can be fused into degree reduction at no extra interaction cost. 

Suppose the parties hold $\numPSS$ degree-$2d$ \PSS{}
$
    \pssf{\vx{z}^{(a)}}{2d}{\scr{h}^{(a)}} \text{for } a\in[\numPSS],
$
where $\vx{z}^{(a)}=(\scr{h}^{(a)}(\beta_1),\dots,\scr{h}^{(a)}(\beta_l))$. We want a fresh degree-$d$ \PSS{} $\pssf{\vx{z}'}{d}{\scr{q}}$, where the target slot vector $\vx{z}'=(z'_1,\dots,z'_l)$ is a public linear combination of the $\numPSS\cdot l$ total input secrets. Let $M^{(a)}\in\mathbb F_p^{l\times l}$ be a public matrix describing how the slots of the $a$-th input sharing contribute to the output:
\begin{center}
    $z'_i=\sum_{a=1}^{\numPSS}\sum_{u=1}^{l} M^{(a)}_{i,u}\,\scr{h}^{(a)}(\beta_u), 
    \qquad \forall i\in[l]$
\end{center}
Let $\scr{s}'$ be a degree-$d$ polynomial satisfying $\scr{s}'(\beta_i)=z'_i$ for all $i\in[l]$.
Then similarly, we obtain
\begin{center}
    $\scr{s}'(\alpha_k)=\sum_{i=1}^{l}\lambda^{(d)}_{\beta_i}(\alpha_k)\,z'_i
=\sum_{a=1}^{\numPSS}\sum_{j=1}^{\size} c^{(a)}_{k,j}\,\scr{h}^{(a)}(\alpha_j),$
\end{center}
where the public coefficients are
\begin{center}
    $c^{(a)}_{k,j}
    =
    \sum_{i=1}^{l}\sum_{u=1}^{l}
    \lambda_{\beta_i}^{(d)}(\alpha_k)\, 
    M^{(a)}_{i,u}\, 
    \lambda^{(2d)}_{\alpha_j}(\beta_u)$
\end{center}
As before, the parties re-randomize the deterministic encoding by adding a random degree-$d$ zero encoding. Let $\mu(X)=\prod_{i=1}^l(X-\beta_i)$, and let each party $P_j$ sample a random polynomial $\scr{r}_j(X)$ of degree at most $t-1$. For each pair of parties $(P_j,P_k)$, sender $P_j$ computes
\begin{center}
    $\msg{j}{k}
=
\overset{\textbf{online sum}}{\sum_{a=1}^{\numPSS}}
\underbrace{c^{(a)}_{k,j}}_{\textbf{offline}}
\;\onop{\cdot}\;
\underbrace{\scr{h}^{(a)}(\alpha_j)}_{\textbf{local share}}
\;\onop{+}\;
\underbrace{\mu(\alpha_k)\,\scr{r}_j(\alpha_k)}_{\textbf{offline}}$ \vspace{-2mm}
\end{center}
and sends $\msg{j}{k}$ to $P_k$, which party $P_k$ then sums and obtains\vspace{-2mm}
\begin{center}
    $\scr{q}(\alpha_k)
=
\overbrace{\sum_{j=1}^{\size} \msg{j}{k}}^{\textbf{online sum}}
=
\sum_{a=1}^{\numPSS}\sum_{j=1}^{\size} c^{(a)}_{k,j}\,\scr{h}^{(a)}(\alpha_j)
+
\mu(\alpha_k)\sum_{j=1}^{\size} \scr{r}_j(\alpha_k).$
\end{center}
Thus $\scr{q}(X)=\scr{s}'(X)+\mu(X)\sum_{j=1}^{\size}\scr{r}_j(X)$ is distributed as a fresh degree-$d$ \PSS{} $\pssf{\vx{z}'}{d}{\scr{q}}$ of the repacked vector $\vx{z}'$.

\subsubsection{Cost and Comparison}
We report communication and computation per party, and online communication rounds. The coefficients $c_{k,j}$ and $c_{k,j}^{(a)}$ depend only on public evaluation points, slot points, and, when applicable, the public repacking map, and can thus be precomputed offline. Both $\DegRed$ and $\DegRedRepack$ use a single all-to-all resharing round, where each party sends one field element to every other party. Therefore, each output \PSS{} costs one online round and $O(\size)$ communication. For basic degree reduction, each party computes $\size$ messages, each requiring $2+\log\size$ field operations, yielding $O(\size)$ online computation. With repacking, each message is a linear combination of multiple input sharings; if one output depends on $s$ inputs, the cost is $2s+\log\size$ per message. Compared to prior work, Franklin--Yung~\cite{franklin1992communication} extends BGW to \PSS{} with $O(\size)$ communication but incurs $O(\size^2)$ computation without FFT. The DN protocol~\cite{damgaard2007scalable} and its \PSS{} extension in DIK~\cite{damgaard2010perfectly} reduce communication to $O(1)$ via preprocessing, but still require $O(\size^2)$ computation. Thus, DIK is preferable when communication dominates, while \sysname{} is better for layout changes and compute.
\vspace{-2mm}
\subsection{Truncation} \label{sec:trunc}
We use Mersenne-prime fields ($p = 2^\ell - 1$) for efficient truncation. First, truncation reduces to bit shifts with a simple MSB-based correction. Second, they provide single-bit slack, allowing values up to $2^{\ell-1}$ without overflow. This lets us amortize truncation across multiple multiplications rather than after each step, which we later exploit (Sec.~\ref{sec:sys-description:training}).

We build on the truncation protocol of Liu et al.~\cite{liu2024scalable}, which assumes preprocessed randomness: a mask $r$, its truncated value $r'$, and MSB $r_{\msb}$. Given $x$, the protocol opens $c=x+r$, truncates $c$, and corrects using $r$ and the wraparound term $c_{\msb}(1-r_{\msb})$. However, their design targets standard Shamir sharing; we extend it to $\PSS$, where truncation must apply independently to each packed slot without cross-slot interference. This requires redesigning the randomness generation and correction step. We also avoid explicit preprocessing (Sec.~\ref{sec:sys-description:data}) and compare with Liu et al. in App.~\ref{app:eval:hmmpc}.
\\\noindent\textbf{Packed Randomness.} To apply scalar truncation slot-wise, parties receive common randomness in packed form. Let $\pss{\vx{a}}{d}$ contain slots $a_1,\dots,a_l$. For each $i\in[l]$, truncation requires a mask $r_i$, its truncated value $r_i'=\Trunc{f}(r_i)$, and correction bit $1-r_{i,\msb}$. Parties hold $\pss{\vx{r}}{d}$ and $\pss{\vx{r'}}{d}$ whose $i$-th slots contain $r_i$ and $r_i'$, respectively. The wraparound correction requires more care. After opening the masked value, parties learn a public bit $c_{i,\msb}$ for each slot. Since this bit is slot-dependent, it cannot be applied as a single public scalar to a \PSS{} containing all correction bits. Instead, for each $i\in[l]$, parties receive a sparse \PSS{} $\pss{\vx{\rho_i}}{d}$ whose only potentially nonzero slot is $i$, i.e., 
$\vx{\rho_i}=(0,\dots,0,1-r_{i,\msb},0,\dots,0).$
They multiply $\pss{\vx{\rho_i}}{d}$ by the corresponding public bit $c_{i,\msb}$ and sum the sparse sharings. Importantly, the sparse sharing contains $1-r_{i,\msb}$ directly rather than $r_{i,\msb}$. Otherwise, computing $1-\pss{r_{i,\msb}}{d}$ would place $1$ in every other slot, introducing nonzero corrections where slots should remain zero. Storing $1-r_{i,\msb}$ directly ensures that $c_{i,\msb}(1-r_{i,\msb})$ affects only the intended slot.
\\\noindent\textbf{Fixed-point Multiplication.}
The same idea extends to the truncation and DN multiplication protocol of~\cite{liu2024scalable}. For degree-$2d$ inputs, the mask becomes $\pss{\vx{r}}{2d}$, while $\pss{\vx{r'}}{d}$ and sparse corrections $\{\pss{\vx{\rho_i}}{d}\}_{i\in[l]}$ remain degree-$d$. The output is a degree-$d$ \PSS{} $\pss{\vx{\hat{a}}}{d}$ (App.~\ref{app:trunc}).

%% file: 6_sys_description.tex
\vspace{-2mm}\section{\sysname~Description}\vspace{-1mm}
\label{sec:sys-description}
We present the design of \sysname{} and show how the building blocks from the previous section compose into an end-to-end FL training protocol. The system has three phases: \emph{data sharing}, \emph{training}, and \emph{aggregation} (Fig.~\ref{fig:system}). We first describe the semi-honest protocol, then extend it to the malicious setting.

\noindent\textit{Notations.} Let $\Co = \{\Co_1, \cdots, \Co_\num\}$ denote $\num$ disjoint committees, each of size $\size$ with at most $\thr$ corrupt members, chosen from $\totalC$ users with corruption rate $\eta$.
\\\noindent\textit{Setup.} This one-time phase occurs when users register. As in prior work~\cite{ma2025armadillo,eiffel}, we assume a public key infrastructure (PKI)~\cite{diffie2022new}. Each user communicates with the server over a private authenticated channel, while user-to-user messages are relayed through the server and encrypted under the recipient's public key. Each user $\us_i$ is assigned a unique identifier $i \in [\totalC]$.
\subsection{Data Sharing Phase}\label{sec:sys-description:data}
For each training iteration, \sysname{} begins with a data sharing phase that secret-shares each user’s data across committees. We use the $t$-linear sharing optimization~\cite{goyal2021atlas} to reduce the communication cost of distributing shares. Here, two challenges arise: first, users must share their data blindly, without knowing the recipients. Under the static corruption model, where the adversary fixes corrupted parties in advance, committee assignment must therefore occur \emph{after} data sharing; otherwise, the adversary could corrupt an entire committee and compromise privacy. Second, some of our packing strategies (Sec.~\ref{sec:linear}) require combining data from $\pf$ users into a single vector for efficient computation. Since each user holds only its own data, the sharing procedure must allow committees to assemble these packed representations locally once all shares are received.

\sysname{} resolves both challenges in two steps: a blind masked sharing step, followed by targeted unmasking once committee identities are known. First, each user constructs $\numvec$ secret vectors. For each coordinate $j \in [\numvec]$, user $\us_i$ samples a masking vector $\vx{m_{ij}} \leftarrow \Fp^\pf$ and forms
\begin{center}
    $\vx{s_{ij}} = \big( \vx{x_i}[j] + \vx{m_{ij}}[1], \dots, \vx{x_i}[j] + \vx{m_{ij}}[\pf] \big) \in \Fp^\pf$
\end{center}
The user then shares each $\vx{s_{ij}}$ using $\Share{\vx{s_{ij}}}{d}$, with $\size = 2d+1$ to support one local multiplication and $d = \thr + \pf - 1$ for privacy. All shares are sent to the server. Since every entry is masked by fresh randomness, the shares reveal nothing about the underlying data. Next, \sysname{} runs a committee election protocol (Sec.~\ref{sec:election}) to select $\num$ disjoint committees of size $\size$, each with at most $\thr$ corrupted parties with overwhelming probability. This determines the members of each committee, assigns each user’s data to a committee, and fixes a slot $q_i \in [\pf]$ for each user within the packed vector. Now, each user $\us_i$ knows its committee and slot $q_i$. For each coordinate $j \in [\numvec]$, it constructs an unmasking vector as 
\begin{center}
    $\vx{v_{ij}} = \big( \vx{s_{ij}}[1], \dots, \vx{s_{ij}}[q_i-1],\vx{m_{ij}}[q_i], \vx{s_{ij}}[q_i+1] \dots \vx{s_{ij}}[\pf]\big) \in \Fp^l,$
\end{center}
and secret-shares each $\vx{v_{ij}}$ encrypted under its assigned committee members' public keys. 

To support truncation (Sec.~\ref{sec:trunc}), each user also provides randomness aligned to its slot by sharing sparse encodings of $(r_i, \Trunc{f}(r_i), \allowbreak 1-r_{i,\msb})$ which are routed to the same committee (details in App.~\ref{app:trunc}). Each committee member receives shares from $\totalC/\num$ users, decrypts them, and locally subtracts $\vx{v_{ij}}$ from $\vx{s_{ij}}$ to cancel masks, yielding vectors where $\vx{x_i}[j]$ appears only in slot $q_i$. The committee then aggregates $\pf$ such vectors with distinct slots to form packed shares aligned across slots; the randomness is aggregated identically, ensuring it remains aligned with each user’s data. This entire transformation is local and requires no additional communication.

\subsubsection{Committee Election Protocol.} 
\label{sec:election}
Following prior work~\cite{ma2025armadillo,ma2023flamingo}, we assume access to a public randomness beacon~\cite{rabin1983transaction} (e.g., NIST~\cite{nist_beacon20} or drand~\cite{drand}); if unavailable, it can be instantiated via committee election techniques~\cite{alon2024mpc}. From the beacon output, we derive $a, b \in \mathbb{Z}_{\totalC}$ with $\gcd(a,\totalC)=1$ and define $\pi(i) = (ai + b) \bmod \totalC$. \sysname{} uses $\pi$ to deterministically produce three public outputs: (1) committee membership, (2) user-to-committee assignments, and (3) each user’s slot in the packed vector.

\sysname{} partitions users into $\num$ disjoint committees of size $\size$ as contiguous blocks in the permuted order. For $j = 0, \dots, \num-1$, define
$
\Co_j = \{\, \us_i : \pi(i) \in \{j\size, j\size+1, \dots, j\size+\size-1\} \,\},
$
where $\num\size \le \totalC$. Each user $\us_i$ determines its committee by computing $r_i = \pi(i)$ and setting $j = \lfloor r_i / \size \rfloor$ (if $r_i < \num\size$).

To recover all members of $\Co_j$, users exploit that $\pi$ is invertible. Let $a^{-1}$ be the inverse of $a \bmod \totalC$. For each $r \in \{j\size, \dots, j\size+\size-1\}, \pi^{-1}(r) = a^{-1}(r - b) \bmod \totalC. $
Iterating over these $\size$ positions reconstructs $\Co_j$ in $\Omega(\size)$ time. Committee $\Co_j$ processes users with $\lfloor \pi(i)/\size \rfloor = j$. Each user computes $g_i = r_i \bmod (\totalC/\num)$ and sets its slot as $q_i = (g_i \bmod \pf) + 1$.

\begin{lemma}[Committee Safety Probability]
\label{lem:committee-safety}
The probability that all committees are safe, i.e., each $\Co_j$ contains at most $\thr$ corrupted users): (proof in App.~\ref{app:lem:committee-safety})
\vspace{-2mm}\begin{center}
$\Pr[\mathrm{safe}]
=
\frac{1}{\binom{\totalC}{\eta \totalC}}
\sum_{\substack{0 \le x_1,\dots,x_\num \le t \\ \sum_{i=1}^c x_i \le \eta \totalC}}
\left(
\prod_{i=1}^{\num} \binom{\size}{x_i}
\right)
\binom{\totalC - \num\size}{\eta \totalC - \sum_{i=1}^{\num} x_i}$
\end{center}
\end{lemma}

\subsection{Training Phase}\label{sec:sys-description:training}
After data sharing, each committee holds \PSS{} of its assigned inputs, arranged for the first layer. During training, committees independently compute gradients for their users under \PSS{}, enabling full parallelism with all computation local to each committee. At each iteration, the server distributes \PSS{} of the current model weights to every committee. Each committee performs a forward and backward pass using the secure layer primitives from Sec.~5, producing packed gradient shares that are passed to aggregation.
\vspace{-3mm}

\subsubsection{Forward Pass}
Each committee evaluates the model layer by layer on its packed input sharings. Linear layers use the packing-aware protocols from Sec.~\ref{sec:linear}. Convolutional layers use the packed Winograd layout, while fully connected layers use the diagonal representation and chunking strategy from Sec.~\ref{sec:fc}


\noindent\textbf{PolyRelu.} For non-linear layers, \sysname{} replaces ReLU with a low-degree polynomial approximation, PolyReLU. To balance accuracy and efficiency, \sysname{} incorporates several ML-aware optimizations: (i) He initialization to preserve variance and prevent vanishing or exploding gradients, (ii) quantization-aware polynomial approximation~\cite{diaa2024fast} for numerically stable coefficients, and (iii) an activation regularization term that penalizes inputs outside the regime where PolyReLU closely matches ReLU~\cite{diaa2024fast}.

\noindent\textbf{Inter-layer Connections.}
Different layers require different packing layouts. Convolutional layers pack the same Winograd coordinate across tiles, while fully connected layers use diagonal chunks. Even between convolutional layers, layouts do not align: Winograd outputs are non-overlapping $m\times m$ tiles, whereas the next layer expects overlapping $\tile\times\tile$ inputs. Thus, transitions such as Conv-to-Conv and Conv-to-FC require repacking. Since these are public linear transformations over slots, \sysname{} folds them into the degree-reduction step (Sec.~5.3), incurring no extra interaction cost.

\noindent\textbf{Choosing the cleanup path.}
After a secret-secret multiplication, \sysname{} may require three cleanup operations: degree reduction, truncation, and repacking. Degree reduction restores a degree-$2d$ sharing to degree $d$, truncation restores the fixed-point scale, and repacking adjusts the slot layout when the next computation expects a different packing. These operations have different urgency: degree reduction is mandatory before the next secret-secret multiplication, otherwise the degree would exceed $2d$ and the committee could no longer reconstruct the polynomial; repacking is needed only at layout boundaries; and truncation, while required for fixed-point correctness, can be delayed for a bounded number of multiplications. This is because \sysname{} operates over a 61-bit Mersenne field with a much smaller fractional scale (e.g., 16 bits), allowing several multiplications before truncation as long as range analysis rules out modular wraparound.

Delaying truncation has two benefits: it reduces accumulated rounding error, and, when no repacking is needed, exposes a cost tradeoff between $\DegRed$ and $\TruncDegRed$. The former incurs higher communication but lower online computation, while the latter reduces communication at the cost of higher online computation. \sysname{} therefore selects the cleanup protocol based on both layout requirements and whether truncation is due.
\begin{center}
    $\Cleanup{x} =
\begin{cases}
    \DegRed(x)
    & \text{degree reduction only}\\
    \TruncDegRed(x)
    & \text{truncation + Deg. red.}\\
    \DegRedRepack(x)
    & \text{repacking needed}
\end{cases}$
\end{center}
When repacking is required, \sysname{} uses $\DegRedRepack(x)$ and defers truncation to a later step, provided the fixed-point range remains safe. This is valid since repacking occurs only at layer transitions, and two such steps do not occur consecutively. In the next cleanup step, \sysname{} either applies $\TruncDegRed$ to discharge the deferred truncation, or $\DegRed$ if truncation is still not due.
\vspace{-3mm}

\subsubsection{Backward Pass}

The backward pass mirrors the forward pass. Given the secret-shared loss gradient, each committee propagates gradients in reverse order using the same packing-aware primitives. For PolyReLU, \sysname{} computes derivatives using powers of the activation input already cached during the forward pass. For convolutional layers, gradients w.r.t.\ inputs and kernels are evaluated using the same packed Winograd strategy. For FC layers, the input-gradient has the same structure as the forward pass: if $C = WB$, then $\Delta B = W^T \Delta C$, which is computed using the same chunked lower-diagonal representation by swapping operands. However, the weight-gradient $\Delta W = \Delta C B^T$ differs since both $\Delta C$ and $B$ are in lower-diagonal form, so their entries are not aligned for coordinate-wise products. To avoid an extra communication round for repacking, \sysname{} keeps $\Delta C$ in lower-diagonal form and uses rotated copies of $B$, generated during forward-pass degree-reduction and repacking. This increases computation but preserves the communication-round structure of the forward protocol.
\vspace{-1mm}


\subsection{Aggregation Phase}\label{sec:sys-description:agg}

At the end of training, each committee holds packed shares of its local gradient aggregates, and the goal is to combine these into a global aggregate. A transfer of shares across committees is unsafe, since an adversary corrupting parties in multiple committees could accumulate more than $t$ shares of the same polynomial and violate privacy. To avoid this, we exploit that the last operation in training is a multiplication and skip the last degree-reduction step, leaving all shares at degree $2d$. However, such polynomial formed by multiplying two degree-$d$ sharing polynomials is not distributed as a fresh random packed sharing and does not provide the privacy guarantee against up to $2\thr$ shares. Before transferring the aggregate across committees, we therefore add a fresh random degree-$2d$ zero-sharing that preserves the encoded aggregate while re-randomizing it. A $\PSS$ of degree $2d$ hides the $\pf$ secrets against up to $2\thr + \pf - 1$ shares, so even if an adversary obtains $\thr$ shares from both the source and receiving committees, it sees only $2\thr$ shares. This enables a simple aggregation protocol: for each committee $\Co_i$, party $P_{i,j}$ sends its degree-$2d$ share to the corresponding party $P_{1,j}$ in a designated primary committee $\Co_1$, which locally sums all received shares. As a result, $\Co_1$ obtains a degree-$2d$ sharing of the global aggregate, which is then jointly reconstructed and decoded to produce the final gradient for the server.
\vspace{-1mm}

\subsection{Extension to Malicious Threat Model}\label{sec:mal}

Under a malicious threat model, privacy follows from the threshold property of \PSS, while correctness requires additional well-studied checks~\cite{goyal2021atlas,goyal2020malicious,goyal2020guaranteed}. However, \sysname{} achieves a more efficient design by exploiting an observation specific to FL. Beyond protocol deviations, malicious users can corrupt the aggregate by submitting malformed data. Thus, to ensure \emph{data integrity}, each gradient must satisfy a validity condition. Maliciously secure aggregation alone is therefore insufficient; explicit per-gradient well-formedness checks are required.

\sysname{} turns this requirement into an efficiency gain. Since validation is necessary anyway, we perform the expensive gradient computation in the more efficient semi-honest model and enforce malicious security only during validation and aggregation. This reduces overhead by running the dominant computation in the cheaper semi-honest setting. This slightly relaxes the threat model, as discussed below. Finally, we handle a malicious server distributing inconsistent model parameters with a simple weight-consistency check.\vspace{-2mm}
\subsubsection{Weight Consistency}

A malicious server may distribute inconsistent model weights to different committees at the beginning of a round. \update{Prior work shows that such model inconsistency can allow the server to recover a target participant's gradient even when the updates are protected by secure aggregation~\cite{pasquini2022eluding}.} Therefore, before any computation begins, the committees must verify that they all received the same model. For each packed weight block, let $\pss{\vx{w^{(k)}}}{d}$ denote the sharing received by committee $\mathcal C_k$. \update{Each committee reshapes its sharing into a degree-$(d+t)$ sharing over a designated main committee $\mathcal C_{\mathbb{M}}$ by adding a fresh zero-sharing of degree $d+t$, generated using packed analogues of the techniques from~\cite{goyal2021atlas,goyal2020malicious}. This yields $\pss{\vx{\widetilde{w}^{(k)}}}{d+t}$, so $\mathcal C_{\mathbb{M}}$ holds aligned sharings from all committees.}

\update{To check consistency, $\mathcal C_{\mathbb{M}}$ computes differences $\pss{\vx{\delta_i}}{d+t} = \pss{\vx{\widetilde{w}^{(i)}}}{d+t} \allowbreak- \pss{\vx{\widetilde{w}^{(1)}}}{d+t}$ and forms $\pss{\vx{\Delta}}{d+t} = \sum_{i=2}^c r_i \pss{\vx{\delta_i}}{d+t}$, where $r_i$ is a public random challenge from committee $\mathcal C_i$.} If all weights are identical, then $\vx{\Delta}=\vx{0}$. This is checked without revealing by having each party send an exponential encoding of its share in a public group. The parties then combine these encodings to reconstruct the result in the exponent, and if the check evaluates to the identity element of the group, the values are consistent; otherwise, the check fails.

\update{Additionally, for each $\mathcal C_k$, $\mathcal C_{\mathbb{M}}$ verifies via SCRAPE~\cite{cascudo2017scrape} that the transmitted shares and packed slots of $\pss{\vx{\widetilde{w}^{(k)}}}{d+t}$ lie on a single degree-$(d+t)$ polynomial.} SCRAPE has perfect completeness and soundness error at most $1/|\Fp|$. The check accepts iff all slot equalities and consistency tests pass; otherwise, the protocol aborts.

\vspace{-1mm}
\subsubsection{Gradient Validation and Aggregation}
\label{sec:mal:gradient-validation}

\update{A malicious user may corrupt the model by submitting an invalid gradient, so each committee validates individual gradients before aggregation. Following prior work~\cite{eiffel,lycklama2023rofl}, we model validity as a predicate $\mathrm{Valid}(\cdot)$ over gradient vectors, such as a norm bound. In the malicious setting, per-user gradients are retained through the backward pass. At its end, the committee holds degree-$d$ \PSS{} whose slots jointly encode individual gradient coordinates. It converts each \PSS{} into per-slot degree-$\thr$ Shamir sharings and groups them by user to obtain a shared gradient vector $\g_\us$. The committee then securely evaluates $\mathrm{Valid}(\g_\us)$, obtaining a validity bit $b_\us$. We implement norm-bound validation using secure comparison, such as bit decomposition~\cite{liu2024scalable}, and DIK-based multiplication~\cite{damgaard2010perfectly}. The committee accepts only gradients with $b_\us=1$ and computes $\g_{\mathrm{agg}}=\sum_{\us:b_\us=1}\g_\us$. Full validation and consistency checks appear in App.~\ref{app:gradient-validation}.}

\update{The local aggregates are then combined across committees. Each committee raises its aggregate to degree $2\thr$ by adding a fresh degree-$2\thr$ zero-sharing~\cite{goyal2021atlas,goyal2020malicious} and sends the shares to the primary committee, which verifies consistency with a degree-$2\thr$ Shamir sharing before aggregation. We require $\thr<\pf$. Since $\size=2d+1=2\thr+2\pf-1$, the corresponding Reed--Solomon code has minimum distance $\size-2\thr=2\pf-1>2\thr$. Thus, changes by up to $\thr$ source and $\thr$ primary committee members affect at most $2\thr$ coordinates and cannot transform one valid codeword into another. After all contributions pass, the primary committee sums them into a degree-$2\thr$ sharing of the global aggregate and forwards to the server, which performs a final consistency check during reconstruction. This follows the semi-honest aggregation with extra consistency checks; see App.~\ref{app:gradient-validation}.}
\begin{table}[tbh!]
\centering
\footnotesize
\resizebox{\columnwidth}{!}{%
\begin{tabular}{l l l l}
\toprule
\textbf{Dataset} & \textbf{Name} & \textbf{Architecture} & \textbf{Params} \\
\midrule
BCW & $\bcwarch$ & $\mathrm{FC}(28;28)\rightarrow \mathrm{PolyReLU}\rightarrow \mathrm{FC}(28;28)$ & $\sim 1.6\text{K}$ \\
MNIST & $\mnistarch$ & $\update{\mathrm{Conv}(1;4)\rightarrow\mathrm{PolyReLU}\rightarrow\mathrm{Conv}(4;10)\rightarrow\mathrm{PolyReLU}\rightarrow\mathrm{GAP}\rightarrow\mathrm{FC}}$ & $\sim\! 500$ \\
CIFAR-10 & $\cifararch^{\ast}$ & \makecell[l]{\update{$\mathrm{FC}(4096;256)\rightarrow\mathrm{PolyReLU}\rightarrow\mathrm{FC}(256;64)\rightarrow\mathrm{PolyReLU}$}\\\update{$\mathrm{FC}(64;64)\rightarrow\mathrm{PolyReLU}\rightarrow\mathrm{FC}(64;16)$}} & $\sim\! 1.07\text{M}$ \\
Credit & $\creditarch$ & $\update{\mathrm{FC}(32;64)\rightarrow\mathrm{PolyReLU}\rightarrow\mathrm{FC}(64;4)}$ & $\sim\! \update{2.3}\text{K}$ \\
FashionMNIST & $\fmnistarch$ & $\update{\mathrm{Conv}(1;4)\rightarrow\mathrm{PolyReLU}\rightarrow\mathrm{Conv}(4;8)\rightarrow\mathrm{PolyReLU}\rightarrow\mathrm{GAP}\rightarrow\mathrm{FC}}$ & $\sim\! 400$ \\
SVHN & $\svhnarch$ & $\mathrm{Conv}(3;4)\rightarrow \mathrm{Conv}(4;8)\rightarrow \mathrm{GAP}\rightarrow \mathrm{FC}$ & $\sim\! 500$ \\
\bottomrule
\end{tabular}%
}
\footnotesize
\caption{\footnotesize Model architectures. GAP (App.~\ref{app:gap}) denotes global average pooling. $^\ast$Random projection applied to \cifararch\ (\textsc{Mal}), see App.~\ref{app:rp}.}
\vspace{-4mm}
\label{tab:arch}
\end{table}
\refstepcounter{figure}

\noindent\textbf{Discussion.} In \sysname{}, a malicious user can exhibit two behaviors: (i) \textit{Input manipulation:} supplying malformed data, yielding an invalid gradient. (ii) \textit{Protocol deviation:} deviating during execution to perturb other users’ gradients. Under full malicious security, both are addressed directly: protocol safeguards prevent (ii), while validation filters malformed inputs in (i). Under our relaxed approach, (ii) is not eliminated, but its effect is restricted. This is justified since correctness is defined by the validation predicate: a gradient is accepted iff $\mathsf{Valid}(\g)=1$. Thus, any adversarial perturbation goes undetected \textit{only} if the resulting gradient remains valid, aligning with the server’s acceptance criterion. The server can further limit such influence in two ways. First, it keeps the parameters of $\mathsf{Valid}(\cdot)$ private (e.g., the threshold $\rho$ in $|\g|_2 < \rho$), forcing adversaries to perturb blindly. Second, it monitors rejection rates: if the observed rate exceeds the estimated corruption rate $\eta$, this signals protocol deviations beyond malformed inputs, and the server aborts.

%% file: 7_security_analysis.tex
\vspace{-2mm}\section{Security Analysis}
\label{sec:security}
We use the simulation paradigm~\cite{goldreich2001foundations} for security guarantees.
\begin{mybox2}{Ideal Functionality of \sysname{} (Fig.~\thefigure)}
\label{fig:IF}
\footnotesize
\textbf{Parameters.} $\num$ committees, corruption rate $\eta$, committee size $\size$, threshold $t$, packing factor $l$, security parameter $\kappa$;
\textbf{Parties.} $\totalC$ users $\{\us_i\}_{i\in[N]}$ and server $\Ser$.

\textbf{Functionality.}
Receive honest-user inputs $\vx{x}_i$ and abort if $\textsf{Pr[safe]}<1-2^{-\kappa}$ (Lemma~\ref{lem:committee-safety}); randomly partition users into committees $\{\Co_j\}_{j\in[\num]}$. Receive weights $\{W_j\}$ and validation predicates $\{(\textsf{Valid}_j,\rho_j)\}$ from $\Ser$; abort if the weights differ. The adversary may provide gradients $\g'_i$ for corrupted users and additive perturbations $\delta_i$ for honest users $i\in L$.\footnote{Our protocols allow only additive adversarial errors~\cite{cryptoeprint:2015/154}.}
Compute each honest user's gradient $\g_i$ and set
$\hat{\g}_i=\g'_i$ if corrupted, $\g_i+\delta_i$ if $i\in L$, and $\g_i$ otherwise.
Output to $\Ser$, 
$S=\sum_{i=1}^N
\textsf{Valid}_{j(i)}(\hat{\g}_i,\rho_{j(i)})\hat{\g}_i,
$
along with each committee's number of rejected gradients, where $j(i)$ is the committee containing $\us_i$.
\end{mybox2}
\vspace{-2mm}
\begin{theorem} \label{thm:security}\sysname~achieves the ideal functionality in Fig. \ref{fig:IF} in the presence of a static, malicious adversary
corrupting one server and upto $\eta$ fraction of users (proof in App. \ref{app:security}). 
\end{theorem}

%% file: 8_evaluation.tex
\begin{figure*}[tbh!]
  \centering
  \includegraphics[width=0.49\textwidth]{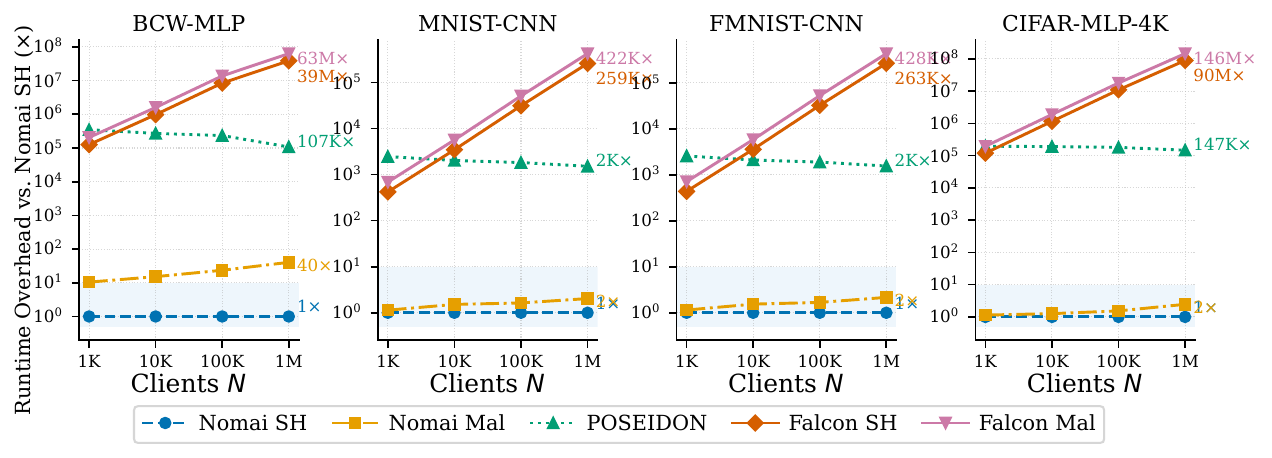}
  \hfill
  \includegraphics[width=0.49\textwidth]{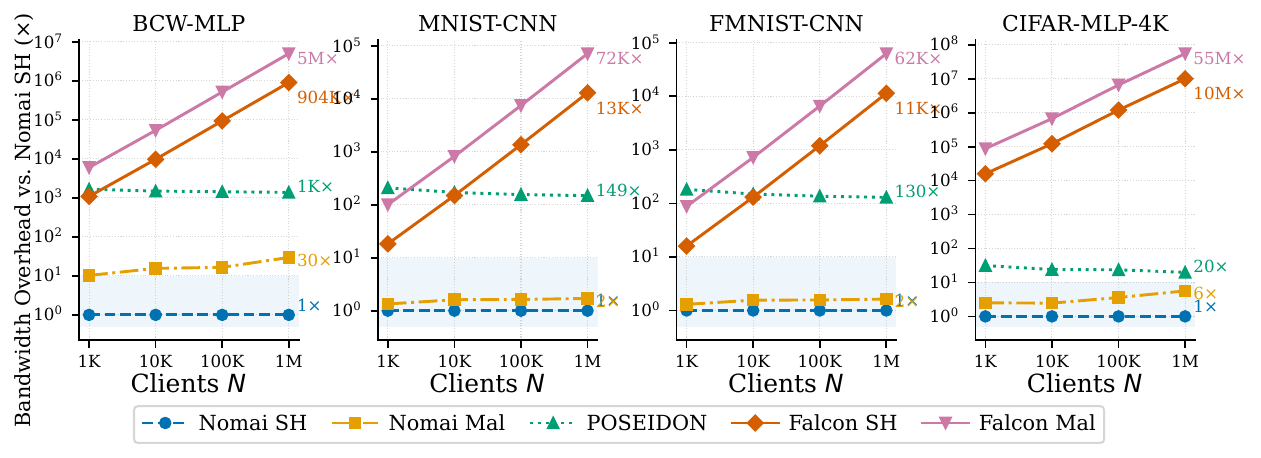}
  \vspace{-2mm}
  \caption{\footnotesize Per-machine overhead relative to \sysname{} (\textsc{SH}) as $\totalC$ increases. \textbf{Left:} compute overhead ($\times$) per iteration. \textbf{Right:} bandwidth overhead ($\times$) per iteration, both compared to baselines.}
  \vspace{-4mm}
  \label{fig:eval-eq1}
\end{figure*}

\begin{table}[tbh!]
\vspace{-3mm}
  \centering
  \caption{SH time vs. corruption rate $\eta$ at $\totalC=100\text{K}$.}
  \label{tbl:eval:beta-sweep}
  \scriptsize
  \setlength{\tabcolsep}{2pt}
  \resizebox{\columnwidth}{!}{%
  \begin{tabular}{@{}r rrr rrr rrr rrr rrr@{}}
    \toprule
    $\eta$ & \multicolumn{3}{c}{\bcwarch} & \multicolumn{3}{c}{\mnistarch} & \multicolumn{3}{c}{\fmnistarch} & \multicolumn{3}{c}{\svhnarch} & \multicolumn{3}{c}{\cifararch} \\
    \cmidrule(lr){2-4}\cmidrule(lr){5-7}\cmidrule(lr){8-10}\cmidrule(lr){11-13}\cmidrule(lr){14-16}
            & $\thr$ & $\size$ & \textit{(ms)} & $\thr$ & $\size$ & \textit{(ms)} & $\thr$ & $\size$ & \textit{(ms)} & $\thr$ & $\size$ & \textit{(ms)} & $\thr$ & $\size$ & \textit{(s)} \\
    \midrule
    0.01 &  13 &  53 &   3.3 &  12 &  43 &   283 &  12 &  43 &  222 &  12 &  43 &   301 &  16 & 103 & 1.72 \\
    0.05 &  26 &  79 &   6.5 &  24 &  67 &   691 &  24 &  67 &  542 &  24 &  67 &   734 &  34 & 139 & 2.01 \\
    0.10 &  42 & 111 &  11.8 &  39 &  97 &  1346 &  39 &  97 & 1048 &  42 & 111 &  1521 &  57 & 185 & 2.49 \\
    0.15 &  63 & 153 &  23.0 &  63 & 153 &  2809 &  63 & 153 & 2221 &  63 & 153 &  2845 &  85 & 241 & 3.06 \\
    0.20 &  92 & 211 &  44.5 &  92 & 211 &  5499 &  92 & 211 & 4360 &  92 & 211 &  5584 & 102 & 243 & 3.94 \\
    0.25 & 138 & 303 &  86.4 & 151 & 341 & 11727 & 151 & 341 & 9335 & 138 & 303 & 11254 & 151 & 341 & 5.44 \\

    \bottomrule
  \end{tabular}}
  \vspace{-2mm}
\end{table}
\vspace{-3mm}
\section{Evaluation}\vspace{-1mm}
\label{sec:evaluation}
\noindent \textbf{Setup.} \sysname{} is implemented using NTL~\cite{shoup2001ntl}. Experiments run on a 12-core Intel Xeon CPU with 125~GB RAM. We report per-machine compute per training iteration (ms) and bandwidth (MB), under both semi-honest (\textsc{SH}) and malicious (\textsc{Mal}) settings. Our setup spans $\totalC \in [1\text{K}, 1\text{M}]$ users, $\num \in [1,25{,}000]$ committees, and $\pf \in [10,50]$. The threshold $\thr$ is chosen so that the probability any committee is unsafe is at most $10^{-12}$ under a $\eta=1\%$ corruption rate (realistic for FL~\cite{corruption}), modeling corrupted parties per committee via a hypergeometric distribution; this yields committee size $\size = 2(\thr+\pf-1)+1 \in {37,\ldots,139}$. For gradient validation, we use an $\mathnormal{l}_2$-norm bound test~\cite{l2norm1, eiffel}.

\noindent \textbf{Datasets.} We use the datasets: (i) BCW \cite{bcw}; (ii) MNIST \cite{mnist_uci_1998}; (iii) CIFAR-10 \cite{krizhevsky2009learning}; (iv) Credit \cite{yeh2009comparisons}; (v) FashionMNIST \cite{FMNIST}; \& (vi) SVHN \cite{netzer2011reading}, with dimensions padded to power of two. 

\noindent \textbf{Model Architectures.}
We evaluate six workloads: BCW, Credit, \& CIFAR-10 use fully connected models; MNIST, FashionMNIST, \& SVHN use compact CNNs (Tab.~\ref{tab:arch}), training details and accuracy in App.~\ref{app:eval:acc}.
\\\noindent \textbf{Baselines.}
We compare \sysname{} against two systems providing both data and weight privacy: \textsc{Poseidon}~\cite{poseidon}, an HE-based $N$-party framework, and \textsc{Falcon}~\cite{falcon}, a maliciously secure three-server MPC protocol (Tab.~\ref{tab:comparison}). \textsc{Poseidon} is the closest prior single-server FL system with weight privacy, albeit only in the \textsc{SH} model and without data integrity. \textsc{Falcon} is a representative secure-training baseline, but assumes multiple non-colluding servers. We also compare against \textsc{HMMPC}~\cite{liu2024scalable}, an \textsc{SH} protocol using unpacked Shamir secret sharing (App.~\ref{app:eval:hmmpc}). These are the closest baselines to \sysname{} in secure-training functionality. A separate line of single-server work focuses only on \emph{secure aggregation}~\cite{eiffel,lycklama2023rofl,bell2023acorn,ma2025armadillo,ma2023flamingo}, leaving model weights and training computation exposed. \sysname{} instead brings both inside MPC, secret sharing the model and user data throughout the forward and backward pass. We therefore do not treat secure aggregation systems as functionality-matched baselines. Nevertheless, in App.~\ref{app:eval:secure-aggregation}, we compare against two state-of-the-art systems, \textsc{Flamingo}~\cite{ma2023flamingo} and \textsc{Willow}~\cite{bell2025willow}. Even comparing \sysname{} committee members directly with their aggregation servers, \sysname{} is substantially faster despite secure training.


\noindent \textbf{Experiments.} We evaluate \sysname{} along two axes:\\
$\bullet$~\textbf{(EQ1)} How does \sysname\ compare to prior systems?\\
$\bullet$~\textbf{(EQ2)} How does \sysname\ scale with $\totalC$, model size, and $\eta$?

Additionally, we discuss how to choose optimal \sysname{} parameters in (App.~\ref{app:eval:eq2}).


\begin{figure}[tbh!]\vspace{-3mm}
  \centering
  \includegraphics[width=\linewidth]{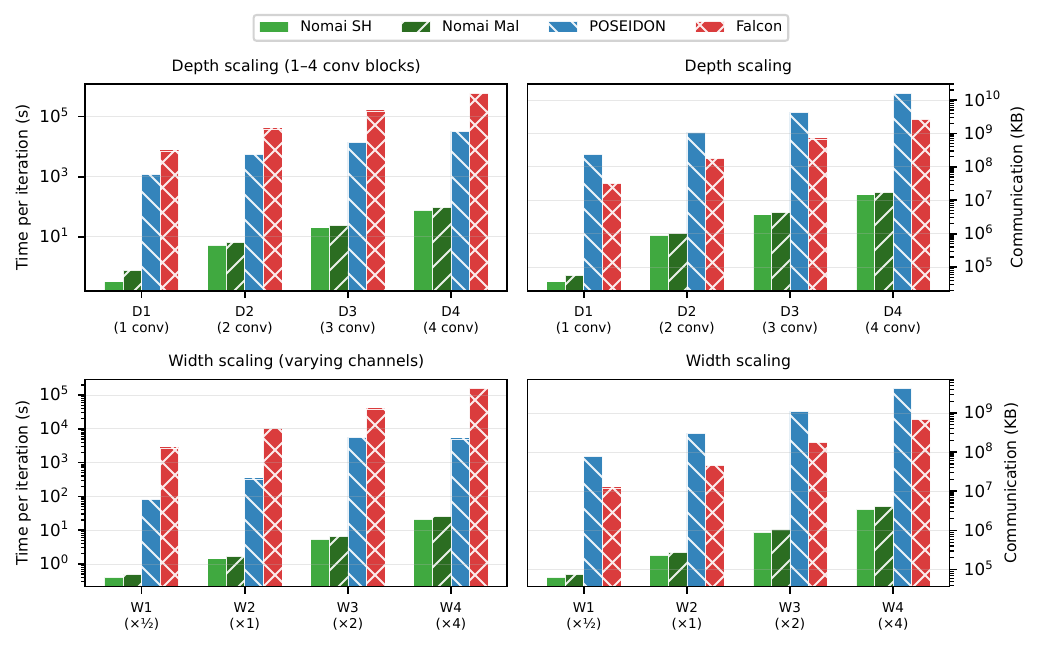}
\vspace{-8mm}  \caption{\footnotesize Per-iteration compute (left) and bandwidth (right) at $\totalC=100\text{K}$.
           Top: depth scaling (D1--D4); bottom: width scaling (W1--W4).}
  \label{fig:eval:scale-arch}\vspace{-3mm}
\end{figure}
\subsection{Comparison with Prior Work (EQ1)}
\label{sec:eval:comparison}
\autoref{fig:eval-eq1} shows compute and bandwidth overhead of baselines over \sysname{}. Tab.~\ref{tbl:eval:eq1-compute} compares per-machine iteration latency and
Tab.~\ref{tbl:eval:eq1-comm} compares bandwidth. We report LAN numbers; WAN also reported in App.~\ref{app:eval:wan}; configurations in App.~\ref{app:tab:optimal}.
\\\textit{In terms of raw performance} ( Tab.~\ref{tbl:eval:eq1-compute} and Tab.~\ref{tbl:eval:eq1-comm}), \sysname\ achieves \textbf{exceptionally low per-machine cost}. Concretely, compute is $<350\,\mathrm{ms}$ (\textsc{SH}) and $<700\,\mathrm{ms}$ (\textsc{Mal}) for MNIST, while scaling to larger models for CIFAR ($3\,\mathrm{s}$ \textsc{SH}, $6\,\mathrm{s}$ \textsc{Mal}). Bandwidth is also efficient: in the \textsc{SH} setting, per-machine bandwidth remains essentially constant in $\totalC$, with slight increases in \textsc{Mal}.  \textit{In terms of gains over baselines} (shown in \autoref{fig:eval-eq1}), \sysname\ outperforms all prior systems significantly. \emph{Against \textsc{Falcon}}, \sysname\ (\textsc{SH}) achieves \textbf{$10^{2}$--$10^{8} \times$ lower per-machine compute}, with the gap increasing linearly in $\totalC$. Specifically, gains range from $420 \times$ on \mnistarch\ at $\totalC = 1\text{K}$ to $10^{8} \times$ at $\totalC = 1\text{M}$ on \cifararch, while bandwidth improvements range from $\! 20 \times$ (CNNs at $\totalC=1\text{K}$) to $\! 10^{7} \times$ (\cifararch\ at $\totalC=1\text{M}$). \emph{Against \textsc{Poseidon}}, despite its per-machine cost being independent of $\totalC$, \sysname\ still achieves \textbf{$2 \times 10^{3}$--$1.7 \times 10^{5} \times$} improvements, including $\! 10^{5} \times$ speedup on \bcwarch\ and \cifararch\ at $\totalC = 1\text{M}$. Corresponding bandwidth reductions are $10^{3} \times$ for \bcwarch, 100$\times$ for \mnistarch\ and \fmnistarch, and $20 \times$--$30 \times$ for \cifararch. In the \textsc{Mal} setting, \sysname{} incurs only modest overhead relative to \textsc{SH} ($\! 10$--$30 \times$ for \bcwarch, $\! 2 \times$ for \mnistarch\ and \fmnistarch, and $\! 2$--$5 \times$ for \cifararch), while still significantly outperforming all baselines. These improvements arise from eliminating the bottlenecks of prior work (e.g., $(\cifararch, \totalC = 10^6)$: $2.11\,\mathrm{s}$ vs.\ $86.4\,\mathrm{hours}$ (\textsc{Poseidon}) and $\! 6$ yrs (\textsc{Falcon})). 

\subsection{Scalability (EQ2)}
\label{sec:eval:scale}

\noindent \textbf{(i) Scaling with Number of Users $\totalC$}. \autoref{fig:eval-eq1} shows that \sysname\ scales \textit{seamlessly} to large populations, achieving \textit{constant} compute and bandwidth w.r.t.\ $\totalC$ in \textsc{SH}. Across $\totalC=1\text{K}$ to $1\text{M}$, both vary by less than $4\times$, making \sysname{} effectively independent of $\totalC$. In \textsc{Mal}, costs grow slightly due to consistency checks scaling with $\num$, with bandwidth increasing by at most $\!3.5\times$. Even at $\totalC{=}1\text{M}$: compute ranges only from $\!$289\,ms (\bcwarch) to $\!$5.18\,s (\cifararch), and bandwidth from $\!$40\,MB (\bcwarch) to $\!$472\,MB (\cifararch).
\\\noindent \textbf{(ii) Scaling with Model Architecture}. 
We fix $\totalC=100\text{K}$ and train CIFAR-10 CNNs from two families (Tab.~\ref{tab:arch-families}): depth scaling adding conv blocks, and width scaling increasing channel width. \autoref{fig:eval:scale-arch} shows that \sysname{} scales smoothly with model size: compute grows from 0.33s to 78.7s with depth and 0.41s to 21.0s with width, while bandwidth follows the same trend. The relative overhead of malicious security decreases for larger models as the training workload increasingly dominates. In contrast, \textsc{Poseidon} is $200$--$3{,}600\times$ slower due to HE costs, and \textsc{Falcon} $3{,}700$--$22{,}000\times$ slower due to sequential mini-batching. 
\\\noindent \textbf{(iii) Scaling with Corruption Rate $\eta$}. We fix $\totalC=100\text{K}$ and vary $\eta \in [0.01,0.25]$. Larger $\eta$ requires larger $\thr$ to keep the probability of any unsafe committee below $10^{-12}$, increasing $\size=2(\thr+\pf-1)+1$. Larger committees reduce the number of disjoint committees, increasing the users handled per committee. Per-machine cost grows from both larger MPC committees and larger per-committee workloads. This particularly affects CNNs, where each user contributes many convolution tiles. For \mnistarch{}, per-machine time grows from $283\,\text{ms}$ at $\eta=0.01$ to $1.35\,\text{s}$ at $\eta=0.10$, and $11.7\,\text{s}$ at $\eta=0.25$. For fixed $\eta$, however, increasing $\totalC$ mainly creates more committees, so per-machine cost changes weakly  (Tab.~\ref{tbl:eval:beta-sweep}).

%% file: 11_conclusion.tex
\vspace{-2mm}\section{Conclusion}
\label{sec:conclusion}\vspace{-1mm}
We presented \sysname{}, the first scalable FL system to achieve data privacy, weight privacy, and integrity in the single-server setting. Our design introduces new cryptographic and ML-aware techniques for efficient secure training.

%% file: 12_appendix.tex

\section{Protocol Details}
In this section, we provide an extended description of protocols used in our paper, building on Sec.~\ref{sec:protocol-details}.

\subsection{\sysname{}: End-to-end Protocol}
We present the complete end-to-end protocol of \sysname{} in Fig.~\ref{alg:end-to-end}.

\subsection{Discussion}
\label{app:nonlinear}
A large body of prior work~\cite{mohassel2017secureml,mohassel2018aby3,falcon,liu2024scalable} optimizes exact non-linear layer computation in secure ML. These techniques are complementary and can be integrated into our framework. However, they typically incur multiple rounds of interaction and substantial correlated randomness, making them less suitable for our massively distributed setting where communication dominates.

Instead, we adopt polynomial approximations for non-linearities, which are more amenable to efficient packed computation. To preserve accuracy, we incorporate ML-aware optimizations that stabilize training under polynomial activations (see Sec.~\ref{sec:sys-description:training}), and empirically observe competitive performance with PolyReLU-style approximations. Our design does not preclude exact non-linear protocols, which can be substituted if desired. Our core contributions, particularly for linear layers and system-level optimizations, are orthogonal to this choice.

\subsection{Efficient Linear Layers}
\label{app:linear}

Alg.~\ref{alg:packed-winograd} gives the full multi-channel packed Winograd
convolution protocol used in Sec.~\ref{sec:linear}. The main text presents the
single-input, single-output channel case for clarity; the algorithm below adds
input-channel accumulation and produces packed output tiles for all output
channels. Here, $c\in[\Cin]$ indexes input channels, $o\in[\Cout]$ indexes
output channels, and $\gamma$ indexes tile groups.

\subsection{Efficient Secure Exponentiation}
\label{app:exp}
We give the full protocol for efficient secure exponentiation in
Alg.~\ref{alg:exp}.  It follows the construction described in Sec.~\ref{sec:exponentation}: convert \PSS{} to additive shares, repartition to $\mathcal Q$, locally compute powers, and
recombine via the multinomial expansion. Correctness and security are proved in App.~\ref{app:security}.

\subsection{Efficient Degree Reduction}
We give the explicit protocols for the degree-reduction procedures described
in Sec.~\ref{sec:degred}. Alg.~\ref{alg:degred} gives the basic packed
degree-reduction protocol, which maps a degree-$2d$ \PSS{} to a fresh
degree-$d$ \PSS{} of the same slots. Alg.~\ref{alg:degred-repack} gives the
repacking variant, which first applies a public linear map to the slots and
then directly outputs a fresh degree-$d$ \PSS{} of the repacked vector. Correctness and security are proved in App.~\ref{app:if:degred}.

\begin{mybox2}{\textsc{Alg.~\ref{alg:exp}: SecureExponentiation} ($\Pi_{\mathsf{Exp}}^{\mathsf{Pack}}$)}
\label{alg:exp}
\textbf{Input:} $\pss{\vx{x}}{d}$, exponent $k$, designated set $\mathcal Q=\{Q_1,\dots,Q_q\}$ \\
\textbf{Output:} $\pss{\vx{x^k}}{d}$
\begin{enumerate}[nosep, leftmargin=*, label=\arabic*.]

\item \textbf{Additive decomposition.}
Convert the packed sharing into slot-wise additive sharings:
\[
\begin{aligned}
\{\ash{x_j}\}_{j\in[l]}
&\leftarrow
\PssToAdd(\pss{\vx{x}}{d}), \\
\ash{x_j}
&=(u_{1,j},\dots,u_{n,j}).
\end{aligned}
\]

\item \textbf{Repartition to $\mathcal Q$.}
Each $P_i$ splits $u_{i,j}=\sum_{r=1}^q u_{i,j}^{(r)}$ and sends $u_{i,j}^{(r)}$ to $Q_r$.  
Each $Q_r$ sums received shares
\[
z^{(r)}_j=\sum_{i=1}^n u_{i,j}^{(r)},
\quad
\text{so that } x_j=\sum_{r=1}^q z^{(r)}_j.
\]

\item \textbf{Local powering.}
Each $Q_r$ holds $\vx{z}^{(r)}=(z^{(r)}_1,\dots,z^{(r)}_l)$ and locally computes
\[
(\vx{z}^{(r)})^e
=
\bigl((z^{(r)}_1)^e,\dots,(z^{(r)}_l)^e\bigr)
\quad
\forall e\in\{0,\dots,k\}.
\]

\item \textbf{Resharing.}
Each $Q_r$ shares $\pss{(\vx{z}^{(r)})^e}{d}$.

\item \textbf{Recombination.}
Output
\[
\pss{\vx{x^k}}{d}
=
\sum_{\substack{a_1,\dots,a_q\ge 0\\ a_1+\cdots+a_q=k}}
\binom{k}{a_1,\dots,a_q}
\prod_{r=1}^q
\pss{(\vx{z}^{(r)})^{a_r}}{d},
\]
products evaluated with a balanced multiplication tree.

\end{enumerate}
\end{mybox2}
\noindent\begin{minipage}{\columnwidth}
\captionof{figure}{Secure Exponentiation Protocol} 
\end{minipage} 

\subsection{Truncation}
\label{app:trunc}

We give the full packed truncation protocols used in Sec.~\ref{sec:trunc}.
We first recall the Mersenne-field truncation semantics used by Liu
et al.~\cite{liu2024scalable}. Let $p = 2^\ell - 1$ be a Mersenne prime. Given an integer $x \in [-(p-1)/2, (p-1)/2]$, we represent $x$ in $\mathbb{F}_p$ by computing $x \bmod p \in [0, p-1]$. Thus, representatives in $[0, (p-1)/2]$ correspond to nonnegative integers, while representatives in $((p-1)/2, p-1]$ correspond to negative integers.

\begin{mybox2}{\textsc{Alg. \ref{alg:packed-winograd}: PackedWinogradConv} ($\Pi_{\mathsf{WinogradConv}}^{\mathsf{Pack}}$)}

\textbf{Input:} Packed input tiles 
$\{\pss{\vx{d^{(c,\gamma)}_{i,j}}}{d}\}_{c,\gamma,i,j}$ and packed kernels 
$\{\pss{\vx{g^{(o,c)}_{p,q}}}{d}\}_{o,c,p,q}$ \\
\textbf{Output:} Packed output tiles 
$\{\pss{\vx{Y^{(o,\gamma)}_{a,b}}}{d}\}_{o,\gamma,a,b}$ \\
\textbf{Public parameters:} Winograd matrices $A,B,G$, tile size $\tile=m+r-1$ \\
\textbf{Index Ranges:} $c\in[\Cin]$, $o\in[\Cout]$, $(i,j),(u,v)\in\{0,\dots,\tile-1\}^2$,  $(p,q)\in\{0,\dots,r-1\}^2$, and $(a,b)\in\{0,\dots,m-1\}^2$.
\begin{enumerate}[nosep, leftmargin=*, label=\arabic*.]
    \item $\pss{\vx{V^{(c,\gamma)}_{u,v}}}{d}
            \leftarrow
            \sum_{i=0}^{\tile-1}\sum_{j=0}^{\tile-1}
            (B_{i,u}B_{j,v})\pss{\vx{d^{(c,\gamma)}_{i,j}}}{d}$ 
            for each $c,\gamma,(u,v)$

    \item $\pss{\vx{U^{(o,c)}_{u,v}}}{d}
            \leftarrow
            \sum_{p=0}^{r-1}\sum_{q=0}^{r-1}
            (G_{u,p}G_{v,q})\pss{\vx{g^{(o,c)}_{p,q}}}{d}$
            for each $o,c,(u,v)$
    
    \item $\pss{\vx{Z^{(o,c,\gamma)}_{u,v}}}{2d}
        \leftarrow
        \pss{\vx{U^{(o,c)}_{u,v}}}{d}\cdot
        \pss{\vx{V^{(c,\gamma)}_{u,v}}}{d}$
        for each $o,c,\gamma,(u,v)$

    \item $\pss{\vx{M^{(o,\gamma)}_{u,v}}}{2d}
        \leftarrow
        \sum_{c=1}^{\Cin}
        \pss{\vx{Z^{(o,c,\gamma)}_{u,v}}}{2d}$
        for each $o,\gamma,(u,v)$
        
    \item $\pss{\vx{\widetilde{Y}^{(o,\gamma)}_{a,b}}}{2d}
        \leftarrow
        \sum_{u=0}^{\tile-1}\sum_{v=0}^{\tile-1}
        (A_{u,a}A_{v,b})\pss{\vx{M^{(o,\gamma)}_{u,v}}}{2d}$
        for each $o,\gamma,(a,b)$

    \item $\pss{\vx{Y^{(o,\gamma)}_{a,b}}}{d}
        \leftarrow
        \Cleanup{\pss{\vx{\widetilde{Y}^{(o,\gamma)}_{a,b}}}{2d}}$
        for each $o,\gamma,(a,b)$
\end{enumerate}
\end{mybox2}
\noindent\begin{minipage}{\columnwidth}
\captionof{figure}{Winograd Convolution under \PSS\label{alg:packed-winograd}} 
\end{minipage}

\begin{mybox2}{\textsc{Alg. \ref{alg:trunc}: PackedTruncation} (\(\Pi_{\mathrm{Trunc}}^{\mathrm{Pack}}\))}
\textbf{Input:} $\pss{\vx{a}}{d}$, where $a_i \in [-2^{\ell-2},\,2^{\ell-2})$ for each $i\in[l]$ \\
\textbf{Output:} $\pss{\vx{\hat{a}}}{d}$, where $\hat{a}_i = \Trunc{f}(a_i) + 2E$ for each $i\in[l]$ \\
\textbf{Correlated Randomness:}
$\pss{\vx{r}}{d}$, $\pss{\vx{r'}}{d}$, and
$\{\pss{\vx{\rho_i}}{d}\}_{i\in[l]}$, where $\rho_i[i]=1-r_{i,\msb}$
and all other slots are zero.
\begin{enumerate}[nosep, leftmargin=*, label=\arabic*.]
    \item $\pss{\vx{b}}{d} \leftarrow \pss{\vx{a}}{d} + 2^{\ell-2}$

    \item $\pss{\vx{c}}{d} \leftarrow \pss{\vx{b}}{d} + \pss{\vx{r}}{d}$

    \item $\vx{c} = \{c_1,\dots,c_l\} \leftarrow \Reveal{\pss{\vx{c}}{d}}$, and compute $\vx{c'} \leftarrow \Trunc{f}(\vx{c})$ and $c_{i,\msb}$ for every $i\in[l]$

    \item $\pss{\vx{c'}}{d} \leftarrow \Share{\vx{c'}}{d}$

    \item $\pss{\vx{e_i}}{d} \leftarrow c_{i,\msb}\cdot \pss{\vx{\rho_i}}{d}$
for each $i\in[l]$

    \item $\pss{\vx{e}}{d} \leftarrow \sum_{i=1}^l \pss{\vx{e_i}}{d}$

    \item Output $\pss{\vx{\hat{a}}}{d} \leftarrow \pss{\vx{c'}}{d} - \pss{\vx{r'}}{d} + \pss{\vx{e}}{d} \cdot (2^{\ell-f}-1) - 2^{\ell-f-2}$
\end{enumerate}
\end{mybox2}
\noindent\begin{minipage}{\columnwidth}
\captionof{figure}{Packed Truncation Protocol\label{alg:trunc}} 
\end{minipage}

\begin{mybox2}{\textsc{Alg.~\ref{alg:degred-repack}: PackedDegRed with Repacking} ($\Pi_{\mathsf{DegRed\text{-}Repack}}^{\mathsf{Pack}}$)}
\textbf{Input:} $\{\pssf{\vx{z}^{(a)}}{2d}{\scr{h}^{(a)}}\}_{a\in[\numPSS]}$,
and public repacking matrices $\{M^{(a)}\in\mathbb F_p^{l\times l}\}_{a\in[\numPSS]}$ \\
\textbf{Output:} A fresh degree-$d$ \PSS{} $\pssf{\vx{z}'}{d}{\scr{q}}$, where
\[
z'_i=\sum_{a=1}^{\numPSS}\sum_{u=1}^{l}
M^{(a)}_{i,u}\scr{h}^{(a)}(\beta_u)
\quad \forall i\in[l].
\]

\begin{enumerate}[nosep, leftmargin=*, label=\arabic*.]

\item \textbf{Offline setup.}
For every $a\in[\numPSS]$ and $k,j\in[\size]$, precompute the public coefficient
\[
c^{(a)}_{k,j}
=
\sum_{i=1}^{l}\sum_{u=1}^{l}
\lambda_{\beta_i}^{(d)}(\alpha_k)\,
M^{(a)}_{i,u}\,
\lambda_{\alpha_j}^{(2d)}(\beta_u).
\]
Each party $P_j$ samples a random degree-$(t-1)$ polynomial $\scr{r}_j(X)$
and precomputes $\mu(\alpha_k)\scr{r}_j(\alpha_k)$ for every receiver $P_k$, where
\[
\mu(X)=\prod_{i=1}^{l}(X-\beta_i).
\]

\item \textbf{Resharing.}
For each receiver $P_k$, party $P_j$ sends
\[
\msg{j}{k}
=
\sum_{a=1}^{\numPSS}
c^{(a)}_{k,j}\scr{h}^{(a)}(\alpha_j)
+
\mu(\alpha_k)\scr{r}_j(\alpha_k).
\]

\item \textbf{Share reconstruction.}
Each party $P_k$ sums all received messages:
\[
\scr{q}(\alpha_k)
=
\sum_{j=1}^{\size}\msg{j}{k}.
\]

\item \textbf{Output.}
Each $P_k$ outputs $\scr{q}(\alpha_k)$. The resulting shares form a fresh
degree-$d$ \PSS{} $\pssf{\vx{z}'}{d}{\scr{q}}$.

\end{enumerate}
\end{mybox2}
\noindent\begin{minipage}{\columnwidth}
\captionof{figure}{Packed Degree Reduction with Repacking Protocol\label{alg:degred-repack}}

\end{minipage}

\begin{definition}[Truncation over a Mersenne field]
For $f > 0$, define $\Trunc{f} : \mathbb{F}_p \rightarrow \mathbb{F}_p$ by
\[
    \Trunc{f}(x) =
    \begin{cases}
        \left\lfloor x / 2^f \right\rfloor, & 0 \le x \le (p-1)/2, \\[1mm]
        p - \left\lfloor (p - x) / 2^f \right\rfloor, & (p-1)/2 < x \le p-1.
    \end{cases}
\]
\end{definition}

Thus, $\Trunc{f}$ removes the $f$ least significant bits of the absolute value while preserving the sign in the field representation. Since $p = 2^\ell - 1$, the sign is determined by the most significant bit, making truncation a signed right shift. Liu et al.~\cite{liu2024scalable} provide an efficient masked-opening protocol for this operation with only a one-bit gap. Their protocol targets ordinary Shamir sharing; we adapt it to packed Shamir sharing, where truncation must be applied independently to each slot.

\begin{mybox2}{\textsc{Alg. \ref{alg:trunc-degred}: PackedTruncationDegRed} (\(\Pi_{\mathrm{Trunc\text{-}DegRed}}^{\mathrm{Pack}}\))}
\textbf{Input:} $\pss{\vx{a}}{2d}$ \\
\textbf{Output:} $\pss{\vx{\hat{a}}}{d}$, where $\hat{a}_i = \Trunc{f}(a_i) + 2E$ for each $i\in[l]$ \\
\textbf{Correlated Randomness:} $\pss{\vx{r}}{2d}$, $\pss{\vx{r'}}{d}$, $\{\pss{\vx{\rho_i}}{d}\}_{i\in[l]}$, where $\rho_i[i]=1-r_{i,\msb}$
and all other slots are zero.
\begin{enumerate}[nosep, leftmargin=*, label=\arabic*.]
    \item $\pss{\vx{b}}{2d} \leftarrow \pss{\vx{a}}{2d} + 2^{\ell-2}$

    \item $\pss{\vx{c}}{2d} \leftarrow \pss{\vx{b}}{2d} + \pss{\vx{r}}{2d}$

    \item $\vx{c} = \{c_1,\dots,c_l\} \leftarrow \Reveal{\pss{\vx{c}}{2d}}$, and compute $\vx{c'} \leftarrow \Trunc{f}(\vx{c})$ and $c_{i,\msb}$ for every $i\in[l]$

    \item $\pss{\vx{c'}}{d} \leftarrow \Share{c'}{d}$

    \item $\pss{\vx{e_i}}{d} \leftarrow c_{i,\msb}\cdot \pss{\vx{\rho_i}}{d}$
for each $i\in[l]$

    \item $\pss{\vx{e}}{d} \leftarrow \sum_{i=1}^l \pss{\vx{e_i}}{d}$

    \item Output $\pss{\vx{\hat{a}}}{d} \leftarrow \pss{\vx{c'}}{d} - \pss{\vx{r'}}{d} + \pss{\vx{e}}{d} \cdot (2^{\ell-f}-1) - 2^{\ell-f-2}$
\end{enumerate}
\end{mybox2}
\noindent\begin{minipage}{\columnwidth}
\captionof{figure}{Packed Truncation and Degree Reduction Protocol\label{alg:trunc-degred}} 
\end{minipage} 

\begin{mybox2}{\textsc{Alg.~\ref{alg:degred}: PackedDegreeReduction} ($\Pi_{\mathsf{DegRed}}^{\mathsf{Pack}}$)}
\textbf{Input:} $\pssf{\vx{z}}{2d}{\scr{h}}$, where party $P_j$ holds $\scr{h}(\alpha_j)$ \\
\textbf{Output:} $\pssf{\vx{z}}{d}{\scr{q}}$
\begin{enumerate}[nosep, leftmargin=*, label=\arabic*.]

\item \textbf{Offline setup.}
For every $k,j\in[\size]$, precompute
\begin{center}
    $c_{k,j}
=
\sum_{i=1}^{l}
\lambda^{(d)}_{\beta_i}(\alpha_k)
\lambda^{(2d)}_{\alpha_j}(\beta_i).$
\end{center}

Each party $P_j$ samples a random degree-$(t-1)$ polynomial $\scr{r}_j(X)$
and precomputes $\mu(\alpha_k)\scr{r}_j(\alpha_k)$ for every receiver $P_k$, where
\[
\mu(X)=\prod_{i=1}^{l}(X-\beta_i).
\]

\item \textbf{Resharing.}
For each receiver $P_k$, party $P_j$ sends
\[
\msg{j}{k}
=
c_{k,j}\scr{h}(\alpha_j)
+
\mu(\alpha_k)\scr{r}_j(\alpha_k).
\]

\item \textbf{Share reconstruction.}
Each party $P_k$ sums all received messages:
\[
\scr{q}(\alpha_k)
=
\sum_{j=1}^{\size}\msg{j}{k}.
\]

\item \textbf{Output.}
Each $P_k$ outputs $\scr{q}(\alpha_k)$. The resulting shares form a fresh
degree-$d$ packed sharing $\pssf{\vx{z}}{d}{\scr{q}}$.

\end{enumerate}
\end{mybox2}
\noindent\begin{minipage}{\columnwidth}
\captionof{figure}{Packed Degree Reduction Protocol \label{alg:degred}}
\end{minipage}

Alg.~\ref{alg:trunc} applies this scalar truncation protocol independently
to each packed slot. Alg.~\ref{alg:trunc-degred} gives the combined
truncation and degree-reduction variant for degree-$2d$ inputs.

\noindent\textbf{Randomness Generation.} As described in Sec.~\ref{sec:sys-description:data}, \sysname{} eliminates preprocessing by having users supply the required randomness during data sharing. Once committee identities and slot assignments are known, each user $\us_i$ generates tuples of the form $(r,\Trunc{f}(r),1-r_{\msb})$. These are encoded as sparse vectors with a nonzero entry only in slot $q_i$, then secret-shared and sent to the assigned committee. The shares are aggregated locally, as with the data shares, ensuring that each user’s randomness aligns with its inputs and is used exclusively for its own data, without additional communication. This slot-aligned generation is important for privacy and is used in the security proof in App.~\ref{app:security}. Even if other slots are supplied by corrupted users, an honest user's slot remains masked by its own independently generated randomness. We formalize this in Lemma~\ref{lem:slot-trunc-privacy}.

\begin{mybox2}{\textsc{Alg.~\ref{alg:end-to-end}: \sysname{} End-to-End Training}}
\textbf{Input:} Users' input data; server-held model for
$\mathsf{Conv}\rightarrow\mathsf{PolyReLU}\rightarrow\mathsf{FC}$ \\
\textbf{Output:} Aggregate gradient revealed to the server

\begin{enumerate}[nosep, leftmargin=*, label=\arabic*.]

\item \textbf{Data sharing.}
Users secret-share their inputs and truncation randomness. Committee election outputs disjoint committees and slot assignments. Each committee locally forms packed input
sharings in the layout required by the first layer.

\item \textbf{Weight sharing.}
The server shares current model parameters to every committee. In
malicious setting, committees run weight-consistency check.

\item \textbf{Forward pass.}
Each committee independently evaluates the model:
\begin{itemize}
    \item packed Winograd convolution
    \item packed PolyReLU
    \item packed fully connected layer
\end{itemize}

\item \textbf{Backward pass.}
Each committee computes the output gradient and propagates it backward:
\begin{itemize}
    \item packed fully connected backward pass
    \item packed PolyReLU backward pass
    \item packed Winograd convolution backward pass
\end{itemize}
Each committee obtains a packed sharing of its local gradient aggregate.

\item \textbf{Validation and aggregation.}
In the malicious setting, committees validate local gradients and discard
invalid slots. Committees then aggregate local gradients into a global
aggregate sharing.

\item \textbf{Reveal and update.}
The global aggregate gradient is reconstructed to the server for model update.

\end{enumerate}
\end{mybox2}
\noindent\begin{minipage}{\columnwidth}
\captionof{figure}{End-to-end training protocol of \sysname{}.\label{alg:end-to-end}}

\end{minipage}

\section{\sysname{} Description: Training}
In this section, we provide an extended description of selected components of the training procedure of \sysname{}, specifically the backward pass and gradient validation during the aggregation step, building on Sec.~\ref{sec:sys-description}.

\subsection{Backward Pass}
The backward pass mirrors the structure of the forward pass. Given the secret-shared loss gradient at the output layer, each committee propagates gradients backward through the network using standard backpropagation equations. All computations use the same packing-aware primitives as in the forward pass: polynomial activation derivatives are evaluated over packed sharings, fully connected gradients use the diagonal and chunked matrix multiplication protocol, and convolutional gradients use the packed Winograd layout.

\noindent\textbf{PolyReLU Backward.}
Consider a PolyReLU layer
\[
    y = P(x) = \sum_{i=0}^{k} a_i x^i,
\]
and let $\pss{\Delta y}{d}$ denote the incoming gradient. During the forward pass, the committee caches either the monomials $\pss{x}{d}, \dots, \pss{x^k}{d}$ or the derivative value $\pss{P'(x)}{d}$. In the former case, the derivative is computed as
\[
    \pss{P'(x)}{d} = \sum_{i=1}^{k} i\cdot a_i \pss{x^{i-1}}{d} .
\]
The activation-gradient contribution is
\[
    \pss{\Delta x_{\mathsf{poly}}}{d}
    =
    \Cleanup{\pss{\Delta y}{d} \cdot \pss{P'(x)}{d}}
\]
The activation-range penalty also contributes to the gradient. Since
\[
    p(x) = \frac{\alpha}{\lambda^2}x^2,
    \qquad
    p'(x) = \frac{2\alpha}{\lambda^2}x,
\]
and the total loss includes $\beta \cdot \frac{1}{K} \sum_{x \in I} p(x)$, the penalty contribution for one activation input is
\[
    \pss{\Delta x_{\mathsf{pen}}}{d}
    =
    \beta \cdot \frac{2\alpha}{K \lambda^2} \pss{x}{d} .
\]
Thus, the total gradient with respect to the activation input is
\[
    \pss{\Delta x}{d}
    =
    \pss{\Delta x_{\mathsf{poly}}}{d}
    +
    \pss{\Delta x_{\mathsf{pen}}}{d}.
\]
\noindent\textbf{Fully Connected Backward.}
Consider a fully connected layer $C=WB$, where
$W\in\Fp^{\nrow\times\ncol}$, $B\in\Fp^{\ncol\times\ncol}$, and
$C\in\Fp^{\nrow\times\ncol}$. \update{For clarity, we first describe the entire fully connected backward pass under the same divisibility assumption $\nrow\mid\ncol$ as the basic forward construction. For arbitrary dimensions, all matrix multiplications below are evaluated blockwise using the arbitrary-dimension construction from Sec.~\ref{sec:fc}.} Given the incoming gradient
$\Delta C\in\Fp^{\nrow\times\ncol}$, the backward pass computes
\[
    \Delta B = W^T\Delta C,
    \qquad
    \Delta W = \Delta C B^T .
\]
\update{As in the forward-pass construction of Sec.~\ref{sec:fc}, the columns of $B$ and $C$ correspond to the $\ncol$ users processed in the batch. The input gradient $\Delta B$ is evaluated using the same packed diagonal matrix-multiplication layout as the forward pass, with $W^T$ as the weight matrix.} We now describe the weight gradient $\Delta W$.

\update{We keep $\Delta C$ and $B$ in lower-diagonal form and $\Delta W$ in upper-diagonal form.} For each $r\in[\ncol]$, define
\[
    \vx{\delta_r}:=\Ldiag{r}{\Delta C},
    \qquad
    \vx{b_r}:=\Ldiag{r}{B}.
\]
\update{For each $k\in[\nrow]$, the $k$th upper diagonal of $\Delta W$ is}
\[
    \Udiag{k}{\Delta W}
    =
    \sum_{r\in[\ncol]}
    \vx{\delta_r}
    \Had
    \rot{-k}\!\left(\vx{b_{r+k}}\right),     
\]
\update{For convenience, define}
\[
    \vx{\widetilde{b}}_{r,k}
    :=
    \rot{-k}\!\left(\vx{b_{r+k}}\right),
\]
\update{where the index $r+k$ is taken modulo $\ncol$. These rotated copies of $B$ are generated during the forward-pass degree-reduction and repacking step described in Sec.~\ref{sec:sys-description:training}. As in the forward pass, each length-$\ncol$ diagonal is split into
$q=\lceil \ncol/l\rceil$ chunks. For each pair $(r,k)$, we first form the full logical
rotated diagonal $\vx{\widetilde{b}}_{r,k}$ and then split both $\vx{\delta_r}$ and $\vx{\widetilde{b}}_{r,k}$ into chunks:}
\[
    \vx{\delta_r}
    =
    \bigl(\vx{\delta_r}^{(0)},\dots,\vx{\delta_r}^{(q-1)}\bigr),
    \qquad
    \vx{\widetilde{b}}_{r,k}
    =
    \bigl(\vx{\widetilde{b}}_{r,k}^{(0)},\dots,
    \vx{\widetilde{b}}_{r,k}^{(q-1)}\bigr).
\]
The rotation is applied before chunking; rotating each chunk independently
would not preserve the diagonal alignment. \update{For each $k$, $r$, and chunk $j$, define the packed product
\[
    \pss{\vx{T_{k,r}}^{(j)}}{2d}
    =
    \pss{\vx{\delta_r}^{(j)}}{d}
    \cdot
    \pss{\vx{\widetilde{b}}_{r,k}^{(j)}}{d}.
\]
Each slot of $\vx{T_{k,r}}^{(j)}$ contains one per-user contribution
$\Delta W_{a,b}^{(u)}:=\Delta C_{a,u}B_{b,u}$. The public indices $(k,r,j)$, together with the slot position, uniquely determine the parameter indices $(a,b)$ and user index $u$. For fixed $k$, $j$, and slot position, varying $r$ enumerates the users whose contributions must be summed for the corresponding entry of $\Udiag{k}{\Delta W}$. Thus, the packed products preserve the individual-user weight-gradient contributions before aggregation.}

\update{In the semi-honest setting, the committee aggregates these per-user contributions immediately. For each upper diagonal $k$ and chunk $j$, it computes
\[
    \pss{\vx{R_k}^{(j)}}{2d}
    =
    \sum_{r\in[\ncol]}
    \pss{\vx{T_{k,r}}^{(j)}}{2d}.
\]
After accumulating the products, the committee applies the required cleanup
step to obtain a degree-$d$ sharing. The chunks
\[
    \bigl(\vx{R_k}^{(0)},\dots,\vx{R_k}^{(q-1)}\bigr)
\]
represent the full upper diagonal $\Udiag{k}{\Delta W}$.}

\update{In the malicious setting, the committee instead postpones the summation over $r$ so that the per-user terms remain separable until gradient validation. It retains the packed products $\{\pss{\vx{T_{k,r}}^{(j)}}{2d}\}_{k,r,j}$ and converts their slots into the ordinary degree-$\thr$ Shamir sharings used by the validation protocol in Sec.~\ref{app:gradient-validation}. The committee groups these scalar terms by user and combines them with that user's other retained gradient coordinates to form the complete gradient vector $\g_\us$, consistent with Sec.~\ref{sec:mal:gradient-validation}. After evaluating $b_\us=\mathrm{Valid}(\g_\us)$, the committee performs the sums over users using only terms associated with users for whom $b_\us=1$, thereby obtaining the validated weight gradient $\Delta W$.}

\noindent\textbf{Convolutional Backward.}
Let $\pss{\Delta Y^{(o,\gamma)}_{a,b}}{d}$ denote the incoming gradient
with respect to the output tile coordinate $(a,b)$ for output channel $o$
and tile group $\gamma$. The backward pass follows the reverse of the forward Winograd
pipeline.

First, the parties propagate through the output transform. For each
$o,\gamma$ and $(u,v)\in\{0,\dots,\tile-1\}^2$,
\[
    \pss{\Delta M^{(o,\gamma)}_{u,v}}{d}
    =
    \sum_{a=0}^{m-1}\sum_{b=0}^{m-1}
    (A_{u,a}A_{v,b})\pss{\Delta Y^{(o,\gamma)}_{a,b}}{d}.
\]
This step is local and linear. Since the forward channel aggregation sets
$M^{(o,\gamma)}_{u,v}=\sum_c Z^{(o,c,\gamma)}_{u,v}$, the gradient is copied
to every input channel:
\[
    \pss{\Delta Z^{(o,c,\gamma)}_{u,v}}{d}
    =
    \pss{\Delta M^{(o,\gamma)}_{u,v}}{d}.
\]
Next, the parties propagate through the transformed-domain product
$Z^{(o,c,\gamma)}_{u,v}=U^{(o,c)}_{u,v}V^{(c,\gamma)}_{u,v}$. Thus,
\[
    \pss{\Delta V^{(c,\gamma)}_{u,v}}{2d}
    =
    \sum_{o=1}^{\Cout}
    \pss{\Delta Z^{(o,c,\gamma)}_{u,v}}{d}
    \cdot
    \pss{U^{(o,c)}_{u,v}}{d},
\]
and
\[
    \pss{\Delta U^{(o,c)}_{u,v}}{2d}
    =
    \sum_{\gamma}
    \pss{\Delta Z^{(o,c,\gamma)}_{u,v}}{d}
    \cdot
    \pss{V^{(c,\gamma)}_{u,v}}{d}.
\]
These are the only secret-secret multiplications in the convolutional
backward pass.

The activation gradient is obtained by applying the transpose of the
input transform:
\[
    \pss{\Delta d^{(c,\gamma)}_{i,j}}{2d}
    =
    \sum_{u=0}^{\tile-1}\sum_{v=0}^{\tile-1}
    (B_{i,u}B_{j,v})
    \pss{\Delta V^{(c,\gamma)}_{u,v}}{2d}.
\]
The weight-gradient branch may either keep the gradient in the transformed-kernel
form $\pss{\Delta U^{(o,c)}_{u,v}}{2d}$, or map it back to the spatial kernel
coordinates by the transpose of the kernel transform:
\[
    \pss{\Delta g^{(o,c)}_{p,q}}{2d}
    =
    \sum_{u=0}^{\tile-1}\sum_{v=0}^{\tile-1}
    (G_{u,p}G_{v,q})
    \pss{\Delta U^{(o,c)}_{u,v}}{2d}.
\]
All transform steps are local linear operations. After the product step and
subsequent linear transforms, the parties apply the required cleanups.



\subsection{Gradient Validation}
\label{app:gradient-validation}

\update{We model validity as a public arithmetic predicate $\mathrm{Valid}:\Fp^{\gradDim}\rightarrow\{0,1\}$ over complete gradient vectors, where $\gradDim$ denotes the number of scalar coordinates in each user's gradient. In our implementation and evaluation, $\mathrm{Valid}$ enforces a norm bound, although the protocol supports any predicate expressible as an arithmetic circuit over $\Fp$. Let $C_{\mathrm{Valid}}$ denote the circuit that computes this predicate. At the end of the backward pass, the committee holds a collection of degree-$d$ \PSS{} blocks whose slots jointly encode the users' gradient coordinates, as described in Sec.~\ref{sec:mal:gradient-validation}.}

\paragraph{Predicate Evaluation.}
\update{Using the public layout map, the committee unpacks the \PSS{} blocks into degree-$\thr$ ordinary Shamir sharings $\sss{g_{\us,j}}{\thr}$, where $\us$ identifies a user and $j\in[\gradDim]$ identifies a gradient coordinate. It then groups the scalar sharings by user to obtain
\[
    \sss{\g_\us}{\thr}
    =
    \bigl(\sss{g_{\us,1}}{\thr},\dots,
    \sss{g_{\us,\gradDim}}{\thr}\bigr).
\]
The committee evaluates the predicate independently for each user:
\[
    \sss{b_\us}{\thr}
    =
    C_{\mathrm{Valid}}\!\left(\sss{\g_\us}{\thr}\right),
\]
where $b_\us=1$ indicates that user $\us$'s complete gradient satisfies the validity predicate.}

\update{Linear gates in $C_{\mathrm{Valid}}$ are evaluated locally. Each multiplication gate is evaluated using the DIK-based protocol~\cite{damgaard2010perfectly}, and the committee records the tuple
\[
    \bigl(\sss{x^{(r)}}{\thr},\,\sss{y^{(r)}}{\thr},\,
    \sss{z^{(r)}}{\thr}\bigr),
\]
where $\sss{z^{(r)}}{\thr}$ is the claimed product of $\sss{x^{(r)}}{\thr}$ and $\sss{y^{(r)}}{\thr}$. Let
\[
    \mathcal T
    =
    \left\{
    \bigl(\sss{x^{(r)}}{\thr},\,\sss{y^{(r)}}{\thr},\,
    \sss{z^{(r)}}{\thr}\bigr)
    \right\}_{r=1}^{M}
\]
denote the set of multiplication tuples generated across all users during predicate evaluation.}

\paragraph{\update{Validation Consistency and Local Aggregation.}}
\update{Before revealing the validity bits, the committee runs a batched multiplication-consistency check on $\mathcal T$ using the protocol of~\cite{goyal2020malicious}. If the check fails, the protocol aborts. Otherwise, the committee reconstructs $b_\us$ for every user $\us$. A user's gradient $\g_\us$ is accepted if and only if $b_\us=1$, and only accepted gradients are included in the local aggregate.}

\paragraph{\update{Cross-Committee Aggregation.}}

\update{After validation, each committee locally sums its accepted degree-$\thr$
gradient sharings to obtain a local aggregate. Before cross-committee
aggregation, source committee $\Co_k$ raises the degree of its aggregate to
$2\thr$; let $\sss{\g_k}{2\thr}$ denote the resulting sharing. Each member of
$\Co_k$ sends its share of $\sss{\g_k}{2\thr}$ to the correspondingly indexed
member of the designated primary committee.}

\update{Crucially, the primary committee does not immediately sum the share
vectors received from different source committees. For each source committee
$\Co_k$, it first verifies that the received vector is consistent with a
degree-$2\thr$ Shamir sharing and incorporates the contribution only if this
check passes. This per-source check prevents deviations from different source
committees from accumulating at different share positions before consistency
is verified.}

\update{To see why this suffices, consider one source committee $\Co_k$.
Let $\mathbf{f}_k$ denote its correct degree-$2\thr$ share vector and let
$\mathbf{y}_k$ denote the vector presented to the primary committee's
consistency check. At most $\thr$ corrupted members of $\Co_k$ can alter the
shares they send, while at most $\thr$ corrupted members of the primary
committee can alter the values supplied at their own positions during the
check. Hence
\[
    \operatorname{wt}(\mathbf{y}_k-\mathbf{f}_k)\le 2\thr.
\]
In the malicious setting, we additionally require $\thr<\pf$. Since
$d=\thr+\pf-1$ and $\size=2d+1=2\thr+2\pf-1$, the Reed--Solomon code induced by
degree-$2\thr$ Shamir sharing has minimum distance
\[
    \delta = \size-2\thr = 2\pf-1 > 2\thr.
\]
Therefore, if $\mathbf{y}_k$ passes the degree-$2\thr$ consistency check,
it cannot be a distinct valid codeword within distance $2\thr$ of
$\mathbf{f}_k$. Thus, except with the soundness error of the
consistency check, every accepted source contribution equals the correct
source sharing; an inconsistent contribution causes the protocol to abort.}

\update{After all source-committee contributions pass their individual
consistency checks, each primary-committee member locally sums the
corresponding accepted shares to obtain a degree-$2\thr$ sharing
$\sss{\g_{\mathsf{global}}}{2\thr}$ of the global accepted gradient.
The primary committee then sends these aggregate shares to the server.
At this point, only corrupted members of the primary committee can modify
the outgoing aggregate, affecting at most $\thr$ share positions. Before
reconstructing the aggregate, the server therefore performs a final degree-$2\thr$
consistency check on the received vector and aborts if it fails. If the
check passes, the server interpolates the degree-$2\thr$ sharing and
recovers $\g_{\mathsf{global}}$.}

\section{Proof of Committee Safety Lemma}
\label{app:lem:committee-safety}
In this section, we provide a proof of Lemma~\ref{lem:committee-safety}.
\begin{proof}
Let $\totalC$ be the total number of users, of which $\eta \totalC$ are corrupted. We model the selection of $\num$ disjoint committees of size $\size$ as sampling without replacement from the $\totalC$ users.

We analyze the process in two steps. First, fix a partition of $\totalC$ users into $\num$ disjoint committees of size $\size$, and a remaining set of size $\totalC - \num\size$. Since committees are sampled uniformly at random,
every such partition is equally likely.

Next, select the $\eta\totalC$ corrupted users uniformly at random from all $\totalC$ users. The total number of possible choices is $\binom{N}{\eta N}$. For the event that all committees are safe, let $x_i$ denote the number of corrupted users assigned to committee $i$, for $i \in [\num]$. A configuration is safe if $0 \le x_i \le t$ for all $i$, and $\sum_{i=1}^{\num} x_i \le \eta\totalC$. 

Fix any tuple $(x_1,\dots,x_\num)$ satisfying these constraints. We count the number of ways to assign corrupted users. Within committee $i$, there are $\binom{\size}{x_i}$ ways to choose $x_i$ corrupted users. Since the committees are disjoint, choices are independent across committees, yielding
\[
\prod_{i=1}^{\num} \binom{\size}{x_i}
\]
ways to assign corrupted users within the committees. Then, assigning the remaining $\eta \totalC - \sum_{i=1}^{\num} x_i$ corrupted users to the remaining $\totalC - \num\size$ users can be done in
\[
\binom{\totalC - \num\size}{\eta\totalC - \sum_{i=1}^{\num} x_i}
\]
ways. Summing over all valid tuples $(x_1,\dots,x_\num)$, the total number of favorable assignments is
\[
\sum_{\substack{0 \le x_1,\dots,x_\num \le t \\ \sum_{i=1}^{\num} x_i \le \eta \totalC}}
\left(
\prod_{i=1}^{\num} \binom{\size}{x_i}
\right)
\binom{\totalC - \num\size}{\eta \totalC - \sum_{i=1}^{\num} x_i}.
\]
Dividing by the total number of assignments $\binom{\totalC}{\eta\totalC}$ yields the claimed expression.
\end{proof}
\section{Ideal Functionalities and Security Proofs}\label{app:security}

\subsection{Truncation}

\begin{mybox2}{\textsc{Ideal Functionality: PackedTruncation} ($\mathcal{F}_{\mathrm{Trunc}}^{\mathrm{Pack}}$)}
\label{func:packed-trunc}

    \textbf{Input:}  The functionality receives $\pss{\vx{a}}{d}$
    
    \textbf{Output:} Compute the following:
    
    \begin{enumerate}[nosep, leftmargin=*, label=\arabic*.]
        \item Reconstruct $\vx{a} = (a_1,\dots,a_l)$.
    
        \item Compute $\hat{a}_i=\Trunc{f}(a_i)$ for each $i\in[l]$.
    
        \item Construct a fresh degree-$d$ packed sharing $\pss{\vx{\hat{a}}}{d}$ and send the corresponding shares to the parties.
    \end{enumerate}
\end{mybox2}
\noindent\begin{minipage}{\columnwidth}
\captionof{figure}{Ideal Functionality: Packed Truncation.} 
\end{minipage}

\begin{mybox2}{\textsc{Ideal Functionality: Trunc-DegRed} ($\mathcal{F}_{\mathrm{Trunc-DegReg}}^{\mathrm{Pack}}$)}
\label{func:packed-trunc-degred}

    \textbf{Input:}  The functionality receives $\pss{\vx{a}}{2d}$
    
    \textbf{Output:} Compute the following:
    
    \begin{enumerate}[nosep, leftmargin=*, label=\arabic*.]
        \item Reconstruct $\vx{a} = (a_1,\dots,a_l)$.
    
        \item Compute $\hat{a}_i=\Trunc{f}(a_i)$ for each $i\in[l]$.
    
        \item Construct a fresh degree-$d$ packed sharing $\pss{\vx{\hat{a}}}{d}$ and send the corresponding shares to the parties.
    \end{enumerate}
\end{mybox2}
\noindent\begin{minipage}{\columnwidth}
\captionof{figure}{Ideal Functionality: Truncation and Degree Reduction.} 
\end{minipage} 
\begin{mybox2}{Ideal Functionality of \sysname{} (Fig.~\thefigure)}
\label{fig:IF-appendix}
\footnotesize
\textbf{Parameters.} Number of committees $\num$, Overall corruption rate $\eta$, Size of each committee $\size$, Per-committee allowable corrupt parties $t$, Packing factor $l$, Security parameter $\kappa$; \textbf{Parties.} $\totalC$ users $\us_i, i \in [N]$ and a server $\Ser$ 
\\
\textbf{Functionality.}
\squishlist
    \item Receive $\vx{x}_i$ from honest user $\us_i$. If $\textsf{Pr[safe]} < 1-\frac{1}{2^\kappa}$ (Lemma~\ref{lem:committee-safety}, computed using $\num,\eta,t,\totalC,$ and $\size$), send $\mathsf{ABORT}$.
    
    \item Partition users randomly into $\num$ committees $\{\Co_j\}_{j \in [\num]}$.

    \item Receive $\num$ copies of the model weights $\{W_j\}_{j \in [\num]}$ from the server; if they are not identical, output $\mathsf{ABORT}$.

    \item Receive from the server validation predicates and parameters $\{(\textsf{Valid}_j,\rho_j)\}_{j \in [\num]}$, one for each committee.

    \item Receive a set $L \subseteq [N]$ of honest users to be perturbed, along with perturbations $\{\delta_i\}_{i \in L}$ from the adversary.

    \item Receive gradients $\{\g'_i\}$ from malicious users.

    \item For each user $\us_i$, compute gradient $\g_i$ on their data $x_i$ using the weights supplied by the server. If $\us_i \in L$ (honest but perturbed), set
    $\hat{\g}_i=\g_i+\delta_i$\footnote{The adversary can only introduce additive errors~\cite{cryptoeprint:2015/154}.}
    and evaluate $\textsf{Valid}_j(\hat{\g}_i)$ for $\us_i \in \Co_j$.
    If $\us_i$ is corrupted, set $\hat{\g}_i=\g'_i$. Otherwise, set $\hat{\g}_i=\g_i$.

    \item Output to the server the aggregate
        $S=\sum_{i=1}^N
        \textsf{Valid}_{j(i)}(\hat{\g}_i,\rho_{j(i)})
        \cdot \hat{\g}_i,$
    where $j(i)$ denotes the committee containing $\us_i$, along with per-committee counts of invalid gradients:
$\left\{
        \sum_{\us_i \in \Co_j}
        \left(1-\textsf{Valid}_j(\hat{\g}_i,\rho_{j(i)})\right)
        \right\}_{j \in [\num]}.$

    \item Compute the gradients for one round, $\g_i$ for $i\in[N]$, and apply the specified perturbations:
    $\hat{\g}_i=\g_i+\delta_i$ for $i\in L$.
\squishend
\end{mybox2}

\begin{lemma}[Slot privacy of packed truncation randomness]
\label{lem:slot-trunc-privacy}
    Consider a packed truncation instance with degree $d=t+l-1$, where each of $l$ slots is supplied by one user. Suppose the committee contains at most $t$ corrupted parties. Fix an honest user in slot $\beta_q$ with secret $a_q$, while other $l-1$ slot owners may be corrupted. If the honest user samples its truncation mask $r_q$ independently, then opening the masked value $c_q = a_q + r_q$ reveals no information about $a_q$ to the adversary.
\end{lemma}

\begin{proof}
The adversary may know the secrets and randomness of all corrupted users occupying the other
$l-1$ slots. It may also see up to $t$ committee shares of the packed sharings containing the honest slot's mask $r_q$. 

By the privacy of \PSS{} of degree $d=t+l-1$, these views cannot determine the value encoded at $\beta_q$. A degree-$d$ \PSS{} has one remaining degree of freedom after fixing $t$ evaluation shares and $l-1$ slot values. Equivalently, the honest slot's mask, and therefore the secret, remain information-theoretically hidden.

It remains to consider the value opened during truncation. In slot $\beta_q$, the
opened value is $c_q = a_q + r_q$. Since $r_q$ is sampled independently by the honest user and remains hidden
from the adversary, this is exactly the masked-opening view of the underlying scalar truncation protocol. Thus, opening $c_q$ reveals no additional information about $a_q$. The corrupted slots do not affect this argument, because their masks are independent of $r_q$ and are applied only to their own slots.
\end{proof}

\begin{theorem} $\Pi_{\mathrm{Trunc}}^{\mathrm{Pack}}$ securely realizes $\mathcal{F}_{\mathsf{Trunc}}^{\mathsf{Pack}}$ in the presence of a semi-honest adversary controlling up to $t$ corrupted parties. 
\end{theorem}

\begin{proof}
    \textbf{Correctness.} Correctness follows slot-wise from the scalar
    truncation protocol of Liu et al.~\cite{liu2024scalable}. For each slot
    $i\in[l]$, the protocol sets $b_i=a_i+2^{\ell-2}$, opens
    $c_i=b_i+r_i$, and applies the scalar correction using
    $r_i'=\Trunc{f}(r_i)$ and $c_{i,\msb}(1-r_{i,\msb})$. The only
    packed-specific issue is that this correction is slot-dependent. This is
    handled by the sparse correction sharing for slot $i$, so
    $c_{i,\msb}(1-r_{i,\msb})$ is applied only to that slot. Thus each output
    slot equals the scalar truncation output for $a_i$, and the parties obtain
    a fresh degree-$d$ packed sharing of the vector specified by
    $\mathcal{F}_{\mathsf{Trunc}}^{\mathsf{Pack}}$.
    
    \textbf{Privacy.} The adversary sees only its shares, the opened value
    $\vx{c}$, and its output shares. For each honest user's slot, the opened
    value is the same masked-opening view as in the scalar truncation
    protocol, using that slot's independently generated mask. By
    Lemma~\ref{lem:slot-trunc-privacy}, this remains true even if the other
    $l-1$ slots are corrupted. Hence the packed protocol leaks no more than
    applying the scalar truncation protocol independently to each slot.
\end{proof}

\begin{theorem} $\Pi_{\mathrm{Trunc-DegRed}}^{\mathrm{Pack}}$ securely realizes $\mathcal{F}_{\mathsf{Trunc-DegRed}}^{\mathsf{Pack}}$ in the presence of a semi-honest adversary controlling up to $t$ corrupted parties. 
\end{theorem}

\begin{proof}
    Both correctness and privacy follow by the same argument as for
$\Pi_{\mathrm{Trunc}}^{\mathrm{Pack}}$, except that the masked value and mask are
degree-$2d$ before opening, and the public truncated value is re-shared at
degree $d$.
\end{proof}

\subsection{Secure Exponentiation}

\begin{theorem} $\Pi_{\mathrm{Exp}}^{\mathrm{Pack}}$ securely realizes $\mathcal{F}_{\mathsf{Exp}}^{\mathsf{Pack}}$ in the presence of a semi-honest adversary controlling up to $t$ corrupted parties. 
\end{theorem}

\begin{proof}
    \textbf{Correctness.} After $\PssToAdd$, each slot satisfies $x_j=\sum_{i=1}^{n}u_{i,j}$. The repartitioning step preserves this value, since each $u_{i,j}$ is split as $u_{i,j}=\sum_{r=1}^{q}u_{i,j}^{(r)}$. Hence
    \begin{center}
     $x_j
    =
    \sum_{i=1}^{n}u_{i,j}
    =
    \sum_{r=1}^{q}z_j^{(r)}.$
    \end{center}
    
    By the multinomial theorem,
    \begin{center}
     $
    x_j^k
    =
    \left(\sum_{r=1}^{q}z_j^{(r)}\right)^k
    =
    \sum_{\substack{a_1,\dots,a_q\ge 0\\ a_1+\cdots+a_q=k}}
    \binom{k}{a_1,\dots,a_q}
    \prod_{r=1}^{q}(z_j^{(r)})^{a_r}.$
    \end{center}
    The protocol computes exactly these terms and sums them, so each output slot is $x_j^k$, and the output is a degree-$d$ packed sharing of $\vx{x^k}$.
        
    \textbf{Privacy.} The conversion from \PSS{} to additive sharing is a local linear transformation of the original shares, and by the privacy of degree-$d$ packed Shamir sharing, any set of at most $t$ corrupted parties learns no information about $\vx{x}$ from these values. During repartitioning, each additive contribution is randomly split into additive parts. Since at least one party is assumed to be honest, the adversary misses at least one additive share in each repartitioned sum, so its observed shares are uniformly random subject to the final value and reveal no information. The reshared powers are fresh degree-$d$ packed sharings, and the remaining steps consist only of packed multiplications, degree reductions, public scalar multiplications, and local additions. These operations reveal no intermediate value in the clear. Therefore, the adversary learns nothing beyond its prescribed output shares.
\end{proof}

\begin{mybox2}{\textsc{Ideal Functionality: SecureExponentiation} ($\mathcal{F}_{\mathsf{Exp}}^{\mathsf{Pack}}$)}
\label{func:exp}
    \textbf{Input:} The functionality receives $\pss{\vx{a}}{d}$ and exponent $k$
    
    \textbf{Output:} Compute the following:
    \begin{enumerate}[nosep, leftmargin=*, label=\arabic*.]
        \item Reconstruct $\vx{a} = (a_1,\dots,a_l)$.
        \item Compute ${a}^k_i$ for each $i\in[l]$.
        \item Construct a fresh degree-$d$ packed sharing $\pss{\vx{a^k}}{d}$ and send the corresponding shares to the parties.
    \end{enumerate}
\end{mybox2}
\noindent\begin{minipage}{\columnwidth}
\captionof{figure}{Ideal Functionality: Secure Exponentiation.} 
\end{minipage} 

\subsection{Degree Reduction and Repacking}
\label{app:if:degred}

\begin{mybox2}{\textsc{Ideal Functionality: DegreeReduction} ($\mathcal{F}_{\mathsf{DegRed}}^{\mathsf{Pack}}$)}
\label{func:degred}
    \textbf{Input:} The functionality receives $\pss{\vx{a}}{2d}$
    
    \textbf{Output:} Compute the following: 
    \begin{enumerate}[nosep, leftmargin=*, label=\arabic*.]
        \item Reconstruct $\vx{a} = (a_1,\dots,a_l)$.
        \item Construct a fresh degree-$d$ packed sharing $\pss{\vx{a}}{d}$ and send the corresponding shares to the parties.
    \end{enumerate}
\end{mybox2}
\noindent\begin{minipage}{\columnwidth}
\captionof{figure}{Ideal Functionality: Degree Reduction.} 
\end{minipage} 

\begin{mybox2}{\textsc{Ideal Functionality: DegRed-Repack} ($\mathcal{F}_{\mathsf{DegRed-Repack}}^{\mathsf{Pack}}$)}
\label{func:degred-repack}
    \textbf{Input:} The functionality receives $m$ packed sharings $\{\pss{\vx{z}^{(a)}}{2d}\}_{a=1}^{m}$ and a public routing function $\rho$, where $\rho(i)=(j,k)$ specifies that output slot $i$ should contain the value from slot $k$ of input sharing $j$.
    
    \textbf{Output:} Compute the following:
    \begin{enumerate}[nosep, leftmargin=*, label=\arabic*.]
        \item Reconstruct each packed vector $\vx{z}^{(a)}=(z^{(a)}_1,\dots,z^{(a)}_l)$ for $a\in[m]$
        \item For each output slot $i\in[l]$, compute $\rho(i)=(j,k)$ and set $z'_i \leftarrow z^{(j)}_{k}$.
        \item Let $\vx{z}'=(z'_1,\dots,z'_l)$. Construct a fresh degree-$d$ packed sharing $\pss{\vx{z}'}{d}$ and send the corresponding shares to the parties.
    \end{enumerate}
\end{mybox2}
\noindent\begin{minipage}{\columnwidth}
\captionof{figure}{Ideal Functionality: Degree Repacking.} 
\end{minipage}

\begin{theorem} $\Pi_{\mathrm{DegRed}}^{\mathrm{Pack}}$ securely realizes $\mathcal{F}_{\mathsf{DegRed}}^{\mathsf{Pack}}$ in the presence of a semi-honest adversary controlling up to $t$ corrupted parties. 
\end{theorem}

\begin{proof}
    \textbf{Correctness.} Let the input be a degree-$2d$ packed sharing $\pssf{\vx{z}}{2d}{\scr{h}}$. The protocol constructs a degree-$d$ polynomial $s$ satisfying $s(\beta_i)=h(\beta_i)=z_i$ for all $i\in[l]$. Each party receives a share of $q(X)=s(X)+\mu(X)r(X)$. Since $\mu(\beta_i)=0$ for every $i\in[l]$, the random zero encoding does not change the slot values $q(\beta_i)=s(\beta_i)=z_i$. Therefore, the output is a degree-$d$ packed sharing of the same secret vector $\vx{z}$.

    \textbf{Privacy.} The protocol reveals no secret value. Each message $m_{j\to k}$ is a linear function of the sender's share masked by an independently sampled zero encoding, assumed to be uniformly random. The output polynomial is randomized by $\mu(X)r(X)$, whose values vanish on the slots but randomize the remaining $t$ degrees of freedom. Hence the resulting sharing is distributed as a fresh degree-$d$ packed sharing of the target vector. By the privacy of packed Shamir sharing, any adversary corrupting at most $t$ parties learns nothing beyond its output shares.
\end{proof}

\begin{theorem} $\Pi_{\mathrm{DegRed-Repack}}^{\mathrm{Pack}}$ securely realizes $\mathcal{F}_{\mathsf{DegRed-Repack}}^{\mathsf{Pack}}$ in the presence of a semi-honest adversary controlling up to $t$ corrupted parties. 
\end{theorem}

\begin{proof}
    \textbf{Correctness.} The proof is identical to the correctness proof of $\Pi_{\mathrm{DegRed}}^{\mathrm{Pack}}$, except that the target slot vector is first replaced by the public repacking map. Specifically, the repacked vector $\vx{z}'$ is defined by $z'_i=\sum_{a=1}^{m}\sum_{u=1}^{l}M^{(a)}_{i,u}h^{(a)}(\beta_u)$ for all $i\in[l]$. The protocol then performs the same degree-reduction procedure as before, but encodes $\vx{z}'$ instead of the original slot vector. Since the added zero encoding $\mu(X)r(X)$ vanishes at every slot point $\beta_i$, it does not change the repacked slot values. Therefore, the output is a degree-$d$ packed sharing of $\vx{z}'$.

    \textbf{Privacy.} Same as $\Pi_{\mathrm{DegRed}}^{\mathrm{Pack}}$.
\end{proof}

\subsection{Proof of Theorem \ref{thm:security}}

\begin{proof}
We prove security of \sysname{} by simulation. 

Let $\adv$ be a static malicious adversary corrupting the server and at most an $\eta$ fraction of users. We construct a simulator $\Sim$ that interacts with $\adv$ and the ideal functionality $\mathcal F$, and generates a view indistinguishable from $\adv$'s view in the real protocol.

Let $\mathsf{Bad}$ be the event that some committee has more than $t$ corrupted parties. The ideal functionality aborts on $\mathsf{Bad}$. Therefore, in the remainder of the proof, we condition on $\neg\mathsf{Bad}$, so every committee contains at most $t$ corrupted parties. 

The simulator internally runs $\adv$ and emulates all honest users and committee parties. It samples the committee assignment using the same public randomness as in the real protocol and gives the assignment to $\adv$. If the assignment causes $\mathsf{Bad}$, then $\Sim$ $\mathsf{ABORT}$ as the ideal functionality and outputs the corresponding abort view.

For corrupted users, $\Sim$ obtains their inputs and adversarially chosen gradients from $\adv$ and forwards them to $\mathcal F$ as corrupted-user inputs. Since the server is corrupted, $\Sim$ also receives the model weights, validation predicates, and parameters chosen by $\adv$ and forwards them to $\mathcal F$. If the server sends inconsistent model weights to different committees, $\Sim$ lets the weight-consistency check reject and $\mathsf{ABORT}$ as in $\mathcal F$.

For honest user, $\Sim$ does not know their private inputs. In the real protocol, honest users' input and all intermediate values are represented by packed Shamir sharings. Whenever $\adv$ should see shares of an honest input or intermediate value, $\Sim$ samples random field elements as the corrupted parties' shares. Conditioned on $\neg\mathsf{Bad}$, $\adv$ sees at most $t$ shares of each sharing. Therefore, by the privacy of packed Shamir sharing, these simulated shares are identically distributed to the real shares and reveal no information about the encoded values. 

For every local linear operation and secret-secret multiplication, $\Sim$ applies the same operation to the simulated corrupted shares. Therefore, they have the same distribution as in the real protocol. For interactive computations, including $\Pi_{\mathrm{DegRed}}^{\mathrm{Pack}}$, $\Pi_{\mathrm{DegRed-Repack}}^{\mathrm{Pack}}$, and $\Pi_{\mathrm{Trunc-DegRed}}^{\mathrm{Pack}}$, $\Sim$ invokes the simulator for the corresponding protocol.

At the end of the protocol, $\Sim$ obtains from $\mathcal F$ the accepted aggregate gradient and the per-committee invalid-gradient counts, and gives these values to $\adv$.

We now argue that the simulated view is indistinguishable from the real view. The committee assignment is sampled exactly as in the real protocol. Conditioned on $\neg\mathsf{Bad}$, $\adv$ sees at most $t$ shares of every honest packed Shamir sharing, so the simulated shares are identically distributed to real shares by packed Shamir privacy. Local operations preserve the distribution. Interactive protocols are replaced by their ideal functionalities with an indistinguishable view by the security of the corresponding protocol proved above. 

For correctness, till gradient computation we can assume that the gradients are now perturbed by some adversary chosen $\{\delta\}$ error.  Malformed gradients check is done in the malicious setting and the correctness follows from that~\cite{goyal2021atlas,goyal2020guaranteed,goyal2020malicious} -- only gradients with $\textsf{Valid}(\g+\delta)=1$ pass as allowed by our ideal functionality. Similarly, server's malicious behavior is caught by the weight-consistency check. The final gradient aggregation step is simple additions which is caught by a consistency check. The only values revealed are the accepted aggregate and invalid-gradient counts, which are exactly the outputs of the ideal functionality. 

Thus, conditioned on $\neg\mathsf{Bad}$, the simulated and real views are indistinguishable. Since $\Pr[\mathsf{Bad}]\le 2^{-\kappa}$, the real and ideal executions differ only with negligible probability. Therefore, \sysname{} securely realizes the ideal functionality $\mathcal F$.

\end{proof}

\subsection{\update{Dropout Tolerance}}
\label{app:dropout}

\update{Let $\dropout$ denote the desired dropout tolerance. For a fixed corruption threshold $\thr$ and packing factor $\pf$, let $d=\thr+\pf-1$ and $\size=2d+1$. In the semi-honest setting, tolerating up to $\dropout$ unavailable members requires at least $2d+1+\dropout$ committee members. In the malicious-with-abort setting, we instead require $2d+1+2\dropout,$ where the additional $\dropout$ members are used to handle malicious cross-committee transfers under different dropout patterns, as described below. In concrete parameter selection, increasing the committee size may change the value of $\thr$ required to satisfy the target committee-safety bound. We therefore choose $\thr$ and the committee size jointly, using the corresponding size requirement above while maintaining the desired bound on the probability that a committee contains more than $\thr$ corrupted members. All input shares, model-weight shares, and user-supplied randomness are distributed to every member of the provisioned committee. Local linear operations and share-wise multiplications require no modification; only interactive subprotocols require additional mechanisms.}

\update{In truncation and degree-reduction Algs.~\ref{alg:trunc} and~\ref{alg:trunc-degred}, the masked value opened in Step~3 is shared using a degree-$2d$ polynomial. Hence, any $2d+1$ available shares uniquely determine the opened value, and each party can interpolate from any $2d+1$ shares it receives. Increasing the committee size to $2d+1+\dropout$ therefore tolerates up to $\dropout$ erasures.}

\update{Degree reduction and repacking in Algs.~\ref{alg:degred} and~\ref{alg:degred-repack} require an additional mechanism because every receiver must apply the degree-reduction map using the same set of $2d+1$ shares. A public ordering is fixed over all members of each committee. The protocol initially selects the
first $\size$ parties in this ordering as the sender set. Since all user-to-user communication is relayed by the server, the server can wait for complete message bundles from the current sender set before relaying. After a timeout, each missing sender is replaced, in order, by the next available member on the list. Because changing the sender set changes the interpolation coefficients used by the degree-reduction map, all parties in
the updated sender set recompute and resend their messages. A member once
identified as unavailable remains excluded from subsequent sender sets.
Since at most $\dropout$ parties are unavailable, at most $\dropout$ sender
replacements are required.}

\update{Cross-committee operations, such as weight consistency and cross-committee gradient aggregation, require additional care because corresponding positions in different committees may experience different dropout patterns. Thus, the number of aligned live positions may fall below the required threshold even when each committee individually has enough active parties. In the semi-honest setting, provisioning $\size+\dropout$ members suffices. For each source committee $\Co_k$ and each active party $P_{1,j}$ in the primary committee $\Co_1$, party $P_{k,j}$ directly sends $f_k(\alpha_j)$ to $P_{1,j}$. If $P_{k,j}$ is unavailable, the current set of live senders in $\Co_k$ privately restores only the missing evaluation $f_k(\alpha_j)$ using fresh user-provided correlated randomness. Thus, every active party in the primary committee obtains its share, allowing cross-committee values to be aggregated locally at corresponding evaluation points. Since public ordering and dropout budget are shared across the committee execution, parties already identified as unavailable remain excluded, and at most $\dropout$ such restorations are required per source committee.}

\update{In the malicious-with-abort setting, we instead provision $\size+2\dropout$ members and use the same direct-transfer and share-restoration procedure. In gradient aggregation, at most $\thr$ directly transferred source positions, $\dropout$ restored positions, and $\thr$ primary-committee positions can be adversarially affected, for a total of $2\thr+\dropout$ positions. The primary committee has at least $\size+\dropout$ active parties. Therefore, the corresponding degree-$2\thr$ Reed--Solomon code has minimum distance at least $(\size+\dropout)-2\thr=2\pf-1+\dropout$. Under the requirement $\thr<\pf$, we have $2\thr+\dropout<2\pf-1+\dropout$, so an incorrect contribution cannot form a distinct valid degree-$2\thr$ codeword. A source gradient contribution is therefore either rejected or accepted as the correct aggregate, except with probability bounded by the soundness error of the consistency check. For weight consistency, the primary committee first verifies that the transferred evaluations form a valid degree-$(d+\thr)$ sharing and then applies the existing weight equality check. Thus, even if a malicious transfer or restoration produces a different valid sharing, the sharing is harmless if it encodes the same packed weights and is rejected otherwise, except with probability bounded by the soundness error of these checks.}

\update{Corrupted committee members may withhold messages or send malformed ones, while a corrupted server may selectively delay or suppress messages or relay inconsistent shares to different receivers. These deviations may cause a party to be treated as unavailable, produce inconsistent sharings, perturb subsequent computation, or force an abort, all of which are permitted by our model. Our dropout mechanism does not correct malformed shares. Any resulting share inconsistency is detected by subsequent consistency checks, while a deviation that still produces a well-formed gradient remains subject to the validation predicate described in Sec.~\ref{sec:mal:gradient-validation}.}

\begin{table}[!t]
\centering
\caption{SH provisioning overhead at $\totalC=100\text{K}$.}
\label{tab:sh-dropout}
\scriptsize
\setlength{\tabcolsep}{1.2pt}
\renewcommand{\arraystretch}{0.92}
\resizebox{\columnwidth}{!}{%
\begin{tabular}{@{}l r r r r r rr rr@{}}
\toprule
Model
& $D$ & $c$ & $l$ & $t$ & $n'$
& $T$ (ms) & $C$ (MB)
& $\Delta T$ (\%) & $\Delta C$ (\%) \\
\midrule

BCW
& 0  & 1860 & 14 & 13 & 53 & 3.28 & 1.31 & --   & -- \\
& 5  & 2000 & 10 & 13 & 50 & 3.86 & 1.67 & 17.7 & 27.1 \\
& 10 & 1790 & 10 & 13 & 55 & 3.88 & 1.67 & 18.4 & 27.1 \\
& 20 & 1290 & 14 & 15 & 77 & 4.91 & 2.12 & 49.9 & 61.4 \\
\midrule

Credit
& 0  & 1640 & 16 & 13 & 57 & 3.19 & 0.58 & --   & -- \\
& 5  & 1890 & 11 & 13 & 52 & 3.33 & 0.63 & 4.4  & 7.6 \\
& 10 & 1640 & 11 & 13 & 57 & 3.35 & 0.72 & 5.2  & 23.7 \\
& 20 & 1390 & 11 & 14 & 69 & 4.48 & 0.89 & 40.5 & 51.3 \\
\midrule

MNIST
& 0  & 2300 & 10 & 12 & 43 & 283.49 & 28.79 & --   & -- \\
& 5  & 2000 & 10 & 13 & 50 & 334.38 & 34.17 & 18.0 & 18.7 \\
& 10 & 1790 & 10 & 13 & 55 & 374.13 & 38.28 & 32.0 & 33.0 \\
& 20 & 1390 & 12 & 14 & 71 & 487.72 & 45.25 & 72.0 & 57.2 \\
\midrule

FMNIST
& 0  & 2300 & 10 & 12 & 43 & 222.00 & 26.67 & --   & -- \\
& 5  & 2000 & 10 & 13 & 50 & 261.68 & 31.68 & 17.9 & 18.8 \\
& 10 & 1790 & 10 & 13 & 55 & 292.77 & 35.47 & 31.9 & 33.0 \\
& 20 & 1390 & 12 & 14 & 71 & 386.84 & 41.95 & 74.2 & 57.3 \\
\midrule

SVHN
& 0  & 2300 & 10 & 12 & 43 & 300.57 & 35.64 & --   & -- \\
& 5  & 2000 & 10 & 13 & 50 & 354.36 & 42.33 & 17.9 & 18.8 \\
& 10 & 1730 & 11 & 13 & 57 & 403.01 & 46.22 & 34.1 & 29.7 \\
& 20 & 1430 & 11 & 14 & 69 & 512.47 & 58.09 & 70.5 & 63.0 \\
\midrule

CIFAR
& 0  & 970 & 36 & 16 & 103 & 1716.68 & 70.32 & --  & -- \\
& 5  & 920 & 36 & 16 & 108 & 1724.14 & 72.68 & 0.4 & 3.3 \\
& 10 & 850 & 37 & 17 & 117 & 1729.46 & 76.73 & 0.7 & 9.1 \\
& 20 & 800 & 36 & 17 & 125 & 1748.13 & 79.18 & 1.8 & 12.6 \\

\bottomrule
\end{tabular}%
}
\end{table}

\update{\\\noindent\textbf{Experiments.} We evaluate the overhead of tolerating up to $\dropout\in\{0,5,10,20\}$ unavailable members. Let $\size=2d+1$ denote the number of active members required by the protocol, and let $\sizeprime=\size+\dropout$ denote the provisioned committee size. Because over-provisioning changes the probability that a committee exceeds the corruption threshold, for each value of $\dropout$, we recompute the minimum $\thr$ that preserves the same global committee-safety bound, subject to $\num\sizeprime\le\totalC$. We then independently repeat the same balanced $(\num,\pf)$ parameter search for each $\dropout$, yielding the optimized deployment for that dropout tolerance. The training workload
itself is unchanged, but the resulting choices of $\num$, $\pf$, and $\thr$
change the active committee size and hence the per-machine computation and
communication costs. This captures the overhead of provisioning the system for
the desired dropout tolerance.}

At $\totalC=100\mathrm{K}$, the overhead remains modest for small dropout budgets and increases gradually as more reserve members are provisioned. For $\dropout=5$, the optimized configurations incur $0.4$--$18.0\%$ additional computation and $3.3$--$27.1\%$ additional communication relative to $\dropout=0$; for $\dropout=10$, the corresponding overheads are $0.7$--$34.1\%$ and $9.1$--$33.0\%$. Even at $\dropout=20$, the optimized configurations remain feasible for all workloads, with computation overhead ranging from $1.8\%$ to $74.2\%$ and communication overhead from $12.6\%$ to $63.0\%$. The overhead varies because larger dropout tolerance changes the feasible committee configurations. In particular, the optimized solutions generally use fewer committees and therefore assign more users to each committee, increasing the per-helper training workload, while increases in \(\thr\) and \(\size\) raise committee-dependent protocol costs. These effects are more pronounced for \mnistarch{}, \fmnistarch{}, \svhnarch{}, whereas \cifararch{} is dominated by its large local FC workload, so the parameter changes contribute only a small fraction of its total computation.


\section{Average Pool Implementation}
\label{app:gap}

Global average pooling (GAP) computes, for each batch element $b$ and channel $c$,
\[
y_{b,c} \;=\; \frac{1}{HW} \sum_{h=1}^{H} \sum_{w=1}^{W} x_{b,c,h,w}.
\]

In the secret-sharing setting, each value $x_{b,c,h,w}$ is shared among the parties, and we denote its sharing by $\pss{x_{b,c,h,w}}{d}$. Since GAP is a linear operation, it can be evaluated locally on the shares. Specifically, each party computes
\[
\pss{y_{b,c}}{d}
\;=\;
\frac{1}{HW} \sum_{h=1}^{H} \sum_{w=1}^{W} \pss{x_{b,c,h,w}}{d}.
\]

The scalar $\frac{1}{HW}$ is public and corresponds to a modular inverse in the field $\Fp$, so multiplication by this constant is a local operation and incurs no interaction.

\section{Evaluation (cntd.)}
\label{app:eval}
\subsection{Setup (cntd.)} 
\label{app:rp}
\noindent \textbf{Random Projection.} The inclusion of the high-dimensional $4096$-input layer significantly increases model size; in the malicious setting, this directly impacts both runtime and communication, as the gradient validation step is costly. 
To mitigate this overhead, we apply random projection~\cite{random-projection} to reduce the dimensionality of updates by projecting to $d' = 10{,}000$ prior to aggregation. In particular, we use fast random projections based on Walsh--Hadamard transforms~\cite{random-projection}, which preserve structure while substantially reducing computation and communication costs.

\subsection{Experiments (cntd.)} 

\subsubsection{The \textit{\textsc{Poseidon}} Crossover: When can it win?}
\label{sec:eval:poseidon-crossover}
The tradeoff is rounds vs.\ compute: \textsc{Poseidon} uses 3 WAN rounds but incurs heavy HE compute, while \sysname{} uses $\sim$36--38 rounds with negligible compute. \textsc{Poseidon} can outperform \sysname{} only if its compute cost is offset by reduced round latency, i.e., $\mathrm{RTT} > T_{\text{comp}}^{\textsc{pos}} / (R_{\sysname} - R_{\textsc{pos}})$ (where $T_{\text{comp}}^{\textsc{pos}}$ is \textsc{Poseidon} per-iteration compute, $R_{\sysname}, R_{\textsc{pos}}$ are round counts, and $\mathrm{RTT}$ is network latency). This condition is highly restrictive: even for unrealistically shallow models and very small batch sizes it requires RTTs of $\sim$144\,ms, and it scales linearly with HE compute, becoming infeasible as model size or batch increases. For example, for \cifararch, $\mathrm{RTT}^* \approx 6264\,\mathrm{s}/(38-3) \approx 3$ minutes per round trip, which is physically impossible. Therefore, \textsc{Poseidon} is never competitive in any realistic deployment: it can only appear favorable under contrived settings with extremely small models, tiny batches, and abnormally high latency. In all practical federated learning regimes, its HE compute cost dominates, and \sysname{} strictly outperforms it.

\subsubsection{Parameter Design Space}
\label{app:eval:eq2}

A configuration of \sysname{} is determined by the number of committees
$\num$, packing factor $\pf$, corruption threshold $\thr$, and the resulting
committee size
\[
    \size = 2(\thr+\pf-1)+1.
\]
Given $\totalC$ users and corruption rate $\beta$, we choose the smallest
$\thr$ such that the probability that any committee contains more than
$\thr$ corrupted members is at most $10^{-12}$. Specifically, for every
candidate $(\num,\pf)$, we require
\[
    \num\cdot
    \Pr\!\left[X>\thr\right]
    \le 10^{-12},
    \qquad
    X\sim
    \mathrm{Hypergeom}
    \!\left(
        \totalC,
        \lfloor\beta\totalC\rfloor,
        \size
    \right),
\]
with $\beta=0.01$, together with the feasibility condition
$\num\size\le\totalC$. For malicious security, we additionally require
$\thr<\pf$. Thus, once $(\num,\pf)$ is fixed, $\thr$ and $\size$ are
determined by the security requirement rather than tuned independently.

The two main tuning parameters, $\num$ and $\pf$, affect different parts of
the protocol. Increasing $\num$ reduces the number of users assigned to each
committee,
\[
    U=\left\lceil\frac{\totalC}{\num}\right\rceil,
\]
which generally reduces user-dependent training and validation work per
committee member. At the same time, more committees increase
cross-committee aggregation overhead and, in the malicious setting, the
cost of checking that all committees received consistent model weights.
Moreover, changing $\num$ may change the minimum safe threshold $\thr$ and
therefore the committee size $\size$. Consequently, increasing the number of
committees is beneficial only while the reduction in per-committee workload
outweighs these additional costs.

The packing factor $\pf$ creates a similar tradeoff. A larger $\pf$ places
more values in each packed sharing and can therefore reduce the number of
packed objects that must be processed. For example, an FC layer with input
dimension $a$ uses $\lceil a/\pf\rceil$ packed chunks, while a convolution
with $K$ Winograd tiles uses $\lceil K/\pf\rceil$ tile groups. However,
increasing $\pf$ also increases the sharing degree
$d=\thr+\pf-1$ and hence typically increases the committee size
$\size=2d+1$. Interactive operations such as degree reduction, truncation,
and repacking therefore become more expensive. The resulting tradeoff
depends on the exact architecture and packing layout, so neither a larger
nor a smaller $\pf$ is universally optimal.

We therefore search over candidate $(\num,\pf)$ pairs directly. For each
pair, we first determine the minimum safe $\thr$, discard infeasible
configurations, and then evaluate the resulting per-machine computation and
sent communication. We retain the Pareto-optimal configurations with respect
to these two objectives and report the Pareto-knee configuration as the
selected operating point (see \autoref{tbl:eval:params}).

\begin{table}[H]
  \centering
  \caption{Optimal \sysname{} protocol parameters $(\num,\pf,\thr,\size)$ for each $(\totalC,\text{arch})$ in \autoref{tbl:eval:eq1-compute}.
  \label{app:tab:optimal}
    \emph{SH}~=~semi-honest Pareto-knee config;
    \emph{Mal}~=~malicious Pareto-knee config.}
  \label{tbl:eval:params}
  \resizebox{\columnwidth}{!}{%
  \begin{tabular}{@{}clrrrrrrrr@{}}
    \toprule
    & & \multicolumn{4}{c}{SH optimal} & \multicolumn{4}{c}{Mal optimal} \\
    \cmidrule(lr){3-6}\cmidrule(lr){7-10}
    Arch & $\totalC$ & $\num$ & $\pf$ & $\thr$ & $\size$ & $\num$ & $\pf$ & $\thr$ & $\size$ \\
    \midrule
    \bcwarch & $1\text{K}$    & \update{27}     & 10 &  9 & 37  & $27$    & 10 &  9 & 37  \\
    \bcwarch & $10\text{K}$   & \update{195} & \update{14} & \update{12} & \update{51} & \update{195} & \update{14} & \update{12} & \update{51}  \\
    \bcwarch & $100\text{K}$  & \update{1{,}860} & \update{14} & \update{13} & \update{53}  & \update{1860}               & \update{14} & \update{14} & \update{53} \\
    \bcwarch & $1\text{M}$    & \update{17{,}900} & \update{14} & \update{14} & \update{55}  & \update{10{,}900} & \update{15} & \update{14} & \update{57} \\[2pt]
    \mnistarch & $1\text{K}$  & \update{27} & \update{10} & \update{9} & \update{37} & \update{27} & \update{10} & \update{9} & \update{37}  \\
    \mnistarch & $10\text{K}$ & \update{240} & \update{10} & \update{11} & \update{41} & \update{200} & \update{13} & \update{12} & \update{49} \\
    \mnistarch & $100\text{K}$& \update{2{,}300} & \update{10} & \update{12} & \update{43}  & \update{1{,}860} & \update{14} & \update{13} & \update{53} \\
    \mnistarch & $1\text{M}$  & \update{20{,}000} & \update{12} & \update{13} & \update{49}  & \update{17{,}300} & \update{15} & \update{14} & \update{57} \\[2pt]
    \cifararch & $1\text{K}$  & \update{12} & \update{32} & \update{10} & \update{83} & \update{19} & \update{17} & \update{9} & \update{51}  \\
    \cifararch & $10\text{K}$ & \update{97} & \update{37} & \update{15} & \update{103} & \update{160} & \update{18} & \update{13} & \update{61}  \\
    \cifararch & $100\text{K}$& \update{970} & \update{36} & \update{16} & \update{103} & \update{1{,}300} & \update{23} & \update{14} & \update{73}  \\
    \cifararch & $1\text{M}$  & \update{7{,}850} & \update{43} & \update{18} & \update{121} & \update{8{,}200} & \update{42} & \update{18} & \update{121}  \\
    \bottomrule
  \end{tabular}}
\end{table}

\begin{table*}[]
  \centering
  \caption{Architecture families for the scaling experiment (CIFAR-10, $32{\times}32{\times}3$).
           $\gradDim$ is the total gradient (parameter) count.
           D2 and W3 share the same architecture.
           All depth models share the tail AvgPool$\to$FC$\to$ReLU$\to$FC(${\to}10$);
           channel sizes of the FC layers adapt to the preceding feature-map width.}
  \label{tab:arch-families}
  \small
  \begin{tabular}{@{}llr@{}}
    \toprule
    ID & Architecture (Conv blocks $\to$ head) & $\gradDim$ \\
    \midrule
    D1 & Conv($3{\to}64$)--ReLU -- AvgPool -- FC($64{\to}256$)--ReLU--FC($256{\to}10$)                          &     20{,}672 \\
    D2 & Conv($3{\to}64$)--ReLU -- Conv($64{\to}128$)--ReLU -- AvgPool -- FC($128{\to}256$)--ReLU--FC($256{\to}10$) &   110{,}784 \\
    D3 & $+$Conv($128{\to}256$)--ReLU (before AvgPool)                                                          &    438{,}464 \\
    D4 & $+$Conv($256{\to}512$)--ReLU (before AvgPool)                                                          & 1{,}683{,}648 \\
    \midrule
    W1 & Conv($3{\to}16$)--ReLU--Conv($16{\to}32$)--ReLU--AvgPool--FC($32{\to}64$)--ReLU--FC($64{\to}10$)      &      7{,}728 \\
    W2 & Conv($3{\to}32$)--ReLU--Conv($32{\to}64$)--ReLU--AvgPool--FC($64{\to}128$)--ReLU--FC($128{\to}10$)    &     28{,}768 \\
    W3 & Conv($3{\to}64$)--ReLU--Conv($64{\to}128$)--ReLU--AvgPool--FC($128{\to}256$)--ReLU--FC($256{\to}10$)  &    110{,}784 \\
    W4 & Conv($3{\to}128$)--ReLU--Conv($128{\to}256$)--ReLU--AvgPool--FC($256{\to}512$)--ReLU--FC($512{\to}10$) &   434{,}560 \\
    \bottomrule
  \end{tabular}
\end{table*}

\subsubsection{Accuracy}
\label{app:eval:acc}
We report the test accuracy of \sysname{} across datasets. Learning parameters are selected via grid search with 3-fold cross-validation on plaintext data; in practice, these can be jointly determined by participating parties using secure statistics protocols~\cite{param-choose1,param-choose2}.

\begin{table}[H]
\centering
\setlength{\tabcolsep}{4pt}
\begin{tabular}{l l r r}
\toprule
\textbf{Dataset} & \textbf{Model} & \textbf{Epochs} & \textbf{Test Acc.} \\
\midrule
BCW          & $\bcwarch$    & 50  & 95.6\% \\
Credit       & $\creditarch$ & 150 & 83.0\% \\
MNIST        & $\mnistarch$  & 100 & 92.0\% \\
FashionMNIST & $\fmnistarch$ & 100 & 80.7\% \\
SVHN         & $\svhnarch$   & 300 & 63.3\% \\
CIFAR-10     & $\cifararch$  & 100 & 52.0\% \\
\bottomrule
\end{tabular}
\caption{Test accuracy of \sysname{} architectures, trained in a single-committee simulation ($l{=}4$, $t{=}2$, $11$ parties) on one machine.}
\label{tab:accuracy}
\end{table}

\subsubsection{\update{Comparison with Single Server Secure Aggregation}}
\label{app:eval:secure-aggregation}

We compare against state of the art single server secure aggregation systems, \textsc{Flamingo}~\cite{ma2023flamingo} and \textsc{Willow}~\cite{bell2025willow}. In \textsc{Flamingo}, users locally compute and mask their model updates using individual and pairwise masks, while a small decryptor committee helps the server in recovering the aggregate under user dropout. In \textsc{Willow}, users encrypt their updates and associated key material; the server homomorphically aggregates them, with decryptors enabling recovery of the aggregate and a verifiers checking the aggregation in the actively secure setting. Thus, both systems protect the \emph{aggregation} of locally computed updates, but not the training computation or model weights. In contrast, \sysname{} executes the full forward and backward pass under MPC while keeping model weights secret shared throughout. We therefore ask whether \sysname{} can provide this stronger protection at competitive per machine cost.

\paragraph{Setup.} For a direct per-machine comparison, we compare the cost of one \sysname{} committee member with the corresponding helper-side cost in each baseline. For \textsc{Flamingo}, we report per decryptor cost. For \textsc{Willow}, we consider both per-verifier and per-decryptor costs and report the larger of the two. \sysname{} committee member cost is evaluated on a 12 core commodity machine using a single thread. Baseline helper costs are evaluated on a 128-core machine, with independent per-user work parallelized across the available cores. This gives the baselines realistic server settings and thus provides a conservative comparison. user dropout is 1\% throughout. 

\paragraph{Results.}
\autoref{tab:eval:sa} reports the per-machine computation cost, showing that \sysname{} remains competitive with, and at large user populations substantially outperforms, the secure-aggregation baseline despite performing a significantly stronger task: each committee member in \sysname{} participates in the complete secure forward and backward pass, whereas the baseline helpers support only secure aggregation of updates that users have already computed locally. At $\totalC=100\text{K}$, this additional work is visible for larger workloads, where \textsc{Flamingo} is still faster than \sysname{}. As the user population increases to $N=10\text{M}$, \sysname{} becomes faster than both baselines across all workloads, gaining roughly $920\times$ over \textsc{Willow} even for the largest \cifararch{}. 


This trend follows from where the population-dependent computation arises in the three systems. In \sysname{}, the dominant cost is the semi-honest secure training computation, whose per-committee-member workload is determined primarily by the model architecture and committee parameters, being independent of the total user population $N$. Increasing $N$ affects only the additional malicious-security mechanisms, mainly through the number of committees that must be coordinated and checked. Therefore, only a relatively small portion of the per-machine cost grows with the user population. In contrast, \textsc{Flamingo} requires each decryptor to process user-dependent recovery material, giving its per-decryptor workload an approximately linear dependence on $N$. \textsc{Willow} distributes verification across multiple verifiers, reducing the workload of an individual verifier, but its per-machine verification cost still grows with $N$ under the scalable verifier configuration. 

As a result, the fixed cost of securely evaluating the model dominates \sysname{} at smaller populations, while its relatively weak dependence on $N$ increases as the population grows. This is clear for \cifararch{}: at $N=100\text{K}$, \textsc{Flamingo} benefits from performing only lightweight aggregation support and is faster than a \sysname{} committee member, whereas at $N=10\text{M}$ its growing user-dependent work outweighs this advantage. Thus, \sysname{} achieves favorable per-machine scalability even though it protects much more of the training process than either secure-aggregation baseline.


For communication, we compare the upload of an ordinary user that do not serve as protocol helpers. A \sysname{} user distributes packed shares of its private input across the committee so that both training data and model computation can remain protected inside MPC. 

\begin{table*}[t]
\centering
\caption{
Per machine comparison with single server secure aggregation systems.
\emph{Top:} per round compute for one \sysname{} committee member compared with 
\textsc{Flamingo} and \textsc{Willow}. Parentheses show slowdown relative to
\sysname{}. \emph{Bottom:} per user upload at $\totalC=100\text{K}$.
}
\label{tab:eval:sa}
\setlength{\tabcolsep}{3pt}
\small
\begin{tabular}{@{}l rrr rrr@{}}
\toprule
& \multicolumn{3}{c}{Compute: $\totalC = 100\text{K}$}
& \multicolumn{3}{c}{Compute: $\totalC = 10\text{M}$} \\
\cmidrule(lr){2-4}\cmidrule(lr){5-7}
Workload
& \sysname{} & \textsc{Flamingo} & \textsc{Willow}
& \sysname{} & \textsc{Flamingo} & \textsc{Willow} \\
\midrule

\bcwarch{} ($d\approx1.6\text{K}$)
& $\mathbf{\update{77}\,\text{ms}}$
& $380\,\text{ms}$ {\scriptsize($\update{5}{\times}$)}
& $2.3\,\text{min}$ {\scriptsize($\update{1{,}792}{\times}$)}
& $\mathbf{\update{972}\,\text{ms}}$
& $51.5\,\text{s}$ {\scriptsize($\update{53}{\times}$)}
& $3.9\,\text{hr}$ {\scriptsize($\update{14{,}444}{\times}$)} \\

\mnistarch{} ($d\approx500$)
& $\mathbf{\update{462}\,\text{ms}}$
& $380\,\text{ms}$ {\scriptsize($\update{0.82}{\times}$)}
& $1.1\,\text{min}$ {\scriptsize($\update{142}{\times}$)}
& $\mathbf{\update{2.3}\,\text{s}}$
& $51.5\,\text{s}$ {\scriptsize($\update{22.4}{\times}$)}
& $1.8\,\text{hr}$ {\scriptsize($\update{2{,}817}{\times}$)} \\

\cifararch{} ($d\approx1.07\text{M}$)
& $\mathbf{\update{2.6}\,\text{s}}$
& $380\,\text{ms}$ {\scriptsize($\update{0.15}{\times}$)}
& $2.1\,\text{min}$ {\scriptsize($\update{48.5}{\times}$)}
& $\mathbf{\update{13.7}\,\text{s}}$
& $51.5\,\text{s}$ {\scriptsize($\update{3.8}{\times}$)}
& $3.5\,\text{hr}$ {\scriptsize($\update{920}{\times}$)} \\

\midrule[\heavyrulewidth]
& \multicolumn{6}{c}{Per user upload at $\totalC=100\text{K}$} \\
\cmidrule(lr){2-7}
Workload
& \multicolumn{2}{c}{\sysname{}}
& \multicolumn{2}{c}{\textsc{Flamingo}}
& \multicolumn{2}{c}{\textsc{Willow}} \\
\midrule

\bcwarch{} ($d\approx1.6\text{K}$)
& \multicolumn{2}{c}{$43\,\text{KB}$}
& \multicolumn{2}{c}{$11\,\text{KB}$}
& \multicolumn{2}{c}{$1.25\,\text{MB}$} \\

\mnistarch{} ($d\approx500$)
& \multicolumn{2}{c}{$5.5\,\text{KB}$}
& \multicolumn{2}{c}{$7\,\text{KB}$}
& \multicolumn{2}{c}{$0.6\,\text{MB}$} \\

\cifararch{} ($d\approx1.07\text{M}$)
& \multicolumn{2}{c}{$29\,\text{MB}$}
& \multicolumn{2}{c}{$4.2\,\text{MB}$}
& \multicolumn{2}{c}{$5.1\,\text{MB}$} \\

\bottomrule
\end{tabular}
\end{table*}

\subsubsection{Comparison with Prior Work\cite{liu2024scalable}}
\label{app:eval:hmmpc}

A prior work by Liu et al.~\cite{liu2024scalable} presents an honest-majority, semi-honest MPC system (HMMPC) for private neural network inference using standard (unpacked) Shamir secret sharing. We benchmark it at $n{=}63$ parties ($t{=}31$) for FC architectures and $n{=}131$ for CNN architectures to match our committee sizes. 

\autoref{tbl:eval:hmmpc} shows the comparison. \sysname{} uses \PSS{} with packing factor $\pf{=}20$ (FC arches) and $\pf{=}50$ (CNN arches), processing $\pf$ samples per MPC operation. For FC architectures, the communication advantage comes from the packing factor itself ($\pf{=}20$), yielding $5.8{\times}$--$12.3{\times}$ bandwidth reduction. For CNN architectures, the advantage is amplified by the large spatial feature maps: \sysname{} needs only one degree-reduction per packed output regardless of spatial dimension, while HMMPC incurs $O(n)$ communication per spatial activation. This yields $3{,}619{\times}$ and $29{,}258{\times}$ bandwidth reduction for \mnistarch{} and \svhnarch{} respectively. Runtime improvements mirror the communication gains: $445{\times}$--$619{\times}$ for FC models and ${\sim}10^4$--$10^5{\times}$ for CNN models.

\input{tables/hmmpc}

\begin{table}[t]
  \centering
  \caption{\textsc{HMMPC} vs.\ \sysname{} forward-pass cost.
    \textsc{HMMPC} numbers are per-party; \sysname{} numbers are per-machine.
    \sysname{} uses $\pf{=}20$ for FC architectures and $\pf{=}50$ for CNN architectures.}
  \label{tbl:eval:hmmpc}
  \setlength{\tabcolsep}{3pt}
  \resizebox{\columnwidth}{!}{%
  \begin{tabular}{@{}l rr rr rr@{}}
    \toprule
    & \multicolumn{2}{c}{\textsc{HMMPC} ($n{=}63$ or $131$)}
    & \multicolumn{2}{c}{\sysname{} ($\size{=}63$ or $131$)}
    & \multicolumn{2}{c}{Reduction} \\
    \cmidrule(lr){2-3}\cmidrule(lr){4-5}\cmidrule(lr){6-7}
    Architecture
      & Time & Comm
      & Time & Comm
      & Time$\downarrow$ & Comm$\downarrow$ \\
    \midrule
    \bcwarch   ($n{=}63$)
      & $208\,\mathrm{ms}$ & $1.82\,\mathrm{MB}$
      & $0.34\,\mathrm{ms}$ & $0.32\,\mathrm{MB}$
      & $619{\times}$ & $5.8{\times}$ \\
    \creditarch ($n{=}63$)
      & $196\,\mathrm{ms}$ & $3.38\,\mathrm{MB}$
      & $0.44\,\mathrm{ms}$ & $0.42\,\mathrm{MB}$
      & $445{\times}$ & $8.1{\times}$ \\
    \cifararch  ($n{=}63$)
      & $1.61\,\mathrm{s}$ & $19.3\,\mathrm{MB}$
      & $0.74\,\mathrm{s}$ & $1.57\,\mathrm{MB}$
      & $2.2{\times}$ & $12.3{\times}$ \\
    \mnistarch  ($n{=}131$)
      & $58.5\,\mathrm{s}$ & $4{,}474\,\mathrm{MB}$
      & $1.84\,\mathrm{ms}$ & $1.24\,\mathrm{MB}$
      & $3.2{\times}10^4{\times}$ & $3{,}619{\times}$ \\
    \fmnistarch ($n{=}131$)
      & $58.5\,\mathrm{s}$ & $4{,}474\,\mathrm{MB}$
      & $1.84\,\mathrm{ms}$ & $1.24\,\mathrm{MB}$
      & $3.2{\times}10^4{\times}$ & $3{,}619{\times}$ \\
    \svhnarch   ($n{=}131$)
      & $670\,\mathrm{s}$ & $44{,}504\,\mathrm{MB}$
      & $2.27\,\mathrm{ms}$ & $1.52\,\mathrm{MB}$
      & $3.0{\times}10^5{\times}$ & $2.9{\times}10^4{\times}$ \\
    \bottomrule
  \end{tabular}}
\end{table}

\subsubsection{WAN End-to-End Latency}
\label{app:eval:wan}

\autoref{tab:wan-latency} reports total per-iteration latency under three representative WAN round-trip times: 10\,ms, 15\,ms, and 100\,ms. Total time is:
\[
  T_{\mathrm{WAN}} = T_{\mathrm{compute}} + R \times \mathrm{RTT},
\]
where $T_{\mathrm{compute}}$ is the per-machine compute time at $N=100\mathrm{K}$ and $R$ is the number of MPC communication rounds. Even at $\mathrm{RTT}=100$ ms, \sysname{} completes a training iteration in about 1.4 s (SH) / 3.4 s (Mal) for FC architectures (BCW-MLP, Credit-MLP), under 4 s (SH) / 7.4 s (Mal) for CNN architectures, and 5.5 s (SH) / 9.3 s (Mal) for CIFAR-MLP-4K.
\begin{table}[H]
  \centering
  \caption{WAN end-to-end per-iteration latency for all \sysname{} architectures.
    $T_{\mathrm{WAN}} = T_{\mathrm{compute}} + R \times \mathrm{RTT}$; compute at $\totalC{=}100\text{K}$ (see \autoref{tbl:eval:eq1-compute} for full $\totalC$ scaling).
    Rounds $R$ are per-architecture constants, independent of $\totalC$.}
  \label{tab:wan-latency}
  \resizebox{\columnwidth}{!}{%
  \begin{tabular}{@{}l r rr rr rr@{}}
    \toprule
    Architecture & Rounds & \multicolumn{2}{c}{RTT\,=\,10\,ms} & \multicolumn{2}{c}{RTT\,=\,15\,ms} & \multicolumn{2}{c}{RTT\,=\,100\,ms} \\
    \cmidrule(lr){3-4}\cmidrule(lr){5-6}\cmidrule(lr){7-8}
    & (SH\,/\,Mal) & SH & Mal & SH & Mal & SH & Mal \\
    \midrule
    \bcwarch   & 14 / 33 & 143 ms & 407 ms & 213 ms & 572 ms & 1.40 s & 3.38 s \\
    \creditarch & 14 / 33 & 143 ms & 437 ms & 213 ms & 602 ms & 1.40 s & 3.41 s \\
    \mnistarch & 36 / 53 & 643 ms & 992 ms & 823 ms & 1.26 s & 3.88 s & 5.76 s \\
    \fmnistarch & 36 / 53 & 582 ms & 903 ms & 762 ms & 1.17 s & 3.82 s & 5.67 s \\
    \svhnarch  & 36 / 69 & 661 ms & 1.16 s & 841 ms & 1.50 s & 3.90 s & 7.37 s \\
    \cifararch$^{\dagger}$ & 38 / 67 & 2.10 s & 3.31 s & 2.29 s & 3.65 s & 5.52 s & 9.34 s \\
    \bottomrule
    \multicolumn{8}{@{}l}{\footnotesize $^{\dagger}$Malicious \cifararch\ uses random projection ($d'{=}10\text{K}$); see App.~\ref{app:rp}.}
  \end{tabular}}
\end{table}
\input{tables/eq1}

%% file: tables/hmmpc.tex

%% file: tables/eq1.tex
\begin{table*}[t]
  \centering
  \caption{Per-machine iteration. All times in ms. \textbf{Gain} is the fastest-baseline compute divided by \sysname{} SH compute; the winning baseline is identified. Bold indicates the fastest system.}
  \label{tbl:eval:eq1-compute}
   \resizebox{\textwidth}{!}{%
  \begin{tabular}{@{}clrrrrrl@{}}
    \toprule
    Architecture & $n$ & \sysname\ SH (ms) & \sysname\ Mal (ms)
      & POSEIDON (ms) & Falcon SH (ms) & Falcon Mal (ms) & Gain (\sysname\ SH) \\
    \midrule
    \multicolumn{8}{@{}l}{\bcwarch, $\approx$1.6K params;\quad (14\,SH / 33\,Mal MPC rounds)} \\[2pt]
    \bcwarch & $1\text{K}$ & $\mathbf{2.20{\times}10^{0}}$ & $2.30{\times}10^{1}$ & $7.68{\times}10^{5}$ & $2.77{\times}10^{5}$ & $4.52{\times}10^{5}$ & $1.26{\times}10^{5}{{\times}}$ vs.\ Falcon \\
    \bcwarch & $10\text{K}$ & $\mathbf{2.83{\times}10^{0}}$ & $4.30{\times}10^{1}$ & $7.68{\times}10^{5}$ & $2.77{\times}10^{6}$ & $4.52{\times}10^{6}$ & $2.72{\times}10^{5}{{\times}}$ vs.\ POSEIDON \\
    \bcwarch & $100\text{K}$ & $\mathbf{3.28{\times}10^{0}}$ & $7.71{\times}10^{1}$ & $7.68{\times}10^{5}$ & $2.77{\times}10^{7}$ & $4.52{\times}10^{7}$ & $2.34{\times}10^{5}{{\times}}$ vs.\ POSEIDON \\
    \bcwarch & $1\text{M}$ & $\mathbf{7.15{\times}10^{0}}$ & $2.89{\times}10^{2}$ & $7.68{\times}10^{5}$ & $2.77{\times}10^{8}$ & $4.52{\times}10^{8}$ & $1.07{\times}10^{5}{{\times}}$ vs.\ POSEIDON \\

    \midrule
    \multicolumn{8}{@{}l}{\mnistarch, $\approx$500 params;\quad (36\,SH / 53\,Mal MPC rounds)} \\[2pt]
    \mnistarch & $1\text{K}$ & $\mathbf{2.10{\times}10^{2}}$ & $2.39{\times}10^{2}$ & $5.16{\times}10^{5}$ & $8.82{\times}10^{4}$ & $1.44{\times}10^{5}$ & $4.20{\times}10^{2}{{\times}}$ vs.\ Falcon \\
    \mnistarch & $10\text{K}$ & $\mathbf{2.55{\times}10^{2}}$ & $3.87{\times}10^{2}$ & $5.16{\times}10^{5}$ & $8.82{\times}10^{5}$ & $1.44{\times}10^{6}$ & $2.02{\times}10^{3}{{\times}}$ vs.\ POSEIDON \\
    \mnistarch & $100\text{K}$ & $\mathbf{2.83{\times}10^{2}}$ & $4.62{\times}10^{2}$ & $5.16{\times}10^{5}$ & $8.82{\times}10^{6}$ & $1.44{\times}10^{7}$ & $1.82{\times}10^{3}{{\times}}$ vs.\ POSEIDON \\
    \mnistarch & $1\text{M}$ & $\mathbf{3.41{\times}10^{2}}$ & $6.93{\times}10^{2}$ & $5.16{\times}10^{5}$ & $8.82{\times}10^{7}$ & $1.44{\times}10^{8}$ & $1.51{\times}10^{3}{{\times}}$ vs.\ POSEIDON \\

    \midrule
    \multicolumn{8}{@{}l}{\cifararch, $\approx$1.07M params;\quad (38\,SH / 67\,Mal MPC rounds)} \\[2pt]
    \cifararch & $1\text{K}$ & $\mathbf{1.58{\times}10^{3}}$ & $1.80{\times}10^{3}$ & $3.11{\times}10^{8}$ & $1.89{\times}10^{8}$ & $3.08{\times}10^{8}$ & $1.20{\times}10^{5}{{\times}}$ vs.\ Falcon \\
    \cifararch & $10\text{K}$ & $\mathbf{1.63{\times}10^{3}}$ & $2.06{\times}10^{3}$ & $3.11{\times}10^{8}$ & $1.89{\times}10^{9}$ & $3.08{\times}10^{9}$ & $1.91{\times}10^{5}{{\times}}$ vs.\ POSEIDON \\
    \cifararch & $100\text{K}$ & $\mathbf{1.72{\times}10^{3}}$ & $2.64{\times}10^{3}$ & $3.11{\times}10^{8}$ & $1.89{\times}10^{10}$ & $3.08{\times}10^{10}$ & $1.81{\times}10^{5}{{\times}}$ vs.\ POSEIDON \\
    \cifararch & $1\text{M}$ & $\mathbf{2.11{\times}10^{3}}$ & $5.18{\times}10^{3}$ & $3.11{\times}10^{8}$ & $1.89{\times}10^{11}$ & $3.08{\times}10^{11}$ & $1.48{\times}10^{5}{{\times}}$ vs.\ POSEIDON \\

    \midrule
    \multicolumn{8}{@{}l}{\creditarch, $\approx$2.3K params;\quad (14\,SH / 33\,Mal MPC rounds)} \\[2pt]
    \creditarch & $1\text{K}$ & $\mathbf{1.99{\times}10^{0}}$ & $3.49{\times}10^{1}$ & $7.68{\times}10^{5}$ & $4.07{\times}10^{5}$ & $6.64{\times}10^{5}$ & $2.04{\times}10^{5}{{\times}}$ vs.\ Falcon \\
    \creditarch & $10\text{K}$ & $\mathbf{2.65{\times}10^{0}}$ & $5.95{\times}10^{1}$ & $7.68{\times}10^{5}$ & $4.07{\times}10^{6}$ & $6.64{\times}10^{6}$ & $2.90{\times}10^{5}{{\times}}$ vs.\ POSEIDON \\
    \creditarch & $100\text{K}$ & $\mathbf{3.19{\times}10^{0}}$ & $1.07{\times}10^{2}$ & $7.68{\times}10^{5}$ & $4.07{\times}10^{7}$ & $6.64{\times}10^{7}$ & $2.41{\times}10^{5}{{\times}}$ vs.\ POSEIDON \\
    \creditarch & $1\text{M}$ & $\mathbf{7.56{\times}10^{0}}$ & $4.01{\times}10^{2}$ & $7.68{\times}10^{5}$ & $4.07{\times}10^{8}$ & $6.64{\times}10^{8}$ & $1.02{\times}10^{5}{{\times}}$ vs.\ POSEIDON \\

    \midrule
    \multicolumn{8}{@{}l}{\fmnistarch, $\approx$400 params;\quad (36\,SH / 53\,Mal MPC rounds)} \\[2pt]
    \fmnistarch & $1\text{K}$ & $\mathbf{1.64{\times}10^{2}}$ & $1.88{\times}10^{2}$ & $4.18{\times}10^{5}$ & $7.14{\times}10^{4}$ & $1.16{\times}10^{5}$ & $4.34{\times}10^{2}{{\times}}$ vs.\ Falcon \\
    \fmnistarch & $10\text{K}$ & $\mathbf{2.00{\times}10^{2}}$ & $3.08{\times}10^{2}$ & $4.18{\times}10^{5}$ & $7.14{\times}10^{5}$ & $1.16{\times}10^{6}$ & $2.09{\times}10^{3}{{\times}}$ vs.\ POSEIDON \\
    \fmnistarch & $100\text{K}$ & $\mathbf{2.22{\times}10^{2}}$ & $3.73{\times}10^{2}$ & $4.18{\times}10^{5}$ & $7.14{\times}10^{6}$ & $1.16{\times}10^{7}$ & $1.88{\times}10^{3}{{\times}}$ vs.\ POSEIDON \\
    \fmnistarch & $1\text{M}$ & $\mathbf{2.71{\times}10^{2}}$ & $5.85{\times}10^{2}$ & $4.18{\times}10^{5}$ & $7.14{\times}10^{7}$ & $1.16{\times}10^{8}$ & $1.54{\times}10^{3}{{\times}}$ vs.\ POSEIDON \\

    \midrule
    \multicolumn{8}{@{}l}{\svhnarch, $\approx$500 params;\quad (36\,SH / 69\,Mal MPC rounds)} \\[2pt]
    \svhnarch & $1\text{K}$ & $\mathbf{2.23{\times}10^{2}}$ & $2.52{\times}10^{2}$ & $8.00{\times}10^{6}$ & $8.42{\times}10^{4}$ & $1.37{\times}10^{5}$ & $3.78{\times}10^{2}{{\times}}$ vs.\ Falcon \\
    \svhnarch & $10\text{K}$ & $\mathbf{2.71{\times}10^{2}}$ & $4.10{\times}10^{2}$ & $8.00{\times}10^{6}$ & $8.42{\times}10^{5}$ & $1.37{\times}10^{6}$ & $3.11{\times}10^{3}{{\times}}$ vs.\ Falcon \\
    \svhnarch & $100\text{K}$ & $\mathbf{3.01{\times}10^{2}}$ & $4.69{\times}10^{2}$ & $8.00{\times}10^{6}$ & $8.42{\times}10^{6}$ & $1.37{\times}10^{7}$ & $2.66{\times}10^{4}{{\times}}$ vs.\ POSEIDON \\
    \svhnarch & $1\text{M}$ & $\mathbf{3.50{\times}10^{2}}$ & $7.23{\times}10^{2}$ & $8.00{\times}10^{6}$ & $8.42{\times}10^{7}$ & $1.37{\times}10^{8}$ & $2.29{\times}10^{4}{{\times}}$ vs.\ POSEIDON \\
    \bottomrule
  \end{tabular}}
\end{table*}

\begin{table*}[t]
  \centering
  \caption{Per-machine bandwidth per iteration (MB).}
  \label{tbl:eval:eq1-comm}

  \tiny
  \setlength{\tabcolsep}{2.3pt}
  \renewcommand{\arraystretch}{0.72}

  \resizebox{0.78\textwidth}{!}{%
  \begin{tabular}{@{}clrrrrr@{}}
    \toprule
    Arch & $n$ & \sysname\ SH & \sysname\ Mal & POSEIDON & Falcon SH & Falcon Mal \\
    \midrule

    \bcwarch
      & $1$K   & $1.13{\times}10^{0}$ & $1.16{\times}10^{1}$ & $1.88{\times}10^{3}$ & $1.23{\times}10^{3}$ & $6.79{\times}10^{3}$ \\
      & $10$K  & $1.26{\times}10^{0}$ & $1.96{\times}10^{1}$ & $1.88{\times}10^{3}$ & $1.23{\times}10^{4}$ & $6.79{\times}10^{4}$ \\
      & $100$K & $1.31{\times}10^{0}$ & $2.17{\times}10^{1}$ & $1.88{\times}10^{3}$ & $1.23{\times}10^{5}$ & $6.79{\times}10^{5}$ \\
      & $1$M   & $1.36{\times}10^{0}$ & $4.03{\times}10^{1}$ & $1.88{\times}10^{3}$ & $1.23{\times}10^{6}$ & $6.79{\times}10^{6}$ \\

    \midrule
    \mnistarch
      & $1$K   & $2.15{\times}10^{1}$ & $2.88{\times}10^{1}$ & $4.49{\times}10^{3}$ & $3.93{\times}10^{2}$ & $2.16{\times}10^{3}$ \\
      & $10$K  & $2.63{\times}10^{1}$ & $4.23{\times}10^{1}$ & $4.49{\times}10^{3}$ & $3.93{\times}10^{3}$ & $2.16{\times}10^{4}$ \\
      & $100$K & $2.88{\times}10^{1}$ & $4.70{\times}10^{1}$ & $4.49{\times}10^{3}$ & $3.93{\times}10^{4}$ & $2.16{\times}10^{5}$ \\
      & $1$M   & $3.02{\times}10^{1}$ & $5.13{\times}10^{1}$ & $4.49{\times}10^{3}$ & $3.93{\times}10^{5}$ & $2.16{\times}10^{6}$ \\

    \midrule
    \cifararch
      & $1$K   & $5.27{\times}10^{1}$ & $1.32{\times}10^{2}$ & $1.66{\times}10^{3}$ & $8.42{\times}10^{5}$ & $4.63{\times}10^{6}$ \\
      & $10$K  & $6.90{\times}10^{1}$ & $1.70{\times}10^{2}$ & $1.66{\times}10^{3}$ & $8.42{\times}10^{6}$ & $4.63{\times}10^{7}$ \\
      & $100$K & $7.03{\times}10^{1}$ & $2.53{\times}10^{2}$ & $1.66{\times}10^{3}$ & $8.42{\times}10^{7}$ & $4.63{\times}10^{8}$ \\
      & $1$M   & $8.37{\times}10^{1}$ & $4.72{\times}10^{2}$ & $1.66{\times}10^{3}$ & $8.42{\times}10^{8}$ & $4.63{\times}10^{9}$ \\

    \midrule
    \creditarch
      & $1$K   & $4.06{\times}10^{-1}$ & $1.65{\times}10^{1}$ & $1.88{\times}10^{3}$ & $1.81{\times}10^{3}$ & $9.97{\times}10^{3}$ \\
      & $10$K  & $5.35{\times}10^{-1}$ & $2.72{\times}10^{1}$ & $1.88{\times}10^{3}$ & $1.81{\times}10^{4}$ & $9.97{\times}10^{4}$ \\
      & $100$K & $5.85{\times}10^{-1}$ & $3.14{\times}10^{1}$ & $1.88{\times}10^{3}$ & $1.81{\times}10^{5}$ & $9.97{\times}10^{5}$ \\
      & $1$M   & $6.05{\times}10^{-1}$ & $4.80{\times}10^{1}$ & $1.88{\times}10^{3}$ & $1.81{\times}10^{6}$ & $9.97{\times}10^{6}$ \\

    \midrule
    \fmnistarch
      & $1$K   & $1.99{\times}10^{1}$ & $2.59{\times}10^{1}$ & $3.63{\times}10^{3}$ & $3.18{\times}10^{2}$ & $1.75{\times}10^{3}$ \\
      & $10$K  & $2.43{\times}10^{1}$ & $3.79{\times}10^{1}$ & $3.63{\times}10^{3}$ & $3.18{\times}10^{3}$ & $1.75{\times}10^{4}$ \\
      & $100$K & $2.67{\times}10^{1}$ & $4.20{\times}10^{1}$ & $3.63{\times}10^{3}$ & $3.18{\times}10^{4}$ & $1.75{\times}10^{5}$ \\
      & $1$M   & $2.80{\times}10^{1}$ & $4.57{\times}10^{1}$ & $3.63{\times}10^{3}$ & $3.18{\times}10^{5}$ & $1.75{\times}10^{6}$ \\

    \midrule
    \svhnarch
      & $1$K   & $2.66{\times}10^{1}$ & $3.38{\times}10^{1}$ & $6.96{\times}10^{4}$ & $3.75{\times}10^{2}$ & $2.06{\times}10^{3}$ \\
      & $10$K  & $3.25{\times}10^{1}$ & $4.93{\times}10^{1}$ & $6.96{\times}10^{4}$ & $3.75{\times}10^{3}$ & $2.06{\times}10^{4}$ \\
      & $100$K & $3.56{\times}10^{1}$ & $5.27{\times}10^{1}$ & $6.96{\times}10^{4}$ & $3.75{\times}10^{4}$ & $2.06{\times}10^{5}$ \\
      & $1$M   & $3.82{\times}10^{1}$ & $6.02{\times}10^{1}$ & $6.96{\times}10^{4}$ & $3.75{\times}10^{5}$ & $2.06{\times}10^{6}$ \\
    \bottomrule
  \end{tabular}%
  }
\end{table*}